\documentclass[12pt]{article}

\usepackage{verbatim,color,amssymb,bbm}
\usepackage{amsmath}					
\usepackage{amsthm}					
\usepackage{natbib}
\usepackage{multirow}
\usepackage{breqn}
\usepackage{setspace}
\usepackage[mathscr]{euscript}
\usepackage{fancyhdr}
\usepackage{enumitem}
\usepackage{graphicx,grffile}
\usepackage{geometry}
\usepackage{xcolor}
\usepackage{physics,xfrac}
\usepackage[colorlinks,allcolors=blue]{hyperref}
\usepackage{caption,subcaption}
\usepackage{algorithm,algpseudocode}

\graphicspath{{Figures/}}

\usepackage{multibib}
\newcites{latex}{References}

\usepackage{tabularx, booktabs}
\newcolumntype{Y}{>{\centering\arraybackslash}X}
\usepackage{lscape}

\usepackage{float}
\usepackage{lineno}

\newtheorem{remark}{Remark}
\newtheorem{lemma}{Lemma}
\newtheorem{theorem}{Theorem}
\newtheorem{corollary}{Corollary}
\newtheorem*{theorem*}{Theorem}
\newtheorem{proposition}{Proposition}

\def\pP{\mathbb{P}}  
\def\nN{\mathbb{N}}
\def\rR{\mathbb{R}}
\def\eE{\mathbb{E}}

\def\D{{\cal D}}
\def\E{{\cal E}}
\def\F{{\cal F}}

\def\M{{\cal M}}

\def\diag{\mathrm{diag}}

\def\wt{\widetilde}

\def\FDR{\hbox{FDR}}

\def\var{\hbox{var}}
\def\cov{\hbox{cov}}

\def\trace{\mathrm{trace}}

\def\vect{\mathrm{vec}}

\def\Beta{\hbox{Beta}}

\def\Dir{\hbox{Dir}}

\def\Exp{\hbox{Exp}}

\def\Ga{\hbox{Ga}}

\def\Normal{\hbox{N}}
\newcommand{\mn}{\mathrm{N}}

\def\P_25_ICML{{\it Proceedings of the 25th international conference on Machine learning}}

\def\bse{\begin{eqnarray*}}
\def\ese{\end{eqnarray*}}
\def\be{\begin{eqnarray}}
\def\ee{\end{eqnarray}}
\def\bq{\begin{equation}}
\def\eq{\end{equation}}

\def\trans{^{\rm T}}

\def\th{^{th}}

\def\bA{{\mathbf A}}

\def\bB{{\mathbf B}}
\def\bc{{\mathbf c}}

\def\bD{{\mathbf D}}
\def\b1e{{\mathbf e}}
\def\b1f{{\mathbf f}}

\def\bg{{\mathbf g}}

\def\bH{{\mathbf H}}
\def\bI{{\mathbf I}}
\def\bL{{\mathbf L}}

\def\bP{{\mathbf P}}
\def\bQ{{\mathbf Q}}

\def\bR{{\mathbf R}}

\def\bu{{\mathbf u}}

\def\bv{{\mathbf v}}
\def\bV{{\mathbf V}}
\def\bw{{\mathbf w}}

\def\bx{{\mathbf x}}
\def\bX{{\mathbf X}}
\def\by{{\mathbf y}}

\def\bz{{\mathbf z}}

\def\bzero{{\mathbf 0}}

\def\simind{\stackrel{\mbox{\scriptsize{ind}}}{\sim}}
\def\simiid{\stackrel{\mbox{\scriptsize{iid}}}{\sim}}

\newcommand{\bmu}{\mbox{\boldmath $\mu$}}
\newcommand{\bvarepsilon}{\mbox{\boldmath $\varepsilon$}}
\newcommand{\bdelta}{\mbox{\boldmath $\delta$}}
\newcommand{\bDelta}{\mbox{\boldmath $\Delta$}}
\newcommand{\bphi}{\mbox{\boldmath $\phi$}}

\newcommand{\bxi}{\mbox{\boldmath $\xi$}}
\newcommand{\bepsilon}{\mbox{\boldmath $\epsilon$}}
\newcommand{\btheta}{\mbox{\boldmath $\theta$}}

\newcommand{\bbeta}{\mbox{\boldmath $\beta$}}

\newcommand{\bzeta}{\mbox{\boldmath $\zeta$}}

\newcommand{\bSigma}{\mbox{\boldmath $\Sigma$}}
\newcommand{\balpha}{\mbox{\boldmath $\alpha$}}

\newcommand{\blambda}{\mbox{\boldmath $\lambda$}}
\newcommand{\bnablan}{\mbox{\boldmath $\nabla$}_{n}}
\newcommand{\bLambda}{\mbox{\boldmath $\Lambda$}}
\newcommand{\bLambdan}{\mbox{\boldmath $\Lambda$}_{n}}
\newcommand{\bLambdann}{\mbox{\boldmath $\Lambda$}_{0n}}
\newcommand{\tLambdann}{\widetilde {\mbox{\boldmath $\Lambda$}}_{0n}}
\newcommand{\bDeltann}{\mbox{\boldmath $\Delta$}_{0n}}
\newcommand{\deltann}{\delta_{0n}}
\newcommand{\deltan}{\delta_{n}}
\newcommand{\Pnn}{\mathcal{P}_{0n}}
\newcommand{\Pnon}{\mathcal{P}_{1,n}}
\newcommand{\Pntw}{\mathcal{P}_{2,n}}
\newcommand{\Pn}{\mathcal{P}_{n}}
\newcommand{\Snt}{S_{0}}
\newcommand{\bDeltan}{\mbox{\boldmath $\Delta$}_{n}}
\newcommand{\bGammann}{\mbox{\boldmath $\Gamma$}_{0n}}

\newcommand{\bXinn}{\mbox{\boldmath $\Xi$}_{0n}}

\newcommand{\bOmega}{\mbox{\boldmath $\Omega$}}

\newcommand{\bpsi}{\mbox{\boldmath $\psi$}}

\newcommand{\half}{\sfrac{1}{2}}

\newcommand{\de}{\mathrm{d}}

\newcommand{\ssparse}{\wt{s}}
\newcommand{\sn}{s_{n}}
\newcommand{\qn}{q_{n}}
\newcommand{\qnn}{q_{0n}}
\newcommand{\dn}{d_{n}}
\newcommand{\en}{e_{n}}
\newcommand{\an}{a_{n}}
\newcommand{\bnn}{b_{n}}
\newcommand{\Hn}{H_{n}}
\newcommand{\tn}{t_{n}}
\newcommand{\pin}{\Pi_{n}}

\newcommand{\kl}[2] {\mathbb{KL} \left(#1\parallel #2\right)  }
\newcommand{\lnt}[2] {\ell_{0} \left[#1, #2\right]  }
\newcommand{\klv}[2] {\mathbb{V} \left(#1 \parallel #2\right)  }
\newcommand{\bOmegan}{\mbox{\boldmath $\Omega$}_{n}}
\newcommand{\bOmegann}{\mbox{\boldmath $\Omega$}_{0n}}

\newcommand{\bLamtil}{\wt{\bLambda}}
\newcommand{\bDeltil}{\wt{\bDelta}}

\newcommand{\delmin}{{\delta}_{\min}}

\newcommand*{\Scale}[2][4]{\scalebox{#1}{$#2$}}%

\newcommand{\fnorm}[1] {  \norm{#1}_F}

\newcommand{\specnorm}[1] {  \norm{#1}_2}
\newcommand{\hamm}[2]{ \left\langle {#1}, {#2} \right\rangle_{H} }
\newcommand{\smin}[1] { s_{\min}\left(#1\right) }

\renewcommand\footnoterule{\kern-3pt \hrule \textwidth 2in \kern 2.6pt}

\usepackage{cancel}
\usepackage{soul}
\usepackage[normalem] {ulem}

\def\colred#1{\textcolor{red}{ #1}}

\def\boxit#1{\vbox{\hrule\hbox{\vrule\kern6pt \vbox{\kern6pt \textcolor{blue}{#1}\kern6pt}\kern6pt\vrule}\hrule}}

\def\authorfootnote#1{{\let\thefootnote\relax\footnotetext{#1}}}

\allowdisplaybreaks

\begin{document}
\thispagestyle{empty}
\baselineskip=28pt

\begin{center}
{\LARGE{\bf Bayesian Scalable Precision Factor\\
		\vskip -9pt
		Analysis 
		for Massive Sparse \\ 
		Gaussian Graphical Models
}}
\end{center}
\baselineskip=12pt
\vskip 15pt 

\newcommand{\authors}{\begin{center}
		Noirrit Kiran Chandra$^{a}$ (noirrit.chandra@utdallas.edu)\\
		Peter M\"uller$^{b,c}$ (pmueller@math.utexas.edu)\\
		Abhra Sarkar$^{b}$ (abhra.sarkar@utexas.edu)\\
		\vskip 10pt
		$^{a}$Department of Mathematical Sciences, \\
		The University of Texas at Dallas,\\
		800 W. Campbell Rd, 
		Richardson, TX 75080-3021, USA\\ 
		\vskip 5pt
		$^{b}$Department of Statistics and Data Sciences, \\
		The University of Texas at Austin,\\ 2317 Speedway D9800, Austin, TX 78712-1823, USA\\
		\vskip 5pt 
		$^{c}$Department of Mathematics, \\
		The University of Texas at Austin,\\ 2515 Speedway, PMA 8.100, Austin, TX 78712-1823, USA\\
\end{center}}

\authors


\vskip 15pt 
\begin{center}
{\Large{\bf Abstract}} 
\end{center}
\baselineskip=12pt

	We propose a novel approach to estimating the precision matrix of multivariate Gaussian data
	that relies on decomposing them 
	into a low-rank and a diagonal component. 
	Such decompositions are very popular for modeling large covariance matrices 
	as they admit a latent factor based representation that allows easy
	inference.
	The same is however not true for precision matrices due to the lack of
	computationally convenient representations 
	which restricts inference to low-to-moderate dimensional problems.
	We address this remarkable gap in the literature by building on a 
	latent variable representation for such decomposition for precision matrices. 
	The construction leads to an efficient 
	Gibbs sampler that scales very well to high-dimensional problems 
	far beyond the limits of the current state-of-the-art. 
	The ability to efficiently explore the full posterior space also allows 
	the model uncertainty to be easily assessed. 
	The decomposition crucially additionally 
	allows us to adapt sparsity inducing priors to  
	shrink the insignificant entries of the precision matrix toward zero, 
	making the approach adaptable to 
	high-dimensional small-sample-size sparse settings. 
	Exact zeros in the matrix encoding the underlying conditional independence graph 
	are then determined 
	via a novel posterior false discovery rate control procedure. 
	A near minimax optimal posterior concentration rate for estimating precision matrices is attained 
	by our method under mild regularity assumptions. 
	We evaluate the method's empirical performance through synthetic experiments 
	and illustrate its practical utility in data sets from two different application domains.

\vskip 20pt 
\baselineskip=12pt
\noindent\underline{\bf Key Words}: 
Factor models, 
False discovery rate control, 
Gaussian graphical models, 
Markov chain Monte Carlo, 
Posterior concentration, 
Precision matrix estimation, 
Scalable computation,
Shrinkage priors

\par\medskip\noindent
\underline{\bf Short/Running Title}: Precision Factor Analysis

\par\medskip\noindent
\underline{\bf Corresponding Author}: Abhra Sarkar (abhra.sarkar@utexas.edu)

\pagenumbering{arabic}
\setcounter{page}{0}
\newlength{\gnat}
\setlength{\gnat}{26pt}
\baselineskip=\gnat


\newpage
\section{Introduction}
\label{sec:intro}
For multivariate Gaussian distributed data $\by=(y_{1},\dots,y_{d})\trans \sim \mn_{d}(\bzero, \bSigma)$, 
all conditional dependence information is contained in the inverse covariance matrix $\bSigma^{-1} = \bOmega = ((\omega_{j,j'})) $, or the precision. 
Two `nodes' $y_{j}$ and $y_{j'}$ are conditionally independent given the rest if and only if $\omega_{j,j'}=0$. 
The underlying conditional (in)dependence graph is then obtained by connecting the pairs of nodes $\{(j,j'): \omega_{j,j'} \neq 0\}$ by undirected `edges'.
Estimating the precision matrix, including especially its sparsity patterns, 
for such data is therefore an important statistical problem 
\citep{lauritzen1996graphical,koller2009probabilistic}.

{In this article, we model $\bOmega$ via a \textit{low-rank and diagonal} (LRD) decomposition. 
	\citet{bhattacharya2016fast} introduced a representation for efficient sampling from Gaussian distributions with known LRD structured precision matrices.
	In this article, we {adapt the} representation 
	in a different way to obtain a novel {factor analytic framework} for unknown precision matrices modeled via LRD decompositions with applications to graphical models.
	While such decomposition is already a widely used standard tool for modeling high-dimensional $\bSigma$, 
	to our knowledge, 
	the construction has not been utilized before for modeling $\bOmega$. 
	Our research addresses this remarkable gap in the literature.} 
Although any positive definite matrix can always be factorized this way, 
the main challenge is to introduce such constructions for $\bOmega$ in a way that allows efficient and scalable posterior inference. 
Going significant steps further, we also adapt this approach to sparse high-dimensional settings. 
Noting that any sparse $\bOmega$ can always be represented with a sparse factorization, we impose sparsity in $\bOmega$ by using sparsity-inducing priors for the factorization.

\noindent \textbf{Existing Methods for Covariance and Precision Matrix Estimation:} 
The existing literature on sparse covariance and precision matrix estimation is vast. 
In sparse covariance matrix estimation problems, the entire matrix $\bSigma$ is often directly penalized \citep{levina2008sparsecov,bien2011covest}. 
In an alternative approach, $\bSigma$ is assumed to admit a LRD structure $\bSigma=\bLamtil \bLamtil\trans+\bDeltil$ 
where $\bLamtil$ is a $d\times q$ order matrix 
and $\bDeltil$ is a diagonal matrix with all positive entries. 
In theory, all positive definite matrices admit such a representation for some $0 \leq q \leq d$,
and in practice, $q\ll d$ often suffice to produce good approximations. 
Also importantly, this model admits the latent variable representation $\by= \bLamtil \wt{\bu}+\wt{\bv}$ 
where $\wt{\bu} \sim \mn_{q}(\bzero, \bI_{q})$ and $\wt{\bv} \sim \mn_{d}(\bzero, \bDeltil)$.
The formulation allows massive scalability in computation, 
making LRD based methods popular in the high-dimensional covariance matrix estimation literature \citep{fan2011penalizingfactor,fan2018penalizingfactor,daniele2019penalizingfactor}, 
especially in Bayesian settings \citep[and others]{bhattacharya2011sparse,pati2014, zhu2014,KASTNER2019}.
A sparse $\bSigma$, however, does not usually produce a sparse $\bOmega$ 
and the strategy of inverting the estimated $\bSigma$ to obtain an estimate of $\bOmega$ 
tends to exhibit poor empirical performance \citep{pourahmadi2013high}.

As in covariance matrix estimation problems,
penalized likelihood based methods that directly penalize the number and/or absolute values of non-zero entries in $\bOmega$ 
have also been developed in the frequentist setting \citep{yuan2007model, banerjee2008, rothman2008sparse, d2008first, friedman2008sparse, witten2011new, mazumder2012graphical,zhang2014}.
Alternatively, \citet{meinshausen2006high,peng2009partial} developed
neighborhood selection methods that learn the {edges} by regressing each variable on the rest with penalties on large regression coefficients. 

When the main focus is on estimating the underlying dependence graph, 
Bayesian approaches instead rely on defining 
a hierarchical prior on $\bOmega$ preceded by a prior on the graph.
Choices for the latter include uniform priors over graph sizes \citep{armstrong2009bayesian},
priors centered around some informed location \citep{mitra2013bayesian}, 
priors with edge inclusions following a binomial distribution \citep{dobra2004sparse,carvalho2009objective},  
a hyper-Markov distribution on decomposable graphs \citep{dawid1993hyper},
its generalizations to non-decomposable settings \citep{roverato2002hyper,khare2018bayesian}, etc. 
However, such hierarchical construction with a separate model layer for the underlying graph structure makes posterior exploration quite challenging. 
Markov chain Monte Carlo (MCMC) algorithms have been designed specifically for such models \citep[and others]{dellaportas2003bayesian,atay2005monte,carvalho2007simulation,dobra2011bayesian,green2013sampling,lenkoski2013direct,mohammadi2015bayesian} 
but these strategies still rely on expensive local exploration moves, 
often involving trans-dimensional proposals in the graph space and/or approximations of intractable normalizing constants, 
hence remaining computationally infeasible beyond only a few tens of dimensions \citep{jones2005experiments}.

Bayesian methods that directly penalize $\bOmega$, thereby avoiding to have to specify a separate prior for the underlying graph, 
have started to get some attention but the literature remains sparse. 
{\citet{yoshida_west} proposed a factor model with complex constraints 
	enforcing identical sparsity patterns in the covariance and precision matrices which may be restrictive in practice while also having scalability issues.} 
\cite{banerjee2015bayesian} studied spike and slab type priors \citep{Ishwaran2005spikeslab} to shrink irrelevant off-diagonals to zero.
For point mass mixture priors, MCMC based model space exploration can generally be daunting and may lead to slow mixing and convergence 
even in simple mean and linear regression problems.
Continuous shrinkage priors \citep{polson2010shrink} that allow fast and efficient posterior exploration have also been adapted for precision matrices. 
\cite{wang2012bayesian,khondker2013bayesian} and \cite{li2019graphical} designed block Gibbs samplers 
that allow updating entire columns of $\bOmega$ at once.
{\citet{mohammadi2021jasa} proposed an approximated sampler that}
can scale up to a few hundreds but the problem
remains infeasible for modern applications with many thousands of nodes. 
More recent developments along these lines \citep{gan2019bayesian,li2019bayesianB,deshpande2019simultaneous} 
have focused on 
fast deterministic Expectation-Maximization (EM) algorithms instead,  
which scale well to problems with a few hundred dimensional nodes but only estimate the posterior mode (MAP) and not the full posterior.
\cite{KsheeraSagar2021precision} studied both MCMC and MAP estimation for element-wise horseshoe like priors on the precision matrix 
and studied their convergence properties. 
The approach, however, suffers from similar scalability issues.
Also, in many of these approaches \citep{friedman2008sparse, peng2009partial}, 
the estimated $\bOmega$ is not guaranteed to be positive definite, 
requiring post-hoc analysis to fix the estimate.

Owing to the lack of computationally tractable hierarchical structures, 
in high-dimensions, 
posterior explorations in existing Bayesian precision matrix and graph estimation methods thus still remain prohibitively expensive if not entirely impossible. 
With the few exceptions, 
existing Bayesian approaches also do not come with rigorous theoretical guarantees.

\noindent {\bf Our Proposed Precision Factor Model:} 
In contrast to covariance matrices, 
there are currently no flexible LRD decomposition methods that 
admit easily tractable latent variable representations for precision matrices, 
posing significant methodological and computational challenges, especially for MCMC based Bayesian inference. 

In this article, 
we derive a factor analytic framework 
{building on} a 
representation for Gaussian distributions with an LRD decomposed precision matrix
$\bOmega=\bLambda\bLambda\trans +\bDelta$ for some $d\times q$ order matrix $\bLambda$ and some diagonal matrix $\bDelta$ \citep{bhattacharya2016fast}\label{page:bhattacharya}.  
The representation is immediately useful 
in facilitating efficient posterior simulation in Bayesian inference of $\bOmega$ 
and easily scales to problems with dimensions $d$ far beyond the limits of the current state-of-the-art.

Since all positive definite matrices admit LRD representations ($q=0$ and $q=d$ being the two extremes), 
the proposed approach {imposes no} restrictive assumptions on the
precision matrix or the underlying graph.
{Additionally, we discuss a constructive approach to find an LRD representation of arbitrary sparse $\bOmega$.}
In simulation examples, we show
that $q \ll d$ often suffices 
to produce good approximations of $\bOmega$ even when they do not exactly admit an LRD for such $q$. 
%
Conversely, since the representation always produces a positive
definite matrix, unlike many existing procedures such as
\cite{meinshausen2006high,friedman2008sparse,gan2019bayesian}, we
obtain positive definite estimates simply by design.

As the off-diagonals elements of $\bOmega$ are contributed entirely by $\bLambda\bLambda\trans$,
a sparse $\bLambda$ {is expected} to produce a sparse $\bOmega$
{(see Figure \ref{fig: sparsity patterns} and Section \ref{sec: sm sparsity})}.
A suitably chosen penalty on $\bLambda$ therefore {allows to flexibly adapt its non-zero elements such that} 
{insignificant off-diagonals of $\bOmega$ are shrunk towards zero}.
The strategy has been successfully employed in high-dimensional sparse covariance matrix estimation literature in both Bayesian \citep{pati2014} and frequentist paradigms \citep{daniele2019penalizingfactor}.
In this article, we adapt the Dirichlet-Laplace shrinkage priors \citep{bhattacharya2015dirichlet} on $\bLambda$
for their theoretical and computational tractability.
As an artifact of Bayesian methods with continuous shrinkage priors, 
exact zeroes do not appear in the posterior samples.
Thus, we address the problem of non-zero off-diagonal/edge selection from the posterior MCMC samples
through a novel multiple hypothesis testing {approach}
that allows the posterior false edge discovery rate (FDR) \citep{chandra2019non} to be controlled at any desired level. 
This is another salient feature of our proposed method that properly accommodates posterior uncertainty in reporting a point estimate for the graph. 

With our model and prior specifications, we establish near minimax-optimal contraction rates of the posterior around the true $\bOmega$ under mild assumptions when the number of nodes increases {exponentially} with the sample size. 
We evaluate the proposed method's finite sample efficacy through simulation experiments 
where it either outperforms or is competitive with previously existing methods in moderately high-dimensional problems 
while also scaling to dimensions far beyond the reach of many of those methods. 
We illustrate our method's practical utility in real data sets from two different application domains.

\noindent{\bf Our Key Contributions:} 
Overall, our main contributions to the literature include 
(a) proposing a likelihood based approach built on a low-rank and diagonal decomposition of the precision matrix,
(b) build on a latent factor representation for such decomposition that provides new interpretations for such models, 
(c) designing a Gibbs sampling algorithm that exploits this latent factor representation and scales very well to high-dimensional problems, 
(d) adapting shrinkage priors for the low-rank component that further makes the method applicable to high-dimensional small-sample-size sparse settings, 
(e) developing a novel FDR control procedure for graph selection from the posterior samples of the precision matrix, and 
(f) establishing rigorous asymptotic properties of the posterior of the proposed approach.

\noindent{\bf Outline of the Article:} The rest of this article is organized as follows. 
Section \ref{sec: models} details our matrix decomposition based model. 
Section \ref{sec: precision factor analysis} discusses our novel factor analytic representation; 
{Section \ref{subsec:lowrank_model} discusses how a low-rank approach can be used to learn arbitrary sparse precision matrices;}
Section \ref{sec: priors} discusses the priors; 
Section \ref{sec: post inference main} describes the posterior sampling algorithm; 
Section \ref{sec: graph} presents our graph selection procedure via FDR control;
Section \ref{sec: asymptotics} discusses the posterior's asymptotic properties. 
Section \ref{sec: sim studies} summarizes the results of simulation experiments.
Section \ref{sec: applications} presents the results for two real data sets. 
Section \ref{sec: discussion} contains concluding remarks.


\section{Gaussian Precision Factor Models} \label{sec: models}
We consider a random sample of $n$ observations assumed to be independently and identically distributed (iid) $d$-dimensional random vectors following
a multivariate Normal distribution with mean zero and precision matrix $\bOmega$ as
\vspace{-8.5ex}\\
\be
\by_{i} \simiid \mn_{d}(\bzero,\bOmega^{-1}), ~~~ i=1,\dots,n, \label{eq: data}
\ee
\vspace{-8.5ex}\\
where 
$\by_{i} = (y_{i,1},\dots,y_{i,d})\trans$. 
The primary goal is to estimate $\bOmega$, especially identifying its sparsity pattern that characterizes conditional independence relationships between different components of $\by$.

To this end, we consider an LRD decomposition of $\bOmega$ as 
\vspace{-8ex}\\
\be
\bOmega=\bLambda \bLambda\trans +\bDelta, \label{eq: precision decomposition}
\ee
\vspace{-8ex}\\
where 
$\bLambda^{d \times q} = ((\lambda_{r,c}))$ 
and $\bDelta = \diag(\delta_{1}^{2},\dots,\delta_{d}^{2})$. 
Factorization \eqref{eq: precision decomposition} is completely flexible in the sense that 
a matrix is positive definite if and only if it admits such a representation {for some $0\leq q\leq d$}. 
Any precision matrix $\bOmega$ can therefore be written as \eqref{eq: precision decomposition} for a sufficiently large value of $q$. 
For most practical cases, 
values of $q \ll d$ suffice to approximate $\bOmega$ well, resulting in a huge reduction in dimensions and allowing massive  scalability in computation.
In model (\ref{eq: precision decomposition}), $\bLambda$ is not strictly identifiable 
since, for example, $\bLambda\bLambda\trans = \bLambda\bQ\bQ\trans\bLambda\trans$ for any orthogonal matrix $\bQ$. 
%
Inference on $\bOmega$ being the primary interest, individual identifiability or interpretation of the model parameters in \eqref{eq: precision decomposition} is, however, not required.

\subsection{Factor Analytic Representation} \label{sec: precision factor analysis}
{One main advantage of modeling covariance matrices via LRD decompositions is 
	the existence of a latent factor  
	representation 
	that greatly facilitates computation. 
		Tractable latent factor representations are, however, not known to exist for precision matrices, presenting major barriers against efficient inference. 
		Recently \citet{bhattacharya2016fast} introduced a representation that greatly facilitates 
		fast sampling from multivariate Gaussian distributions with known precision matrices with an LRD structure. 
		We show that, 
		by exploiting this representation in a different way,  
		a latent factor model can in fact be obtained for LRD decomposed precision matrices as well 
		that allows 
		efficient posterior inference of its unknown components 
		in much the same way as classical factor models facilitate computation for unknown LRD decomposed covariance matrices. 
		The representation, formalized in Proposition \ref{prop: main construction} below and 
		referred to in this article as the `precision factor model', 
		also provides novel insights into the very construction of Gaussian graphical models.} 
	

	\begin{proposition}   \label{prop: main construction}
		The model $\by_{i} \simiid \mn_{d}(\bzero, \bOmega^{-1}), i=1,\dots,n$, 
		with $\bOmega = \bLambda\bLambda\trans + \bDelta$, 
		where $\bLambda$ is ${d \times q}$ with $q \leq d$ 
		and $\bDelta = \diag(\delta_{1}^{2},\dots,\delta_{d}^{2})$, 
		admits the equivalent representation 
				\vskip-9ex
		\be
		& \by_{i} = -\bDelta^{-1}\bLambda (\bI_{q}+\bLambda\trans\bDelta^{-1}\bLambda)^{-1}\bu_{i} + \bv_{i}, \label{eq: latent factor model}\\
		& \text{where}~		
		\begin{bmatrix}
			\bu_{i}\\
			\bv_{i}
		\end{bmatrix} \simiid 
		\mn_{q+d}\left(\bzero, \begin{bmatrix}
			\bI_{q}+\bLambda\trans\bDelta^{-1}\bLambda & \bLambda\trans\bDelta^{-1}\\
			\bDelta^{-1} \bLambda & \bDelta^{-1}
		\end{bmatrix}\right).    \label{eq: joint_dist}
		\ee
	\end{proposition}
	The proposition follows straightforwardly using Sherman-Woodbury identity for the covariance matrix $\bSigma = \bOmega^{-1}$ given by 
		\vskip-9ex
	\bse
	\bOmega^{-1}=(\bLambda\bLambda\trans+\bDelta)^{-1}=\bDelta^{-1}- \bDelta^{-1} \bLambda(\bI_{q}+\bLambda\trans\bDelta^{-1}\bLambda)	^{-1}\bLambda\trans\bDelta^{-1}.
	\ese
		\vskip-14pt	
	Note that $\cov(\by_{i},\bu_{i})=\bzero$ in (\ref{eq: latent factor model}).
	A reverse-engineered construction that is particularly useful for posterior simulation in Bayesian settings 
	utilizes this fact and 
	generates 
	$\bu_{i} \simiid \mn_{q}(\bzero,\bP)$ with $\bP = (\bI_{q}+\bLambda\trans\bDelta^{-1}\bLambda)$ independent of $\by_{1:n}$ first  
	and then sets 
	$\bv_{i} = \bDelta^{-1}\bLambda \bP^{-1}\bu_{i} + \by_{i}$.
	Clearly then, $\bv_{i}\simiid \mn_{d}(\bzero,\bDelta^{-1})$.  
	Also, \\
		\vspace*{-6ex}
	\bse
	\begin{bmatrix}
		\bu_{i}\\
		\by_{i}
	\end{bmatrix} \simiid \mn_{q+d}\left(\bzero, \begin{bmatrix}
		\bP& \bzero\\
		\bzero & \bOmega^{-1}
	\end{bmatrix}\right)
	\text{ then implies }
	\begin{bmatrix}
		\bu_{i}\\
		\bv_{i}
	\end{bmatrix} \simiid \mn_{q+d}\left(\bzero, \begin{bmatrix}
		\bP& \bLambda\trans\bDelta^{-1}\\
		\bDelta^{-1} \bLambda & \bDelta^{-1}
	\end{bmatrix}\right),
	\ese
	as in (\ref{eq: joint_dist}) in Proposition \ref{prop: main construction}. 
	Importantly, we can also write 
		\vskip-9ex
	\be
	\bu_{i} = \bLambda\trans \bv_{i} +\bvarepsilon_{i}, \text{ where } \bvarepsilon_{i} \simiid \mn_{q}(\bzero, \bI_{q}).  \label{eq:factormodel}
	\ee
		\vspace{-9ex}\\
	{Although mathematically simple, Proposition \ref{prop: main construction} has far-reaching implications.}
	An efficient and highly scalable Gibbs sampler 
	follows immediately from the construction 
	by first generating the latent vectors $\bu_{i},\bv_{i}$ given $\bLambda,\bDelta$ and $\by_{i}$ as described above, 
	and then updating the rows of $\bLambda$ given $\bu_{1:n}$ and $\bv_{1:n}$ using (\ref{eq:factormodel}). 

	
	\begin{figure}[!ht]
		\centering
		\vskip -10pt
		\includegraphics[trim={1cm 1cm 1cm 0cm},clip, width=0.485\linewidth]{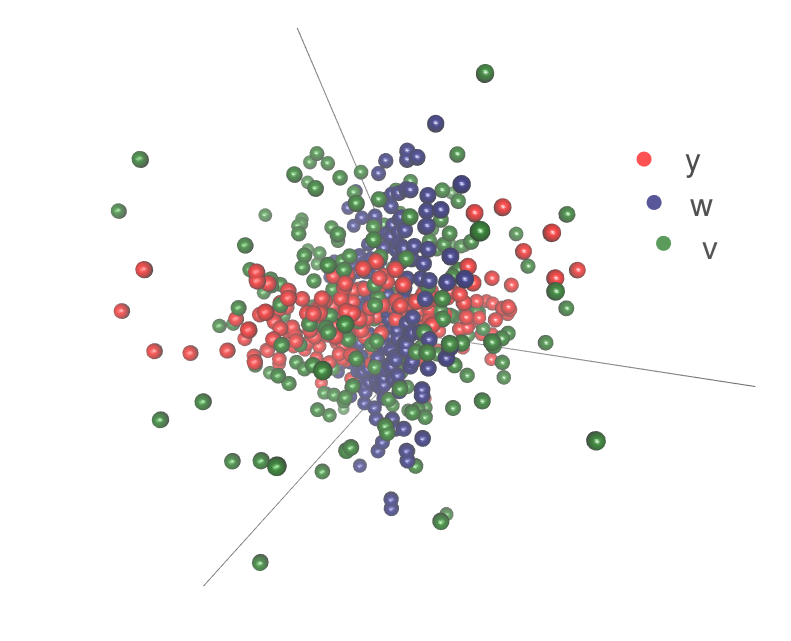}\quad
		\includegraphics[trim={1cm 1cm 1cm 0cm},clip, width=0.485\linewidth]{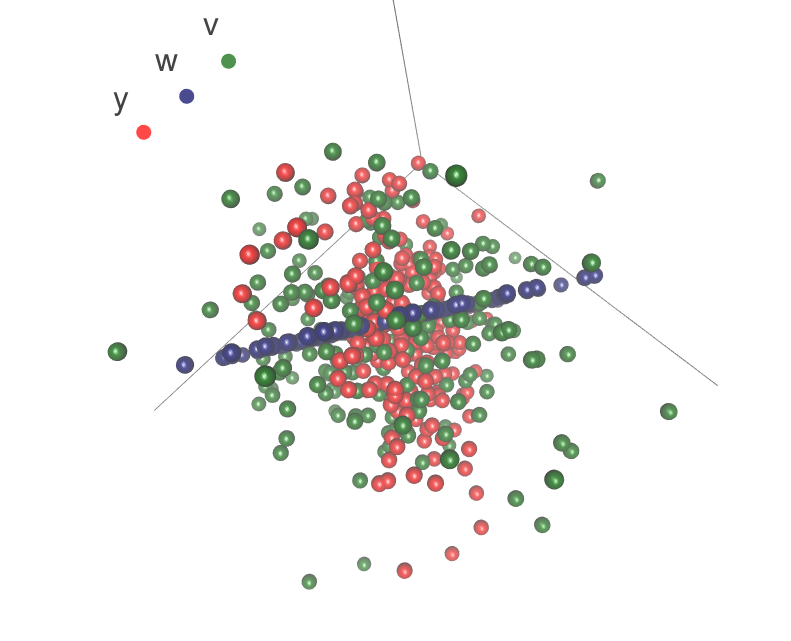}
		\caption{Graphical view of the precision factor model: 
			Plots of observed $\by_{i}$, latent $\bw_{i}=\bDelta^{-1}\bLambda \bP^{-1}\bu_{i}$ and latent $\bv_{i} = \bw_{i} + \by_{i}$  
			for a model with dimension $d=3$ and rank $q=2$ from different viewing angles.
			The $d$-dimensional vectors $\bw$ are supported on a lower $q$-dimensional plane. 
			Also, $\bu$ and $\by$ are independent of each other. 
			Since $\bw$ are linear transformations of $\bu$, 
			they are also independent of $\by$. 
			In the plots, $\bw$ and the average values of $\by$ can be seen to be living on planes that are orthogonal to each other.
			The vectors $\bv$ are more scattered than $\by$ and, 
			in obtaining the variance-covariance of $\by=\bv-\bw$,
			the larger variances of $\bv$ 
			are perfectly compensated by the negative covariances between $\bv$ and $\bw$.
		}
		\label{fig:latents}	
	\end{figure}

	Continuing the thread at the beginning of this subsection, 
	parallels can be drawn with latent factor models for covariance matrices 
	where a similar decomposition $\bSigma=\bLamtil \bLamtil\trans + \bDeltil$ arises from the model 
	$\by_{i} = \bLamtil \wt\bu_{i} + \wt\bv_{i}$ 
	with independent latent components $\wt\bu_{i}$ and $\wt\bv_{i}$,
	where the latent factors $\wt\bu_{i} \simiid\mn_{q}(\bzero,\bI_{q})$, with $q$ typically $\ll d$, and the errors $\wt\bv_{i} \simiid \mn_{d}(\bzero,\wt\bDelta)$. 
	The covariance between the components of $\by_{i}$ and a part of the variance of $\by_{i}$ are thus explained by 
	(a) the variance-covariance of the latent factors $\wt\bu_{i}$ 
	(b) while the remaining unexplained variance is attributed to the errors $\wt\bv_{i}$. 
	
	In contrast, for the precision factor model depicted in Figure \ref{fig:latents}, 
	we have from model \eqref{eq: latent factor model} that 
	$\by_{i}= -\bDelta^{-1}\bLambda \bP^{-1}\bu_{i} + \bv_{i}$ 
	with dependent latent components $\bu_{i}$ and $\bv_{i}$,
	where the latent factors $\bu_{i} \simiid\mn_{q}(\bzero,\bP)$, with $q$ expected again to be $\ll d$, and the `errors' $\bv_{i} \simiid \mn_{d}(\bzero,\bDelta^{-1})$. 
	The variance-covariance of $\by_{i}$ is thus explained by (a) the variance-covariance of the latent factors $\bu_{i}$, 
	(b) the variance of the `errors' $\bv_{i}$, and (c) the covariance between $\bu_{i}$ and $\bv_{i}$.

	If independence between the variance contributing components is desired, 
	an alternative view $\bv_{i} = \bDelta^{-1}\bLambda \bP^{-1}\bu_{i} + \by_{i}$ of model \eqref{eq: latent factor model} is to see 
	the latent $\bv_{i}$'s be composed of independent components $\bu_{i}$ and $\by_{i}$, 
	where the latent factors $\bu_{i} \simiid \mn_{q}(\bzero,\bP)$, with $q \ll d$ as before, and the `errors' $\by_{i} \simiid \mn_{d}(\bzero,\bSigma)$. 
	The $\bv_{i}$'s can thus be represented in an orthogonal decomposition with components $\bw_{i}=\bDelta^{-1}\bLambda \bP^{-1}\bu_{i}$ and $\by_{i}$.
	The larger variances $\bDelta^{-1}$ of $\bv_{i}$ are now being explained by (a) the variance-covariance of the latent factors $\bu_{i}$ and 
	(b) the variance-covariance of the `errors' $\by_{i}$, the off-diagonal covariance terms of these components perfectly cancelling each other to produce the diagonal matrix $\bDelta^{-1}$.

	In yet another view based on equation \eqref{eq:factormodel}, 
	the vectors $\bv_{i}$ can instead be interpreted as independent heterogeneous latent factors and $\bu_{i}$ the associated response vectors subject to white noises $\bvarepsilon_{i}$. 
	Contrary to the classical factor model, in this view, we have the factor dimension $d \gg$ the response dimension $q$. 
	Clearly, $\bOmega^{-1}$ is now the conditional variance-covariance of $\bv_{i}$ given $\bu_{i}$, that is, given the response variables $\bu_{i}$, 
	we can study the behavior of the latent factors $\bv_{i}$ by examining $\bOmega$.

	Although the precision factor model described here has immediate practical implications for Bayesian inference of Gaussian precision matrices considered in this article, 
	the representation itself is not specific to the chosen inferential paradigm 
	but is a mathematical identity 
	that may be of broad general interest in understanding the basic construction of such models.

	\subsection{Low-rank Modeling of Sparse Precision Matrices}
	\label{subsec:lowrank_model}
	In this section we discuss how a sparse precision matrix can indeed admit an LRD decomposition. 
	Then we discuss our strategy of inducing sparsity in high-dimensional precision matrix estimation problems.

	\subsubsection{LRD Decomposition of Sparse Precision Matrices}
	\label{subsec:lrd_sparse}
	We first discuss how a sparse $\bOmega$ can admit an LRD decomposition.
	To get some insights, consider the stylized example from Figure \ref{fig:toy_example} where $d=5$ and $\bOmega$ has $\ssparse=3$ non-zero off-diagonals.
	We recall the factor analytic representation of our model from Section \ref{sec: precision factor analysis} where we write $\bu=\bLambda\trans\bv +\bvarepsilon$ in equation \eqref{eq:factormodel} 
	and interpret $\bOmega=((\omega_{j,h}))$ as the precision matrix of $\bv$ conditionally on $\bu$.
	Under this representation, $\omega_{j,h}\neq 0$ if and only if  $\blambda_{j}\trans \blambda_{h}\neq 0$, that is, if $v_{j}$ and $v_{h}$ are connected to the same $u_{\ell}$, $1 \leq \ell \leq q$.
	
	\begin{figure}[!ht]
		\centering
		\includegraphics[width=.6\linewidth,trim={.5cm .5cm .6cm .45cm},clip]{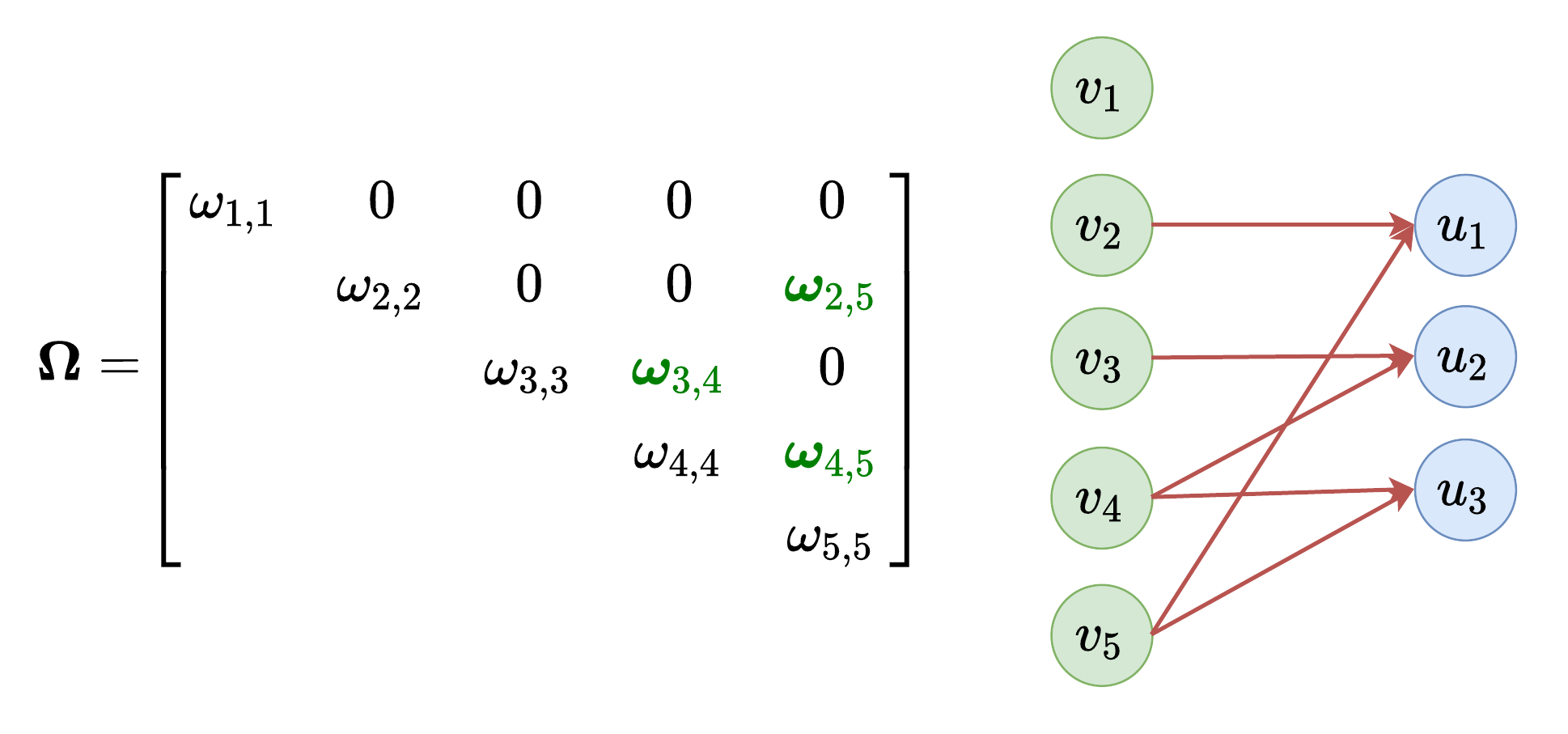}
		\vskip-1ex
		\caption{Constructing an LRD decomposition of a sparse $\bOmega^{5\times 5}$: 
			For each non-zero $\omega_{j,h}$, we connect $v_{j}$ and $v_{h}$ to $u_{\ell}$.
			Notably, $y_{1}$ is marginally independent to every other variable and therefore we do not connect $v_{1}$ to any $u_{\ell}$.}
		\label{fig:toy_example}
	\end{figure}	
	
	Likewise, we can construct a sparse $\bLambda$ as follows
	\vskip-3ex
	\begin{equation*}
		\underbrace{\begin{bmatrix}
				u_{1}\\
				u_{2}\\
				u_{3}
		\end{bmatrix}}_{\bu} = 
		\underbrace{\begin{bmatrix}
				0 & \lambda_{2,1} & 0 & 0 &\lambda_{5,1} \\
				0 & 0 & \lambda_{3,2} & \lambda_{4,2} & 0 \\
				0 & 0 & 0 & \lambda_{4,3} & \lambda_{5,3} 
		\end{bmatrix} }_{\bLambda\trans}
		\underbrace{\begin{bmatrix}
				v_{1}\\
				v_{2}\\
				v_{3}\\
				v_{4}\\
				v_{5}\\
		\end{bmatrix}}_{\bv} +
		\underbrace{\begin{bmatrix}
				\varepsilon_{1}\\
				\varepsilon_{2}\\
				\varepsilon_{3}
		\end{bmatrix}}_{\bvarepsilon},
	\end{equation*}
	\vskip-2ex
	\noindent
	subject to $\lambda_{2,1} \lambda_{5,1} = \omega_{2,5}$,  $\lambda_{3,2} \lambda_{4,2} = \omega_{3,4}$, $\lambda_{4,3} \lambda_{5,3} = \omega_{4,5}$ and $\delta_{j}^{2} + \blambda_{j}\trans \blambda_{j} = \omega_{j,j}$ for all $j=1,\dots,d$.
	Although this construction is not unique, we see that a sparse LRD decomposition indeed exists for our concerned $\bOmega$.
	
	We can generalize this argument for arbitrary sparse $\bOmega$ with $\ssparse$ number of non-zero off-diagonals.
	We start with a $d\times q$ order $\bLambda$ with all entries equal to zero.
	Let $\E$ be the set of edges in the conditional dependence graph. 
	Note that each edge corresponds to a partial correlation between $Y_{j}$ and $Y_{h}$, $1\leq j,h \leq d$.
	For an edge  $\omega_{j,h} \in \E$, we set $\lambda_{j,\ell_{1}}$ and $\lambda_{h,\ell_{1}}$ to be non-zero for some $\ell_{1}$ so that $\blambda_{j}\trans \blambda_{h}\neq 0$.
	For the next edge $\omega_{j',h'} \in \E$, suppose $\omega_{j,j'}=\omega_{j,h'}=\omega_{h,h'}=0$.
	Then we set $\lambda_{j,\ell_{2}}$ and $\lambda_{h,\ell_{2}}$ to be non-zero for some $\ell_{2} \neq \ell_{1}$ so that $\blambda_{j'}\trans \blambda_{h'}\neq 0$ but $\blambda_{j}\trans \blambda_{j'}= \blambda_{j}\trans \blambda_{h'}= \blambda_{h}\trans \blambda_{h'}=0$, and so on.
	Note that in this construction (which need not be unique), 
	for each edge $\omega_{j,j'}\neq 0$, we need to connect $v_{j}$ and $v_{j'}$ to a $u_{\ell}$.	
	\uline{Therefore, we do not need $q$ to be larger than $\ssparse$}.
	The existence of such a $\bLambda$ and $\bDelta$ is subject to the solution of the following system of quadratic equations on $\{\vect(\bLambda),\bdelta_{1:d}\}$ with $dq+d$ unknowns and $d+\ssparse$ equations
		\vskip-5ex
	\begin{equation*}
		\blambda_{j}\trans \blambda_{h} =\omega_{j,h} \text{ for all } \omega_{j,h}\in \E, \quad
		\delta_{j}^{2}  + \blambda_{j}\trans \blambda_{j} = \omega_{j,j} \text{ for all } j=1,\dots,d.
		\label{eq:quad_eqns}
	\end{equation*}
		\vskip-2ex
	In practical high-dimensional applications, even though $d$ is large, the conditional dependence graph is often sparse, that is, $\ssparse \ll d$. 
	\citet{BARVINOK2022quad_system} noted that when the number of unknown variables are considerably larger than the number of equations, a (not necessarily unique) solution to quadratic system of equations exists under very generic conditions.	

\subsubsection{Inducing Sparsity in $\bOmega$ via Sparsity in $\bLambda$} \label{sec: sm sparsity}

\begin{figure}[!ht]	
	\centering
	\begin{equation*}
		\includegraphics[scale=.36, trim={0cm 2cm 0cm 0cm}]{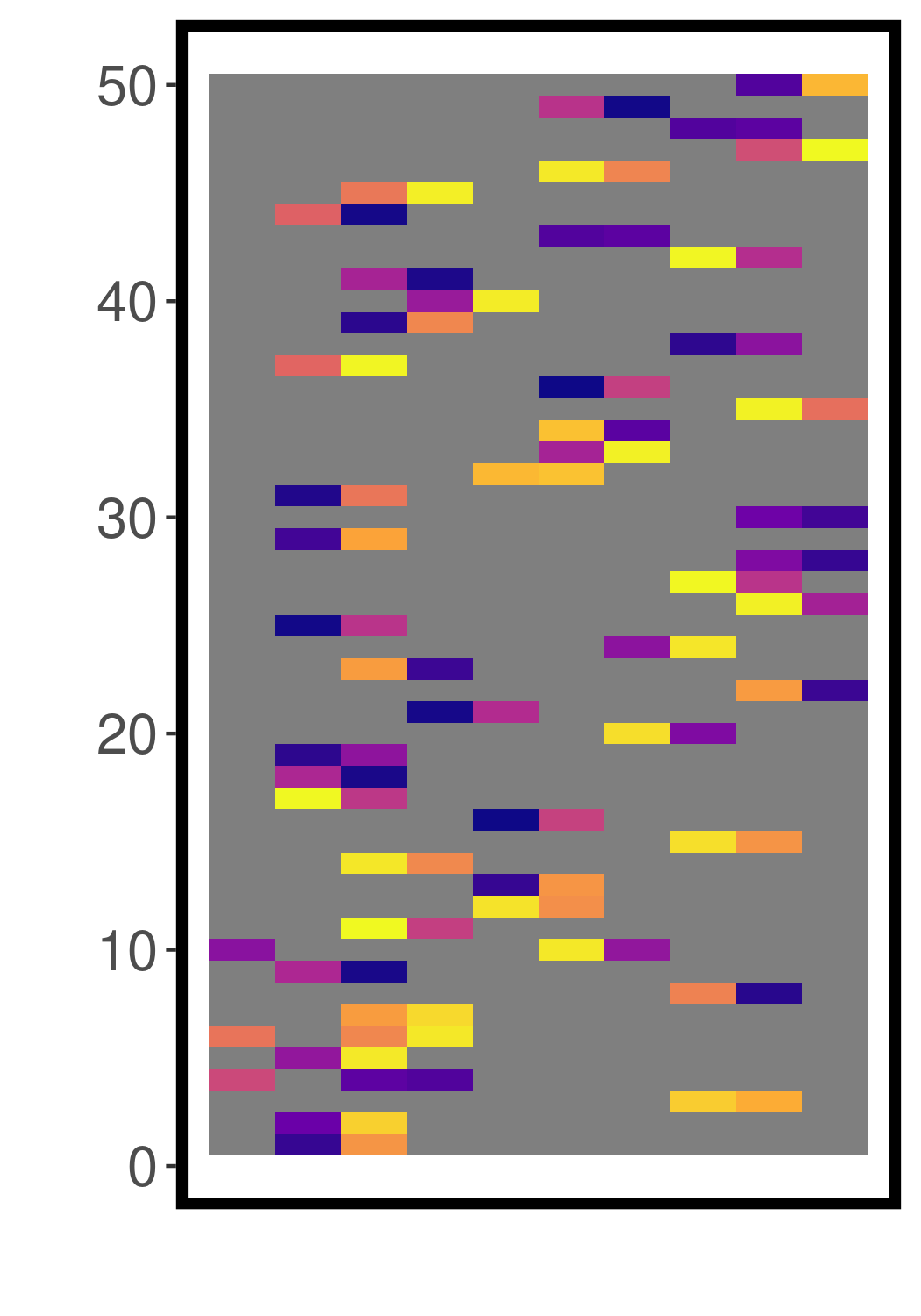} ~ \Scale[2.5]{\times}
		\includegraphics[scale=.36, trim={0cm 2cm 0cm 0cm}]{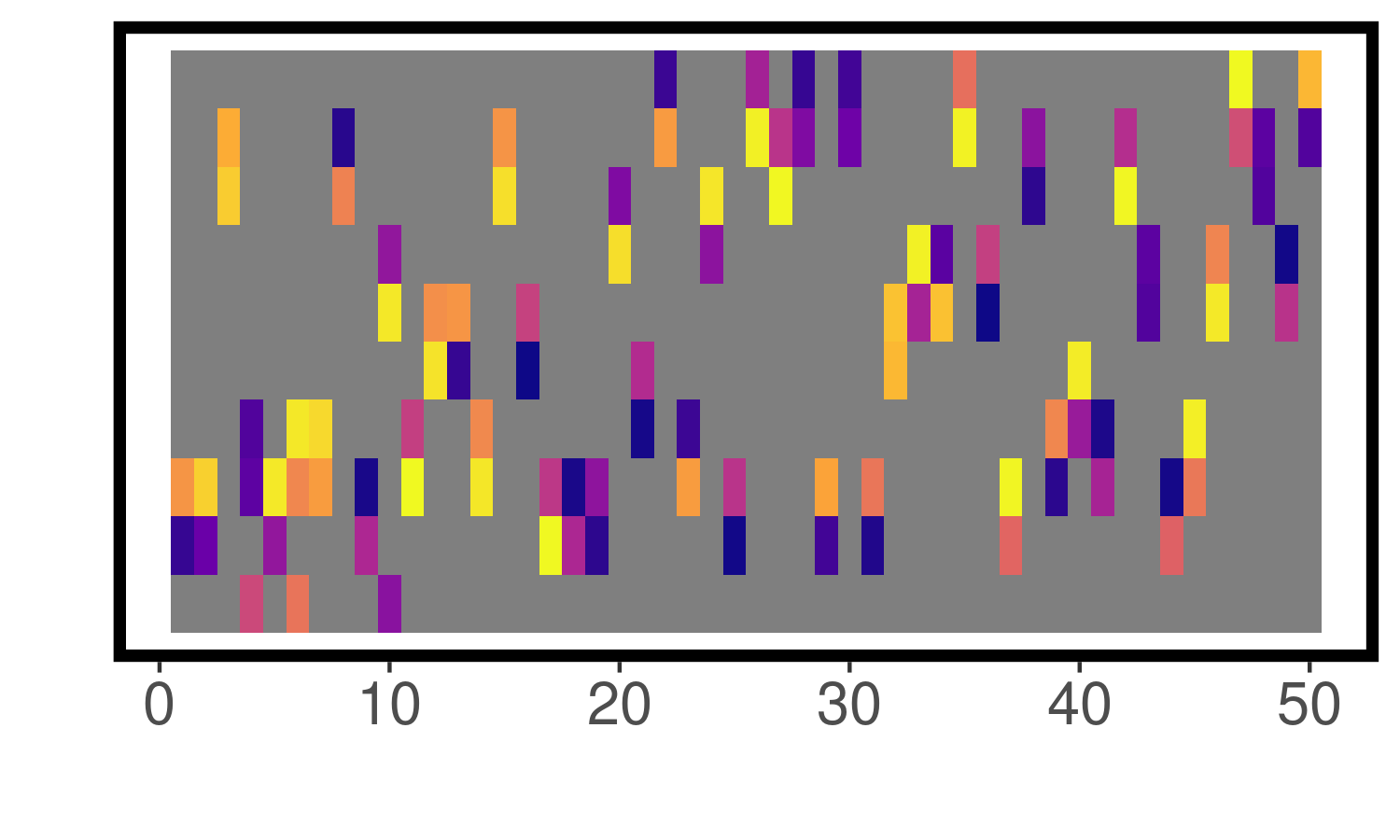} ~ \Scale[2.5]{=}
		\includegraphics[scale=.36, trim={0cm 2cm 0cm 0cm}]{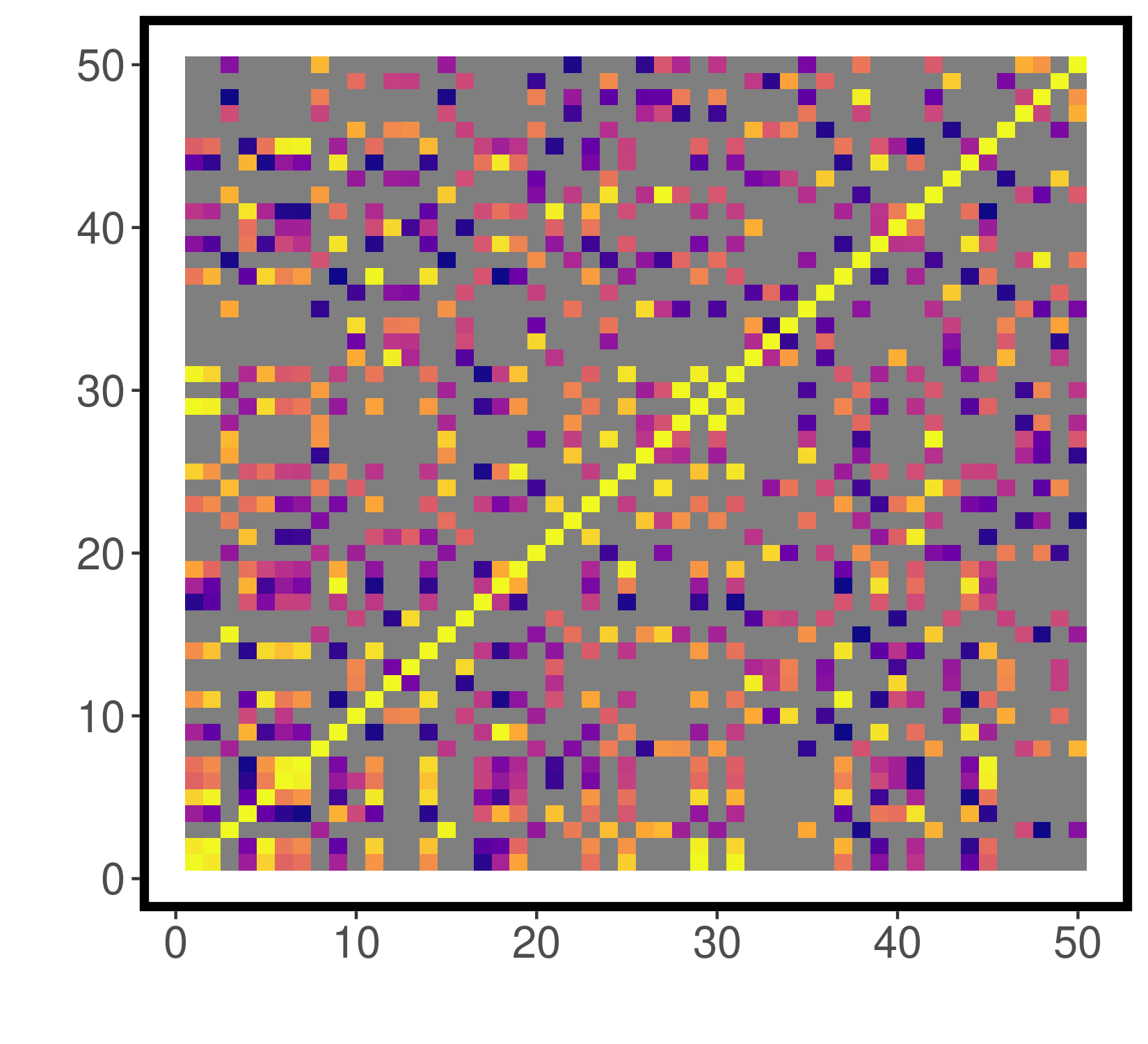}
	\end{equation*}	
	\caption{Sparse $\bLambda$ producing a sparse $\bLambda\bLambda\trans$. 
		Here, the gray cells represent exact zeros, and purple to yellow represent smaller (negative) to larger (positive) values.}
	\label{fig: sparsity patterns}
\end{figure}

The off-diagonals elements of $\bOmega=\bLambda\bLambda\trans+\bDelta$ are contributed entirely by $\bLambda\bLambda\trans$. 
{This allows to achieve sparsity in $\bOmega$ by inducing sparsity in $\bLambda$. 
	Following the discussion in Section \ref{subsec:lrd_sparse}, a carefully designed data-adaptive penalty on $\bLambda$ 
	(e.g., a shrinkage prior on $\bLambda$ that sufficiently increases the probability of obtaining a sparse $\bOmega$) 
	would induce sparsity in $\bLambda$ in such a (not necessarily unique) way that 
	the sparsity patterns in $\bOmega$ are also accurately recovered.}
In this section, we show that by inducing sparsity in $\bLambda^{d\times q}= ((\lambda_{j,h}))$, sparsity {is} also induced in $\bOmega^{d\times d}=((\omega_{j,j'}))$. 
We begin with the spike-and-slab prior \citep{Ishwaran2005spikeslab} that allows exact zeros via a spike at zero and large nonzero elements via a continuous slab $g(\cdot)$ as 
\vskip-6ex
\begin{equation*}
	\lambda_{j,h}\mid \pi, \theta_{j,h} \sim \pi \mathbbm{1}_{\{0\}}(\cdot)+ (1-\pi) g(\cdot),     \quad    \pi\sim \Beta(a_{\pi}, b_{\pi}), 
	\label{eq:globallocalprior}
\end{equation*}
\vskip-2ex
\noindent where $\pi$ is the prior probability of observing a zero for $\lambda_{j,h}$. 
For such priors $\Pi(\omega_{j,j'}=0 \mid \pi)$ can be analytically reduced to a simple form that helps gain insights into 
how inducing sparsity in $\bLambda$ can induce sparsity in $\bOmega$ with high probability.
Specifically, since $\omega_{j,j'}=\blambda_{j} \blambda_{j'}\trans$, where 
$\blambda_{j}$ is the $j\th$ row of $\bLambda$,
$\omega_{j,j'}=0$ can only happen when the $r\th$ entries of $\blambda_{j}$ and $\blambda_{j'}$ are both not from the slab distribution $g()$ for all $r=1,\dots,q$.
Therefore, we have 
\vskip-6ex
\begin{equation*}
	\textstyle{\Pi(\omega_{j,j'}=0 \lvert \pi) 
	=\sum_{r=0}^{q} \left\{{q \choose r} \pi^{r} (1-\pi)^{q-r}\right\} \times \pi^{q-r}= \pi^{q}(2-\pi)^{q}}.
\end{equation*}
\vskip-2ex
The behavior of $\Pi(\omega_{j,j'}=0 \mid \pi)$ for varying values of $\pi$ and $q$ are shown in Figure \ref{sm fig:sparsity_omega}.
It can be seen that by controlling $\pi$, it is possible to induce any desired level of sparsity in $\bOmega$.
{The hierarchical Beta prior on $\pi$ allows for data-adaptive learning  and shrink the elements of $\bOmega$ accordingly \citep{scott10}.}
\begin{figure}[h]
	\centering
	\includegraphics[width=.45\linewidth]{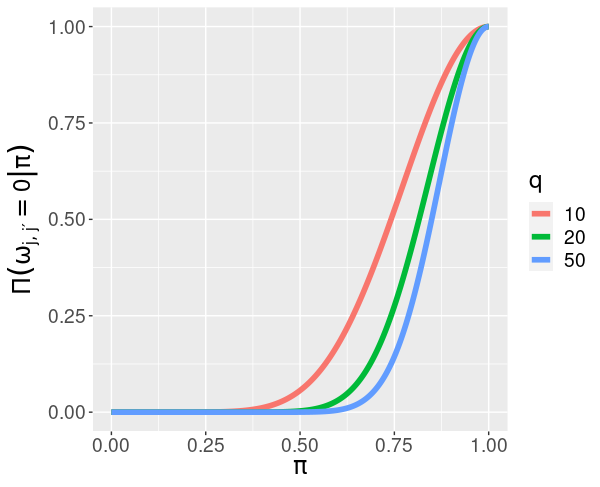}
	\caption{$\pi$ versus $\Pi(\omega_{j,j'}=0 \mid \pi)$ for different values of $q$. 
		Clearly, by controlling $\pi$, it is possible to induce desired level of sparsity in $\bOmega$.
		A hierarchical prior on $\pi$ adaptively learns from the data and shrinks the elements of $\bOmega$ accordingly.
	}
	\label{sm fig:sparsity_omega}
\end{figure}

Implementation of spike-and-slab type mixture priors is computationally challenging.
Thus, 
continuous global-local shrinkage priors have gained popularity in the Bayesian sparse estimation literature 
as they often greatly simplify posterior computation \citep{polson2010shrink} while retaining almost similar statistical properties of the classical spike-and-slab prior.
In particular, we use the Dirichlet-Laplace  \citep[DL,][]{bhattacharya2015dirichlet} prior that has equivalent asymptotic properties with respect to the spike-and-slab in a similar context \citep[Theorem 5.1]{pati2014} while being amenable to scalable posterior computation.
{The results for the spike-and-slab discussed here provides the intuitions on how continuous shrinkage priors like the DL induce sparsity in $\bOmega$ by inducing shrinkage on $\bLambda$.}


\subsubsection{Advantages of the LRD Decomposition}

For high-dimensional sparse covariance matrix estimation, this LRD decomposition strategy is extensively used in both Bayesian \citep{pati2014} and frequentist paradigms \citep{daniele2019penalizingfactor} 
where the associated latent factor formulation allows 
scalable computation in big data problems \citep[and others]{bhattacharya2011sparse,fan2011penalizingfactor,sabnis2016divide, fan2018penalizingfactor, KASTNER2019}.
Penalizing $\bLambda$ is thus a sensible approach to induce sparsity in $\bOmega$. 
It also comes with many practical advantages described below.

Cholesky factorization based covariance and precision matrix estimation methods also use a similar strategy and penalize the lower triangular matrix $\bL$ to induce sparsity in $\bOmega = \bL\bL\trans$ \citep{dallakyan2020fused}.
However, for undirected graphs, the estimates can vastly differ depending on the ordering of the variables \citep{Xiaoning2020}. 
In contrast, 
our proposed LRD decomposition based approach with a penalty on $\bLambda$ 
is invariant to the ordering of the variables. 

In most existing approaches that directly penalize $\bOmega$, 
frequentist or Bayesian, 
ensuring the positive definiteness of $\bOmega$ 
is also non-trivial, particularly in high dimensional settings.
In frequentist regimes, for example, the graphical lasso \citep{friedman2008sparse} does not guarantee positive definiteness of $\bOmega$ \citep{mazumder2012graphical};
similar behavior has also been observed for neighborhood selection methods 
\citep{meinshausen2006high, peng2009partial} 
and post-hoc treatments are required for positive definiteness. 
{In Bayesian MAP approaches such as} \citet{gan2019bayesian}, 
{special care is needed in designing the optimization algorithm to ensure positive definiteness.}
For many existing Bayesian methods relying on MCMC \citep{ wang2012bayesian, khondker2013bayesian, li2019graphical}, 
constrained update of $\bOmega$ in each step is difficult and expensive.
Applications of these methods thus remain fairly limited to small to moderate dimensional problems.
In contrast, the LRD formulation is always guaranteed to produce a positive definite estimate by construction.

	As also discussed in the Introduction, 
	posterior exploration is extremely challenging 
	in Bayesian approaches that  
	penalize the underlying the dense graphs first and then assign a prior on $\bOmega$ conditional on the graph, 
	often requiring restrictive assumptions on the graph \citep{dawid1993hyper} 
	while still 
	remaining computationally infeasible beyond only a few tens of dimensions \citep{jones2005experiments}. 
	Our proposed approach on the other hand, works directly in the $\bOmega$ space, albeit via the LRD decomposition, 
	avoiding having to define separate complex priors on the graph space 
	thereby also avoiding restrictive assumptions on the graph 
	while also achieving scalability far beyond such approaches via its latent factor representation.

	\subsection{Prior Specification} \label{sec: priors}
	To accommodate high-dimensional data, with $d \gg n$, it is crucial to reduce the effective number of parameters in the $d \times q$ loadings matrix $\bLambda$.  
	A wide variety of sparsity inducing shrinkage priors for $\bLambda$ can be considered for this purpose. 
	In this article, we employ a {two-parameter} generalization of the original Dirichlet-Laplace (DL) prior from \cite{bhattacharya2015dirichlet} that allows more flexible tail behavior. 
	On a $d$-dimensional vector $\btheta$, our DL prior with parameters $a$ and $b$, denoted by $\mbox{DL}(a,b)$, can be specified in the following hierarchical manner 
		\vskip-8ex
	\bse
	\theta_{j} \lvert \bpsi,\bphi,\tau \simind \Normal(0,\psi_{j}\phi_{j}^{2}\tau^{2}),~~\psi_{j} \simiid \Exp(1/2),~~ \bphi\sim \Dir(a,\ldots,a),~~ \tau \sim \Ga(da,b), 
	\ese
		\vskip-3ex \noindent
	where $\theta_{j}$ is the $j\th$ element of $\btheta$, $\bphi$ is a vector of same length as $\btheta$, 
	$\Exp(a)$ is an exponential distribution with mean $1/a$, 
	$\Dir(a_{1}, \dots, a_{d})$ is the $d$-dimensional Dirichlet distribution 
	and $\Ga(a,b)$ is the gamma distribution with mean $a/b$ and variance $a/b^{2}$. 
	The original DL prior is a special case with $b=1/2$.
	We let $\mbox{vec}(\bLambda) \sim \mbox{DL}(a,b)$. 
	
	The column dimension $q$ of $\bLambda$ will almost always be unknown.
	Assigning a prior on $q$ and implementing a reversible jump MCMC \citep{rjmcmc} type algorithm can be inefficient and expensive.
	In this paper, we adopt an empirical Bayes type approach to set $q$ to a large value determined from the data 
	and let the prior shrink the extra columns to zeros, substantially simplifying the computation. 
	The strategy is discussed in greater details later in this section.
	We show in Section \ref{sec: asymptotics} that this approach is sufficient for the recovery of true $\bOmega$ 
	under very mild conditions.

	When $d\gg n$, $d$ distinct $\delta_{j}^{2}$'s also result in over-parametrization.
	To reduce the number of parameters, we assume the $\delta_{j}^{2}$'s to comprise a small number of unique values. 
	To achieve this in a data adaptive way, we use a Dirichlet process (DP) prior \citep{ferguson73} on the $\delta_{j}^{2}$'s as 
		\vskip-8ex
	\be
	\delta_{j}^{2} \lvert G\simiid G,~~~G\lvert \alpha \sim \mathrm{DP}(\alpha,G_{0})\text{ with } G_{0}=\Ga(a_{\delta},b_{\delta}), ~~~\alpha \sim \Ga(a_{\alpha},b_{\alpha}),    \label{eq:dpm}
	\ee
		\vskip-4ex \noindent
	where 
	$\alpha$ is the concentration parameter and $G_{0}$ is the base measure.
	Samples drawn from a DP are almost surely discrete, inducing a clustering of the $\delta_{j}^{2}$'s and thus reducing the dimension. 
	Integrating out $G$, model \eqref{eq:dpm} leads to a recursive Polya urn scheme 
	illustrative of the clustering mechanism 
	while also being convenient for posterior computation \citep{escobar1995bayesian, neal2000markov}. 
	Specifically, we have
		\vskip-8ex
	\bse
	\textstyle \delta_{j+1}^{2}\lvert \alpha, \bdelta_{1:j}^{2} \propto \sum_{\ell=1}^{k_{j}} d_{j,\ell} \mathbbm{1}(\delta_{j+1}^{2}=\delta_{\ell}^{*2})+\alpha G_{0}(\delta_{j+1}^{2}),
	\ese
		\vskip-2ex \noindent
	where 
	$\{\delta_{1}^{*2},\dots,\delta_{k_{j}}^{*2}\}$ denote the unique values among $\bdelta_{1:j}^{2}$ 
	and $d_{j,r}=\sum_{\ell=1}^{j} \mathbbm{1}(c_{\ell}=r)$ denote their multiplicities, 
	$c_{1},\dots,c_{d}$ being the latent variables 
	such that $c_{j}=r$ if $\delta_{j}^{2}$ belongs to the $r\th$ cluster. 
	
	
	\noindent\textbf{Choice of Hyperparameters:} 
	To specify the value of the latent dimension $q$, we adopt a principal component analysis (PCA) based empirical Bayes type approach \citep{bai2008factorest}.
	First, we perform a sparse PCA \citep{irlba} on the data matrix $\by_{1:n}$  and compute the reciprocals of the singular values.
	We set $q$ to be the number of inverse singular values in decreasing order that adds up to $95\%$ of the total sum.
	We set $a=0.5$ and $b=2.0$ in all our simulation studies and real data applications.
	For the prior on the residual variances and the DP concentration parameter, we set $a_{\delta}=b_{\delta}=a_{\alpha}=b_{\alpha}=0.1$.

	\subsection{Posterior Computation} \label{sec: post inference main}
	The latent factor construction discussed in Section \ref{sec: precision factor analysis} leads to a novel, elegant Gibbs sampler 
	that is operationally simple and free of any tuning parameter. 
	There is also no need for additional 
	constraints to ensure positive definiteness 
	which used to be a major setback for MCMC based methods for precision matrix estimation. 
	These features allow the sampler to be applied to dimensions far beyond the reach of the current state-of-the-art. 
	Scalable posterior computation in LRD decomposed covariance matrix models 
	is pretty standard in the literature 
	but remained a major barrier for precision matrix models. 
	The Gibbs sampler described here addresses this significant gap in the literature.
	
	
	Starting with some initial values of $\bLambda,\bDelta$ and other parameters, our sampler iterates between the following steps.
	Other parameters and hyperparameters being implicitly understood in the conditioning, 
	Step 1 defines a transition for $\bu,\bv \vert \by,\bLambda,\bDelta$; 
	Step 2 for $\bLambda \vert \bu, \bv$ (which is identical to $\bLambda \vert \by,\bu,\bv, \bDelta$); 
	Step 3 for $\bDelta \vert \bv$ (which is identical to $\bDelta \vert \by,\bu,\bv, \bLambda$);
	and Steps 4 and 5 are standard updates for the parameters and hyper-parameters of the DL and DP priors, respectively. 
	
	\begin{description}[leftmargin=0pt]
		\item[Step 1]  Generate $\bu_{1},\dots,\bu_{n} \simiid \mn_{q}(\bzero,\bP)$ with $\bP = (\bI_{q}+\bLambda\trans\bDelta^{-1}\bLambda)$ independently from $\by_{1:n}$ and let $\bv_{i} = \by_{i}+\bDelta^{-1}\bLambda \bP^{-1}\bu_{i}$. 
		
		\item[Step 2] We have $\bu_{i}=\sum_{r=1}^{d} \blambda_{r} v_{r,i}+\bvarepsilon_{i}$, 
		where $\blambda_{r} = (\lambda_{r,1},\dots,\lambda_{r,q})$ is the $r\th$ row of $\bLambda$ and $\bv_{i}=(v_{1,i},\dots,v_{d,i})\trans$. 
		Define $\bu_{i}^{(j)}=\bu_{i}-\sum_{r\neq j}\blambda_{r} v_{r,i}$. 
		Then $\bu_{i}^{(j)}=\blambda_{j} v_{j,i}+\bvarepsilon_{i}$. 
		Conditioned on $\bu_{i}^{(j)}$, $\bv_{i}$ {and the associated hyper-parameters}, $\blambda_{j}$'s can be updated sequentially for $j=1,\dots,d$ from the distribution
				\vspace{-4ex}
		\bse
		\blambda_{j} \sim \mn_{q}\{(\bD_{j}^{-1} + \norm{\bv^{(j)}}^{2} \bI_q )^{-1} \bw_{j},  (\bD_{j}^{-1} + \norm{\bv^{(j)}}^{2} \bI_q )^{-1}\},	
		\ese
				\vskip-4ex
		where  $\bD_{j}=\tau^{2}\diag\left( \psi_{j,1}\phi^{2}_{j,1},\dots, \psi_{j,q}\phi^{2}_{j,q} \right)$, $\bv^{(j)}=(v_{j,1},\dots,v_{j,n})\trans$ 
		and $\bw_{j}=\sum_{i=1}^n v_{j,i} \bu_{i}^{(j)}$.

		\item[Step 3] 	
		Sample the $\delta_{j}^{2}$'s through the following steps. 
		\begin{enumerate}[label=({\roman*})]		
			\item\label{step2} Let $d_{r,-j}=\sum_{\ell\neq j}\mathbbm{1}(c_{\ell}=r)$ and $\bv^{(-j)}$ to be the collection of all $\bv^{(\ell)}$'s, $\ell=1,\dots,d$, excluding $\bv^{(j)}$.
			For $j=1,\dots,d$, sample the cluster indicators sequentially from the distribution
						\vspace{-7ex}\\
			\bse
			p( c_{j}=r ) \propto\begin{cases}
				d_{r,-j} \int  \mn (\bv^{(j)}; 0,  \delta_{r}^{*-2} ) \de G_{0}\left(\delta_{r}^{*2}\lvert \bv^{(-j)}\right) \text{ for } r\in\{c_{\ell}\}_{\ell \neq j};\\
				\alpha \int \mn (\bv^{(j)}; 0,  \delta_{r}^{*-2} )\de G_{0} (\delta_{r}^{*2})\text{ for }r\neq c_{\ell} \text{ for all } \ell \neq j.
			\end{cases} 
			\ese
						\vspace{-6ex}\\
			The above integrals are analytically available and involves the density of a multivariate central Student's $t$-distribution for $G_{0} = \Ga(a_{\delta},b_{\delta})$.
			\item Let the unique values in $\bc_{1:d}$ be \{$1,\dots,k\}$. 
			For $r=1,\dots,k$, 
			set $d_{r}=\sum_{j}\mathbbm{1}(c_{j}=r)$ and $\bV_{r}=\sum_{j:c_{j}=r} \norm{\bv^{(j)}}^{2}$, and 
			independently sample $\delta_{r}^{*2} \sim \Ga\left(a_\delta + {n d_{r}}/{2}, b_\delta + {\bV_{r}}/{2}  \right)$.
			\item Set $\delta_{j}^{2}= \delta_{c_{j}}^{*2}$.
		\end{enumerate}
		
		\item[Step 4] Sample the hyper-parameters in the priors on $\bLambda$ through the following steps. 
		\begin{enumerate}[label=({\roman*})]
			\item For $j=1, \dots, d$ and $h=1,\dots q$ sample $\widetilde{\psi}_{j,h}$ independently from an inverse-Gaussian distribution $\mbox{iG}\left(\tau{\phi_{j,h}}/{\abs{\lambda_{j,h}}},1\right)$ and set $\psi_{j,h}=1/\widetilde{\psi}_{j,h}$.
			\item Sample the full conditional posterior distribution of $\tau$  from a generalized inverse Gaussian  $\mbox{giG}\left\{dq(1-a),2b, 2\sum_{j,h} {\abs{\lambda_{j,h}}}/{\phi_{j,h} }\right\}$ distribution.
			\item Draw $T_{j,h}$ independently with $T_{j,h} \sim \mbox{giG}(a - 1, 1, 2\abs{\lambda_{j,h}} )$ and set $\phi_{j,h}  = T_{j,h}/T$ with $T=\sum_{j,h} T_{j,h}$.
		\end{enumerate}

		\item[Step 5] Following \cite{west1992hyperparameter}, 
		first generate $\varphi\sim \Beta(\alpha+1,d)$, 
		evaluate $\pi/(1-\pi)=(a_{\alpha}+k-1)/\left\{d(b_{\alpha}-\log\varphi)\right\}$ and then generate\\
				\vspace{-7ex}
		\bse
		\alpha\lvert \varphi, k \sim \begin{cases}
			\Ga(\alpha+k,b_{\alpha}-\log\varphi)\text{ with probability }\pi, \\
			\Ga (\alpha+k-1,b_{\alpha}-\log\varphi )\text{ with probability }1-\pi.
		\end{cases}
		\ese
				\vspace{-7ex}	
		
	\end{description}
	

	\begin{remark}
		Note that the main strategies underlying the algorithm above are not specific to the DL prior considered here. 
		We are free to choose any other shrinkage prior that admits a conditionally Gaussian hierarchical representation 
		for the entries of $\bLambda$ and modify Steps 2 and 4 accordingly. 
		This is still a very large class of priors \citep{polson2010shrink}, 
		including, e.g., 
		horseshoe \citep{carvalho2009handling}, 
		multiplicative gamma \citep{bhattacharya2011sparse}, 
		etc. 
		For non-conjugate priors, Steps 3 and 4 can also be modified with appropriate Metropolis-Hastings schemes. 
		The other steps remain the same, making it a very broadly adaptable algorithm. 
		An MCMC scheme for generic priors is outlined in Section \ref{sec: sm post computation} of the supplementary materials. 
	\end{remark}
	\begin{remark}
		In the above sampler, the conditional posterior covariance matrix of $\blambda_{j}$ is diagonal which allows us to update it with linear complexity. 
		Thus, in each MCMC iteration, we only need a single small dimensional $q\times q$ order matrix factorization operation in Step 1 to simulate $\bu_{1:n}$. 
		While high-dimensional matrix factorization operations are usally numerically very expensive, we are able to completely avoid that, facilitating substantial scalability. 
	\end{remark}
	\begin{remark}
		All but Steps 1 and 3 can be divided into parallel operations in a straightforward manner. 
		Additionally, in recent versions of many popular statistical software, including \texttt{R}, 
		matrix operations are inherently parallelized and hence Step 1 is also highly scalable in any decent computing system. 
	\end{remark}

	
	We implemented the Gibbs sampler in \texttt{C++} and ported to \texttt{R} using
	the \texttt{Rcpp} package \citep{rcpp}.
	In each case considered in this article, synthetic or real, we ran $5,500$ iterations
	which takes approximately {37 minutes} on {a system with an i9-10900K CPU and 64GB memory} for a $d=1,000$ dimensional problem with sample size $n=1,000$.
	The initial $1,250$ samples were discarded as burn-in and  
	the remaining samples were thinned by an interval of $5$. 
	In all our experiments, 
	convergence was swift and mixing was excellent. 
	A comparison of the runtimes of our method and a few other existing methods is presented in Figure \ref{fig:time} below.

	\subsection{Graph Selection} \label{sec: graph}
	{As discussed in Section \ref{sec: sm sparsity}, the off-diagonals of $\bOmega$ are penalized by inducing shrinkage on $\bLambda$.
		The theoretical results in Section \ref{sec: asymptotics} and
		numerical experiments in Section \ref{sec: sim studies} show that,
		with our carefully constructed data adaptive shrinkage priors on $\bLambda$, 
		the inferred sparsity patterns in $\bLambda$ are such that 
		a sparse $\bOmega$ is also accurately recovered 
		(see Figures \ref{fig:quantiles} and \ref{fig:precision_PF} for estimates of sparse precision matrices obtained by our method).
		Exact zero estimates are, however, not obtained even for the insignificant off-diagonal elements of $\bOmega$ {for finite samples}.
		This is an artifact of continuous shrinkage priors, since the probabilities of exact zeroes are {\textit{almost surely} null for finite samples} although the posterior probabilities of arbitrary sets around zeroes are very high.}
	
	
	We address the issue of non-zero edge selection through a novel multiple hypothesis testing based approach. 
	For $i=1,\dots,d,~j=i+1,\dots,d$ and some $\epsilon>0$, 
	we consider testing 
		\vskip-12ex
	\be
	H_{0,i,j}: \abs{\rho_{i,j}}\leq \epsilon ~~~ \text{ versus } ~~~ H_{1,i,j}:\abs{\rho_{i,j}}> \epsilon,
	\label{eq:hypothesis}
	\ee
		\vskip-3.5ex \noindent
	where $\rho_{i,j}$ is the $(i,j)\th$ element of  $\diag(\bOmega)^{-\half} ~\bOmega~ \diag(\bOmega)^{-\half}$, the partial correlation matrix derived from $\bOmega$. 
	Here we follow \citet[Chapter 4, pp. 148]{berger_book} in replacing the point nulls $H_{0,i,j}: \rho_{i,j} = 0$ by reasonable interval nulls $H_{0,i,j}: \abs{\rho_{i,j}}\leq \epsilon$. 
	If $H_{0,i,j}$ is rejected in favor of $H_{1,i,j}$, we conclude that there is an edge between nodes $i$ and $j$.
	
	We utilize posterior uncertainty to resolve these testing problems. 
	Specifically, we define $d_{i,j}=\mathbbm{1}\left\{\Pi(H_{1,i,j} \vert \by_{1:n}) >\beta \right\}$ as the decision rule 
	which controls the posterior FDR 
	defined as 
		\vskip-10ex
	\bse
	\FDR_{\by} = \frac{\sum_{i,j} d_{i,j} \Pi(H_{0,i,j} \vert \by_{1:n})}{\max (\sum_{i,j} d_{i,j},1)},
	\ese
		\vskip-4ex\noindent
	at the level $1-\beta$. 
	Importantly, the decision rule also incurs the lowest false non-discovery rate \citep{muller04}. 
	For a fixed $\beta$, the $\FDR_{\by}$ depends on the choice of $\epsilon$. 
	To obtain the optimal $\epsilon$, we compute the $\FDR_{\by}$'s on a grid of $\epsilon$ values in $(0,1)$ 
	and then set $\epsilon = \inf_{\epsilon'}\mathrm{FDR}_{\by}(\epsilon')\leq 1-\beta$. 
	This way, we have the FDR, that is, a quantified statistical uncertainty associated with the estimated graph. 
	In all simulation experiments and real data applications in this paper, we control $\FDR_{\by}$ at the $0.10$ level of significance. 
	
	\begin{remark}
		Although this FDR control procedure is widely applicable, efficient posterior exploration is crucial for this step and hence may not be achievable or scalable to high dimensional problems when adapted to previously existing Bayesian approaches.
	\end{remark}

	\subsection{Posterior Concentration} \label{sec: asymptotics}
	
	\noindent{\bf Preliminaries and Notation: } 
	We let $\fnorm{\bA}$  and $\specnorm{\bA}$ denote the Frobenius and spectral norms of a matrix $\bA$, respectively; 
	$s^{2}_{\min}{(\bA)}$ be the smallest singular value of $\bA\trans \bA$; 
	and $\lambda_{1}(\bA),\dots,\lambda_{d}(\bA)$ be the eigenvalues of $\bA$ in decreasing order when $\bA$ is a $d$-dimensional diagonalizable matrix.
	Also, 
	$\an=o(b_{n})$ and $\an=O(b_{n})$ imply that $\lim \abs{{\an}/{b_{n}}}=0$ and $\limsup \abs{{\an}/{b_{n}}} <\infty$, respectively.
	Throughout, $C, C'$ and $\wt{C}$ are used to denote positive constants whose values might change from one line to the next but are independent from everything else.
	A $d$-dimensional vector $\btheta$ is said to be $s$-sparse if only $s$ among the $d$ elements of $\btheta$ are non-zero.
	We denote the set of all $s$-sparse vectors in $\rR^{d}$ by $\lnt{s}{d}$.
	
	We allow the model parameters to increase in dimensions with sample size $n$, indicated by  associating them with the suffix $n$.
	We let $\pin(\cdot)$ denote the prior and $\pin(\cdot\lvert \by_{1:n})$ the corresponding posterior given data $\by_{1:n}$, respectively.
	
	\noindent{\bf Assumptions on the Data Generating Process: }
	We assume that $\by_{i}\simiid \mn_{\dn}(\bzero, \bOmegann^{-1}), i=1,\dots,n$.
	{In Section \ref{subsec:lrd_sparse} we discussed how an LRD decomposition of a sparse $\bOmegann$ can be constructed.
		Hence, we assume that} 
	$\bOmegann$ admits the factor representation $\bOmegann=\bLambdann\bLambdann\trans+\deltann^{2}\bI_{\dn}$, 
	where $\bLambdann$ is a $\dn\times\qnn$ order sparse matrix and $\deltann^{2}>0$ is a scalar.
	To simplify the theoretical analysis, we deviate here slightly from the proposed model and assume 
	that the $\delta^{2}$'s all come from a single cluster with the common value $\deltann^{2}$.
	We show that the precision factor model with an appropriate DL prior can recover $\bOmegann$ in ultra high-dimensional settings.
	As discussed in Section \ref{sec: priors}, we fix $\qn$ to a liberal large value and let the prior shrink the extra columns to zeros.
	We show that with appropriate sparsity conditions, the posterior of the precision factor model with the DL prior then concentrates around $\bOmegann$, 
	even when the data dimension $\dn$ increases in exponential order with $n$.
	The requisite conditions 
	are stated below.
	\begin{enumerate}[label=(C\arabic*)]					
		\item \label{ass1} 
		Let $\{\qnn\}_{n=1}^{\infty}$ and $\{\sn\}_{n=1}^{\infty}$ be increasing sequences of positive integers such that $\sn^{2}=O(\log\dn)$, $\qnn=O\{\sn\log(\dn\qnn)\}$ {and $\sn\qnn  \log(\dn\qnn)=o(n)$}.
		\item \label{ass3} 
		$\bLambdann$ is a $\dn\times\qnn$ order full rank matrix such that each column of $\bLambdann$ belongs to $\lnt{\sn}{\dn}$, $\liminf \smin{\bLambdann}>0$ and $\specnorm{\bLambdann}=O(1)$.
		\item \label{ass2} 
		The scalar $\deltann^{2}$ lies in some compact set $[\delta_{\min}^{2},\delta_{\max}^{2}]$. 
	\end{enumerate}
	

	\noindent{\bf Specifics of the Postulated Precision Factor Model: } 
	Let $\{\qn\}_{n=1}^{\infty}$ be an increasing sequence of positive integers such that $\qnn\leq\qn=o(n)$ 
	and $\sn\qnn  \log(\dn\qn)=o(n)$. 
	We assume that 
	$\by_{1:n}\simiid \mn_{\dn}(\bzero, \bOmegan^{-1})$, where $\bOmegan=\bLambdan\bLambdan\trans+ \deltan^{2}\bI_{\dn}$ and $\bLambdan$ is a $\dn\times \qn$ matrix.
	We consider a $\mathrm{DL}(\an,\bnn)$ prior on $\vect(\bLambdan)$ with $\an={1}/{\dn\qn}$ and $\bnn=\log^{3/2}(\dn\qn)$. 
	We assume a gamma prior on  $\deltan^{2}\sim\Ga(a_{\delta},b_{\delta})$ independent of $\bLambdan$ and truncated to the compact set $[\delta_{\min}^{2},\delta_{\max}^{2}]$.
	\vskip 10pt

	Condition \ref{ass1} specifies the requisite sparsity conditions. 
	It also imposes a condition on $\dn$.
	Specifically, it can be seen that $\dn$ can be of exponential order of $n$, 
	allowing the recovery of massive precision matrices based on relatively small sample sizes.	
	Note that, subject to appropriate orthogonal transformation, 
	marginally $\var(y_{j})=1/\left\{\lambda_{j}\left(\bLambdann\trans\bLambdann\right)+\deltann^{2} \right\}$.
	Conditions \ref{ass3} and \ref{ass2} ensure that these variances lie in a compact set.
	The prior specification provides a set of sufficient conditions on the class of proposed models. 
	The true number of latent factors $\qnn$ is also assumed smaller than the number of latent factors $\qn$ in the postulated models.
	
	We let $\Pnn$ denote the class of precision matrices satisfying \ref{ass1}-\ref{ass2} 
	and $\Pn$ denote the class of positive definite matrices parametrized by $(\qn,\deltan^{2},\bLambdan)$ in our precision factor model.
	Notably, $\Pnn$ can also be parametrized by $(\qnn,\deltann^{2},\bLambdann)$.
	However, since $\qnn\leq \qn$, for a fixed $\qn$, $\Pn$ may be overparametrized for $\Pnn$. 
	In Theorem \ref{thm: posterior concentration}, we show that with increasing sample size,
	the posterior distribution supported on $\Pn$ still concentrates around the true $\bOmegann$ from $\Pnn$. 
	Using general results from \citet{ghosal_book}, 
	we establish a convergence rate in terms of the operator norm. 
	A rigorous proof is detailed in Section \ref{sec: sm proofs} in the supplementary materials. 
	
	\begin{theorem}
		\label{thm: posterior concentration}
		For any $\bOmegann\in \Pnn$, 
		$\epsilon_{n} = (\sn\qnn)^{4}\log(\dn\qn)\sqrt{\qnn\sn\log(\dn\qn)/n}$ 
		and any sequence $M_{n}\to\infty$, 
				\vspace{-10ex}\\
		\bse
		\lim_{n\to \infty} \eE_{\bOmegann} \pin \left(\specnorm{\bOmegan-\bOmegann}>M_{n}\epsilon_{n}\lvert \by_{1:n} \right) = 0.
		\ese
	\end{theorem}
		\vspace{-4ex}
	
	{In Section \ref{subsec:lrd_sparse} we discussed a constructive way to find a sparse LRD decomposition of arbitrary sparse $\bOmegann$ where the column dimension $\qnn$ of $\bLambdann$ need not exceed the number of non-zero off-diagonals of $\bOmegann$ or edges in the conditional dependence graph.
		Theorem \ref{thm: posterior concentration} implies that the precision matrix is recoverable if \ref{ass1} holds among others.
		Note that \ref{ass1} implies $\qnn=o(n)$.
		Hence, $\bOmegann$ can be learned from the data using the precision factor model as long as the number of edges is bounded by the sample size.
		Such assumptions are required to recover massive dimensional precision matrices from relatively smaller amount of data \citep{meinshausen2006high,KsheeraSagar2021precision}.
	} 
	
	\begin{remark}
		We derive a minimax lower bound of $\sqrt{\sn\log\dn/n}$ for the precision matrix estimation problem in Theorem \ref{thm: sm minimax} in the supplementary materials when $\qnn\leq\qn=O(1)$ {which is a special case of \ref{ass1}}.
		The rate in Theorem \ref{thm: posterior concentration} is ${\sn^{4}\log\dn}$ times the lower bound. 
		If we further let $\dn=O(n^{r})$ for some fixed positive integer $r$, 
		the rate attains the minimax lower bound up to a $(\log n)^{3}$ term. 
	\end{remark}

	\begin{remark}
		An interesting implication of our theoretical results is the robustness with respect to the choice of the latent dimension $q$ under very mild conditions.
		For the class of our postulated precision factor models, we err on the side of overestimating $q$ and let the DL prior shrink the extra factors to zero.
		The theoretical results imply that this strategy is sufficient to recover $\bOmegann$ efficiently.
		The precise recovery rate, however, naturally depends on the choice of $\qn$, 
		a key parameter that distinguishes the class of postulated models from the class of true models \citep{shalizi2009}. 
		The closer $\qn$ is to $\qnn$, the smaller the space to search for the truth, and the better the rate. 
	\end{remark}

\section{Simulation Studies} \label{sec: sim studies}
In this section, we discuss the results of some synthetic numerical experiments. 
We evaluate the performance of estimating the precision matrix $\bOmega$ itself as well as the underlying graph. 
We simulate data from three different cases - 
(i) a Gaussian autoregressive (AR) process of order 2, 
(ii) a multivariate Gaussian distribution with banded precision matrix, and 
(iii) a multivariate Gaussian distribution with randomly generated arbitrarily structured sparse (RSM) precision matrix. 
In all these cases, we take the mean vector to be zero. 
For plotting purposes in the RSM case, we assume two nodes to have an edge between them
if the absolute value of the associated partial correlation exceeds $0.1$ and refer to this as the `true' graph. 
We perform experiments for dimensions $d=$ 50, 100, 200 and 1,000. For $d\leq 200$ we take $n=100$ and for $d=1,000$ we take $n=1,000$. 
We plot the true precision matrices and the associated true graphs for $d=50$ in Figure \ref{fig:quantiles} in the supplementary materials 
and in Figure \ref{fig: truegraphs} here in the main paper.
{For $d \leq 200$ and $d=1,000$ we consider 50 and 20 independent replications for each scenario, respectively.}


We apply our precision factor (PF) model to recover the precision matrices, using the posterior mean as our Bayesian point estimate.
{We compare it with the methods Bayesian graphical model under shrinkage \citep[Bagus,][]{gan2019bayesian}, 
	the graphical Lasso \citep[Glasso,][]{friedman2008sparse} and the neighborhood selection  by \citet[M\&B,][]{meinshausen2006high}.} 
{We do not consider any previously existing Bayesian posterior sampling based methods here as they do not scale well beyond only a few tens of dimensions.} 
We assess the comparative efficacy of these methods using qualitative graphical summaries as well as quantitative performance measures. 
For the proposed PF method, the graphs are estimated following the FDR control procedure outlined in Section \ref{sec: graph}. 
For Bagus, we follow the prescription in Section 4.3 of \citet{gan2019bayesian}. 
{For $d=1,000$ in the banded case, the codes for Bagus did not work and we could not consider it as a competitor in that particular setup.}
For the other methods, the non-zero entries in the estimated precision matrix are considered as edges.

We compute the Frobenius norm between the true and estimated precision correlation matrices. 
To assess the accuracy of the derived graphs, we compute 
specificity $=\frac{TN}{TN+FP}$ and 
sensitivity $=\frac{TP}{TP+FN}$, 
where $TP$ (true positives), $FP$ (false positives), $TN$ (true negatives) and $FN$ (false negatives) 
are based on the detection of the edges in the estimated graphs and comparing them with the corresponding true graphs. 
In Figure \ref{fig:simstudy}, the average values of these criteria across the replications are plotted for the competing methods. 
We also compare the execution times of the different methods.

From Figure \ref{fig:simstudy}, we see that, 
for banded precision matrices, 
the Frobenius norms obtained by our proposed PF method are very similar to those produced by the competitors. 
In the case of the arbitrarily structured precision matrix, however, the PF method performs much better, especially in higher dimensions. 
Also, the PF method yields much higher sensitivity in almost all the scenarios {at the expense of} slightly smaller specificity in some cases. 
This is expected since we are allowing a small margin of error by controlling the FDR at $0.10$.
Notably, the PF method is much more powerful in detecting the true edges. 
The other methods seem to be quite conservative in that regard, as seen in Figure \ref{fig:roc}.
{This explains the high Frobenius norm produced by the PF method in the AR(2) case with $d=1,000$. The true precision matrix being extremely sparse with only two off-diagonals being non-zero, the conservative competitors are performing better in recovering the precision matrix in this particular setup.}
We see in Figure \ref{fig:time} that the PF method is slower than M\&B and Glasso but is much faster than Bagus. 
{Note however that, unlike the competitors, we are exploring the full posterior, not just providing a point estimate.} 
On a related important note, since we are able to sample from the full posterior, unlike the other methods, 
we are able to provide a natural way of quantifying posterior uncertainty. 
Specifically, posterior credible intervals for each element of the precision matrix can be obtained from the MCMC samples. 
For $d=50$ and sample size $n=100$, we plot the lower $2.5\%$ and upper $97.5\%$ quantiles of the entry-wise partial correlations along 
with the simulation truths in Figure \ref{fig:quantiles} in the supplementary materials. 
The plots indicate that the simulation truths are well-within the confidence bounds.

Circos plots \citep{gu2014circlize} of the true and estimated graphs are shown in Figure \ref{fig: graphs}. 
For the AR(2) case, none of the methods are doing well in recovering the graph. 
For the other cases, the proposed PF method outperforms the competitors.
{Figures \ref{fig: sm graphs} and \ref{fig:prec_estimates} in the supplementary materials provide alternate graphical representations 
	useful in visually discerning the superior performance of the PF method.}

\begin{figure}[!ht]
	\centering
	\begin{subfigure}[b]{.99\linewidth}
		\includegraphics[width=\linewidth]{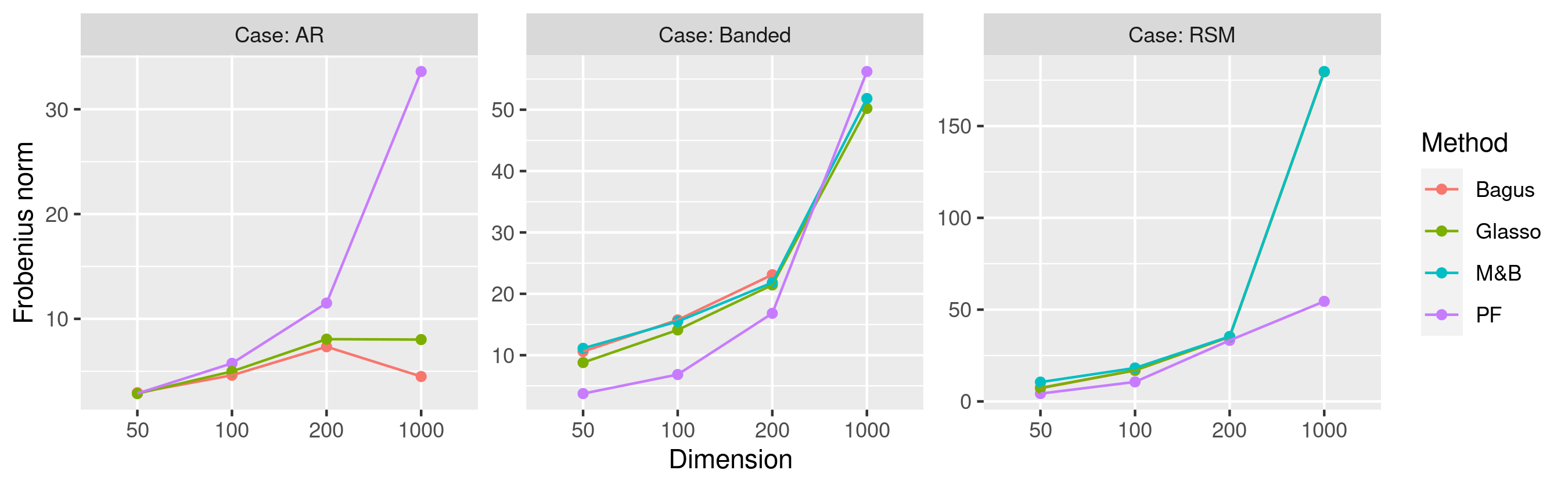}
		\caption{Frobenius norm}
		\label{fig:fnorm}
	\end{subfigure}
	\vskip 10pt
	\begin{subfigure}[b]{.99\linewidth}
		\includegraphics[width=\linewidth]{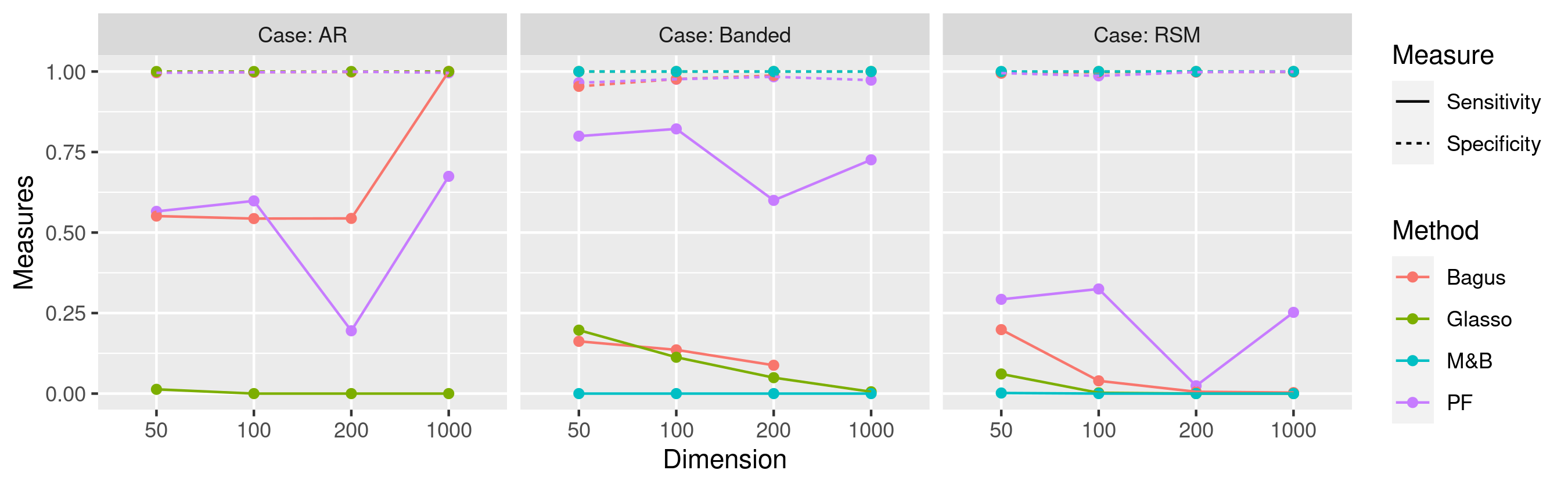}
		\caption{Sensitivity and specificity}
		\label{fig:roc}
	\end{subfigure}
	\vskip 10pt
	\begin{subfigure}[b]{.99\linewidth}
		\includegraphics[width=\linewidth]{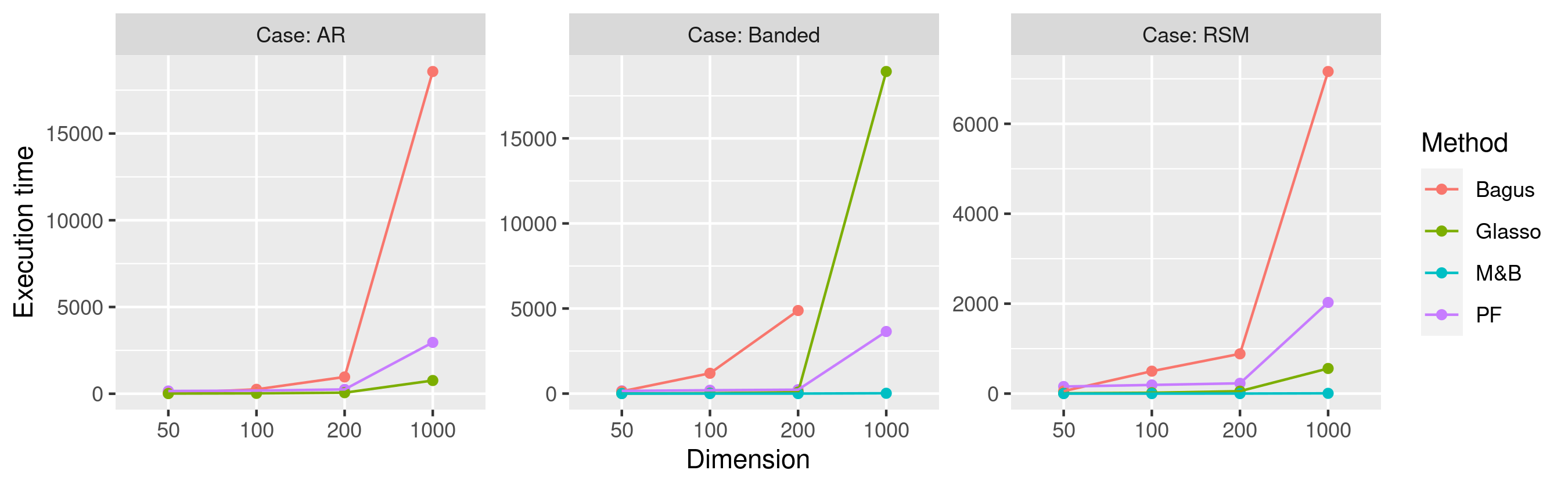}
		\caption{Execution time}
		\label{fig:time}
	\end{subfigure}
	\caption{Results of simulation experiments: Panel (a) shows the Frobenius norms between the true and estimated partial correlation matrices; 
		panel (b) shows the sensitivity and specificity; 
		and panel (c) shows the execution times in seconds.}
	\label{fig:simstudy}
\end{figure}

\begin{figure}[!ht]
	\begin{subfigure}[b]{.99\linewidth}
		\centering
		\includegraphics[trim={0.5cm 0.5cm 0.5cm 0.5cm}, clip, width=0.325\linewidth]{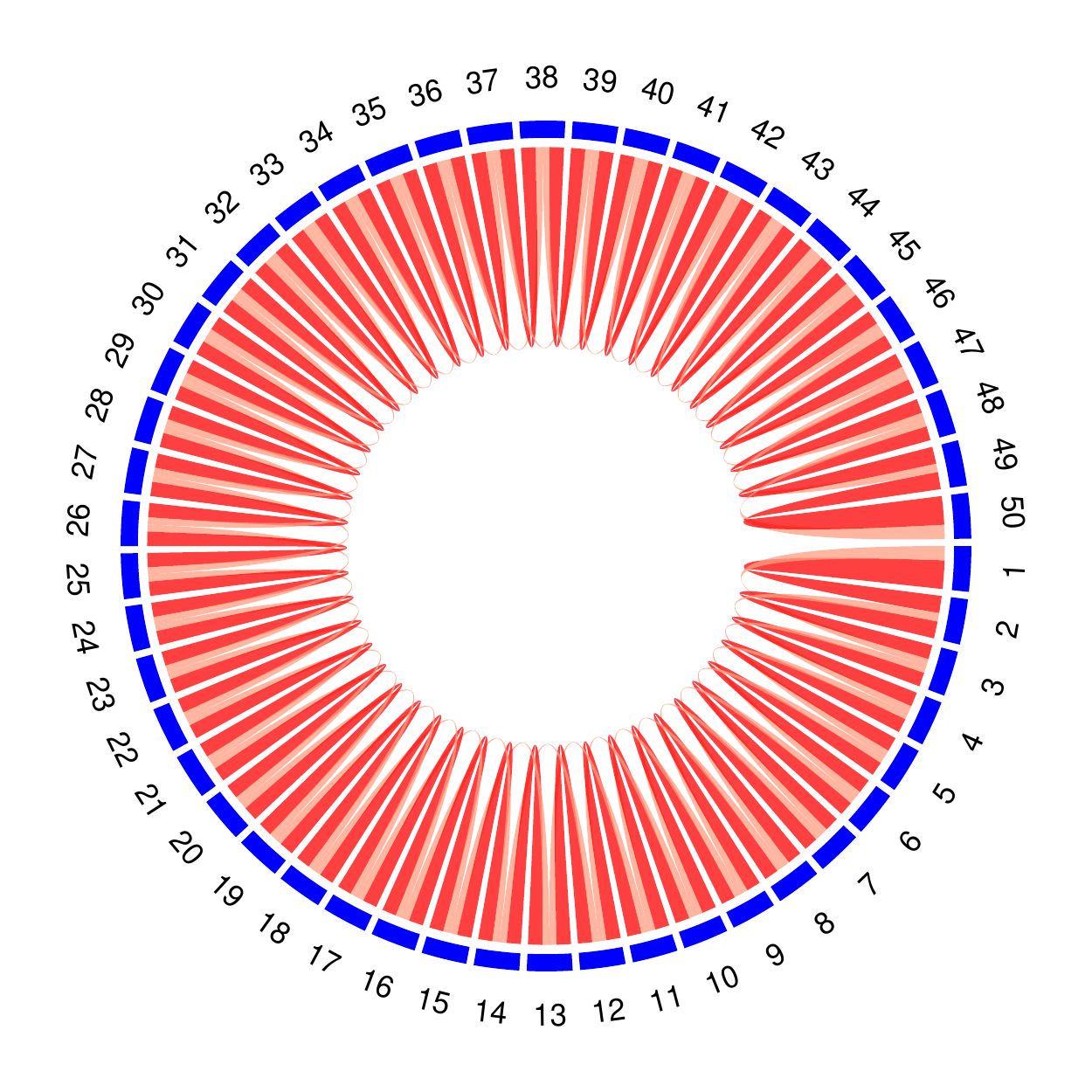}
		\includegraphics[trim={0.5cm 0.5cm 0.5cm 0.5cm}, clip, width=0.325\linewidth]{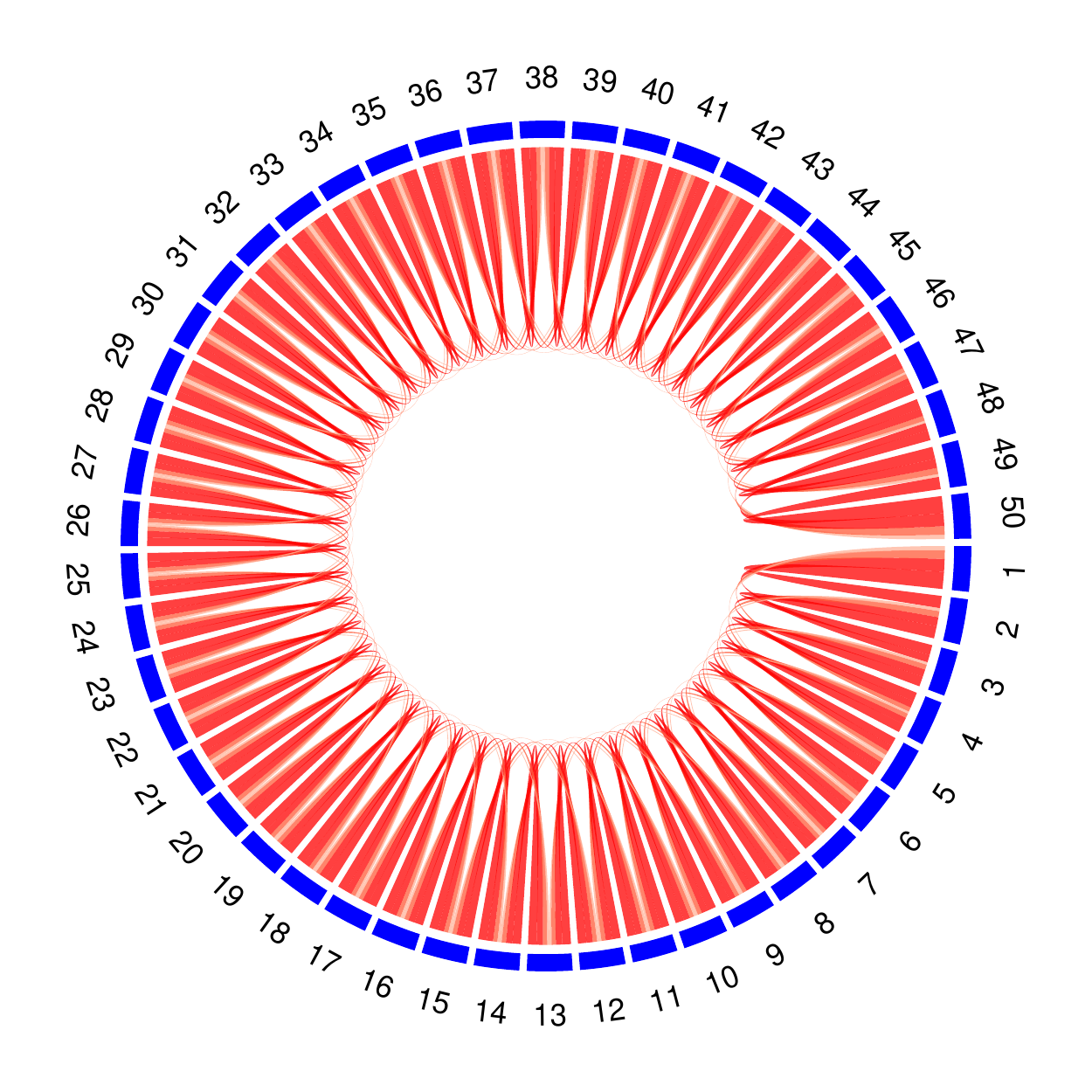} 
		\includegraphics[trim={0.5cm 0.5cm 0.5cm 0.5cm}, clip, width=0.325\linewidth]{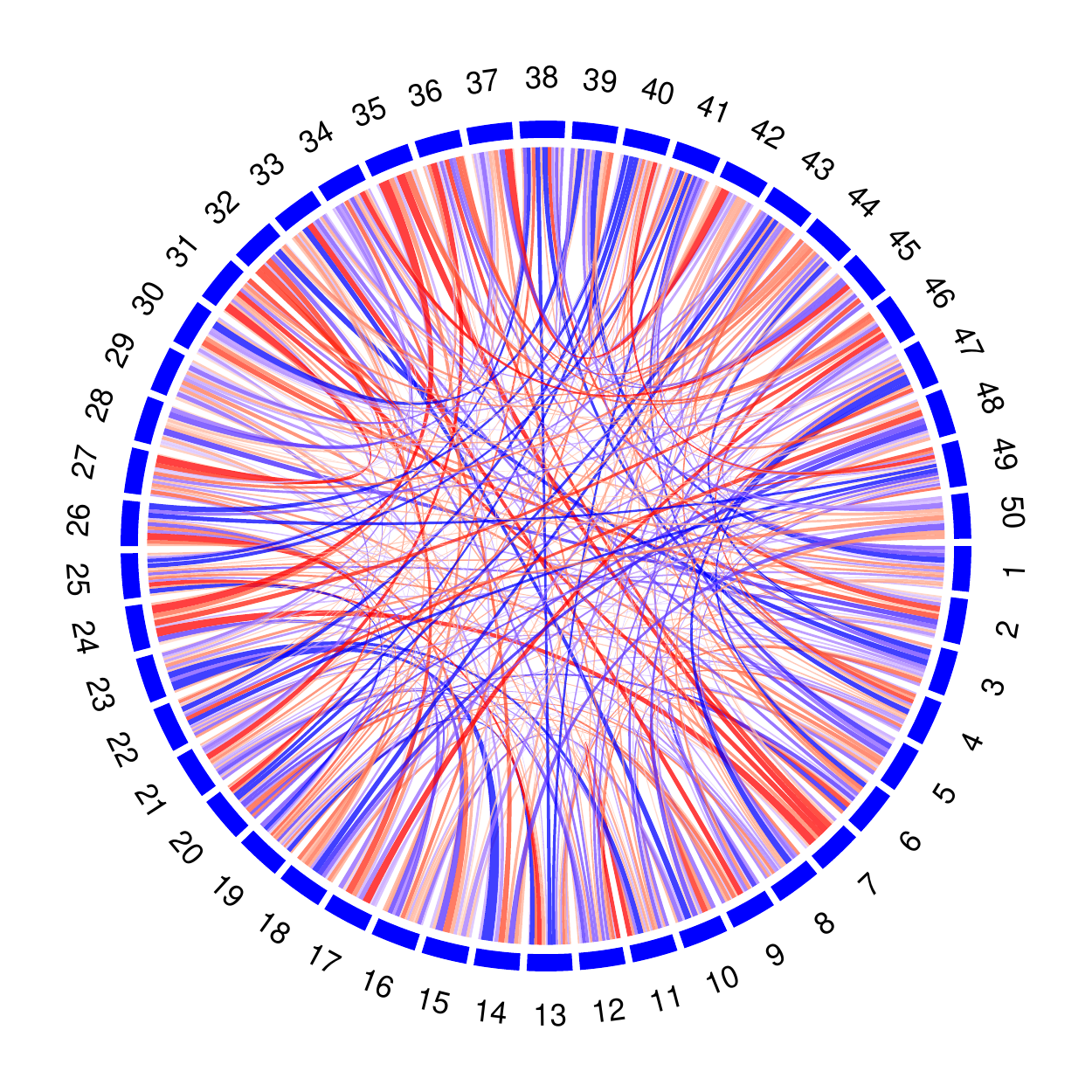} 
		\caption{True graphs.}
		\label{fig: truegraphs}
	\end{subfigure}
	\vskip 5pt
	\begin{subfigure}[b]{.99\linewidth}
		\centering
		\includegraphics[trim={0.5cm 0.5cm 0.5cm 0.5cm}, clip, width=0.325\linewidth]{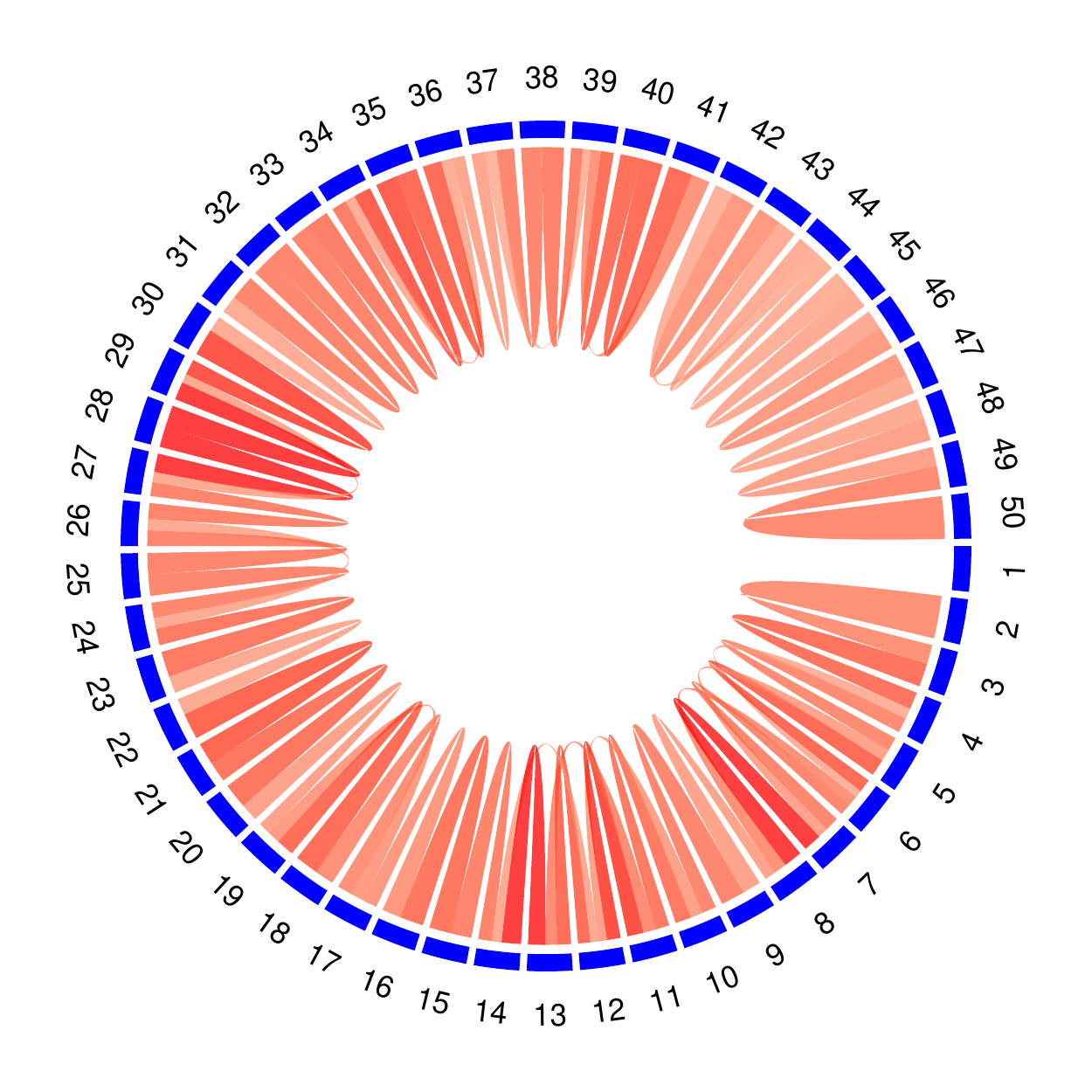} 
		\includegraphics[trim={0.5cm 0.5cm 0.5cm 0.5cm}, clip, width=0.325\linewidth]{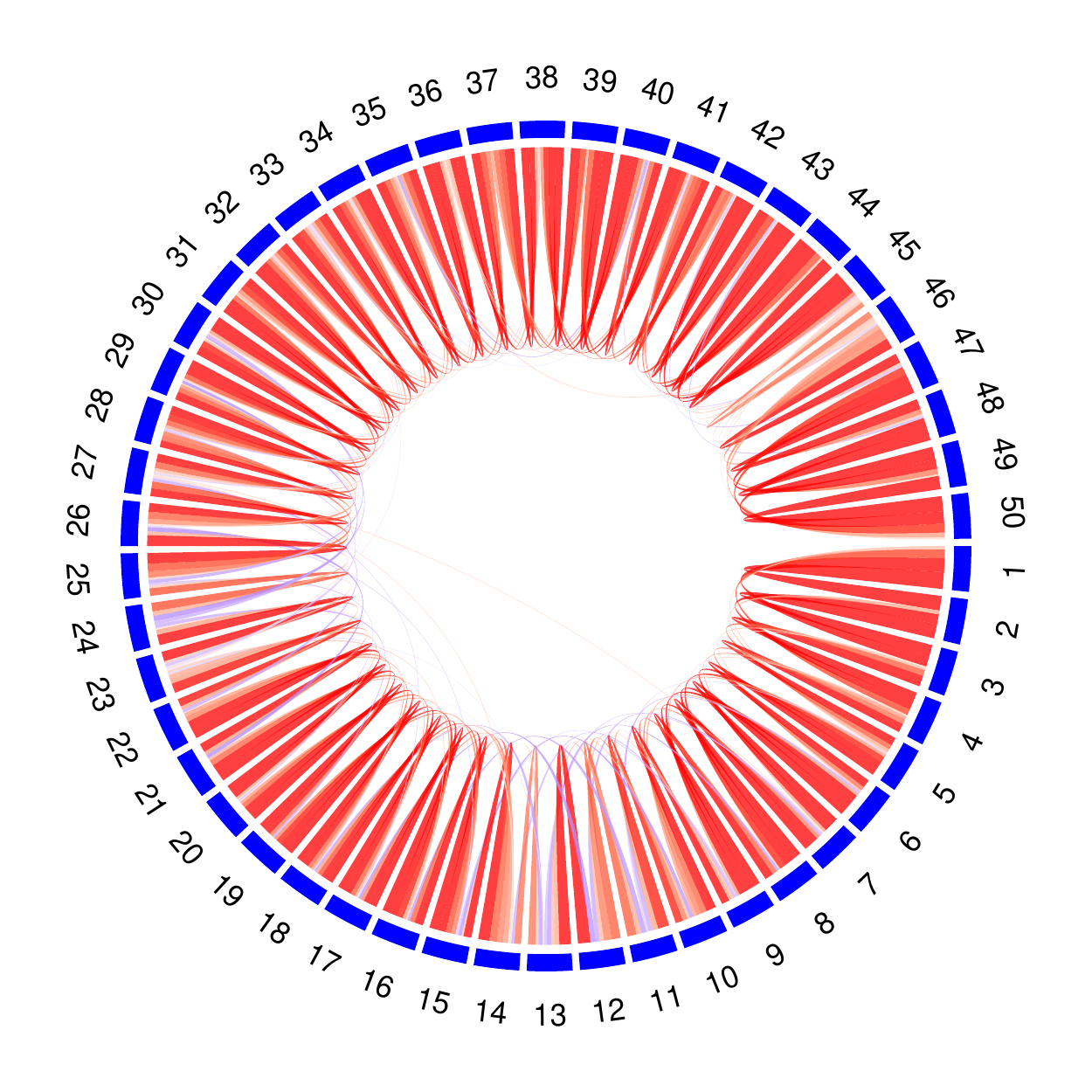} 
		\includegraphics[trim={0.5cm 0.5cm 0.5cm 0.5cm}, clip, width=0.325\linewidth]{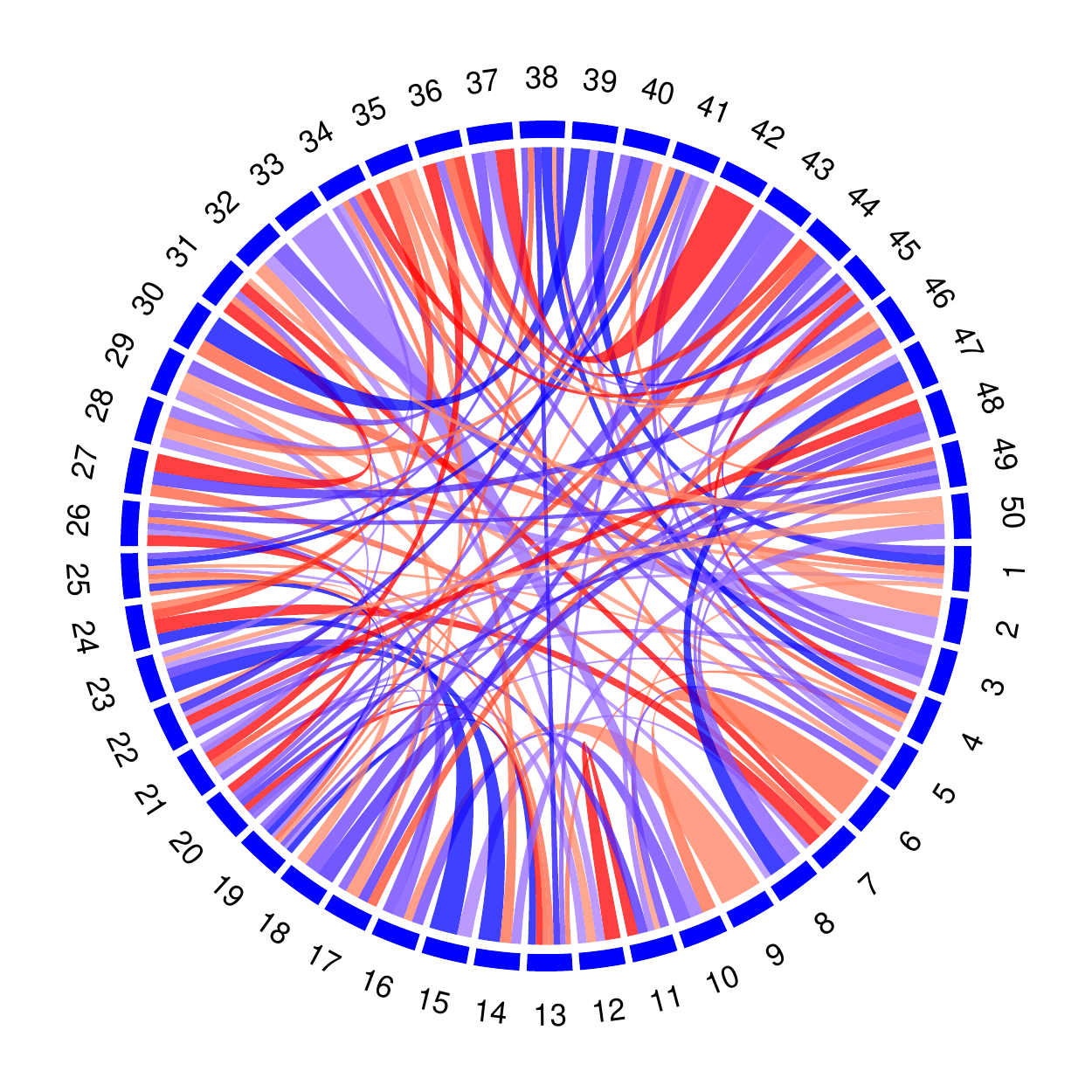} 
		\caption{Estimated graphs by the PF method.}
		\label{fig: graphs_PF}
	\end{subfigure}
	\begin{subfigure}[b]{.99\linewidth}
		\centering
		\includegraphics[trim={0.5cm 0.5cm 0.5cm 0.5cm}, clip, width=0.325\linewidth]{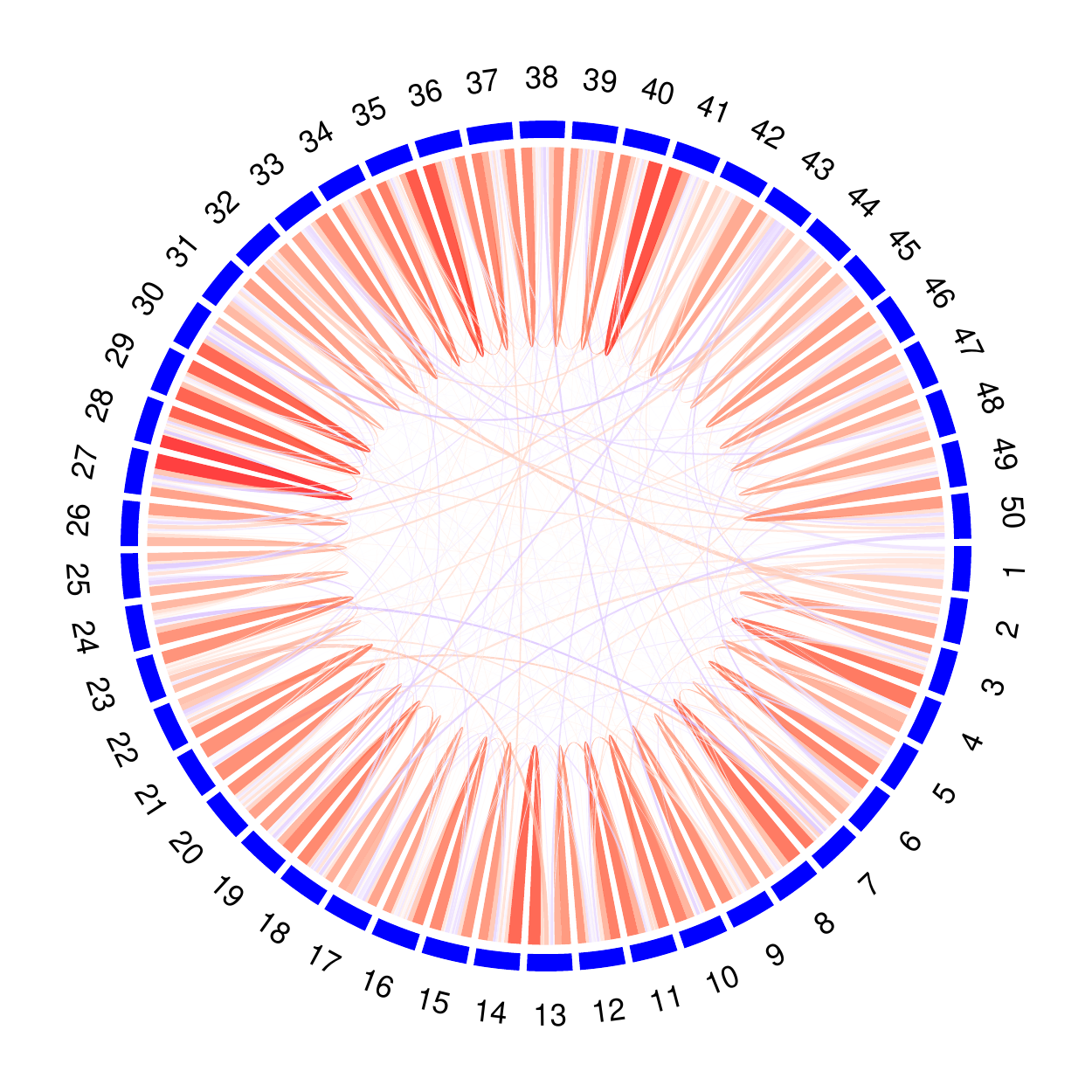} 
		\includegraphics[trim={0.5cm 0.5cm 0.5cm 0.5cm}, clip, width=0.325\linewidth]{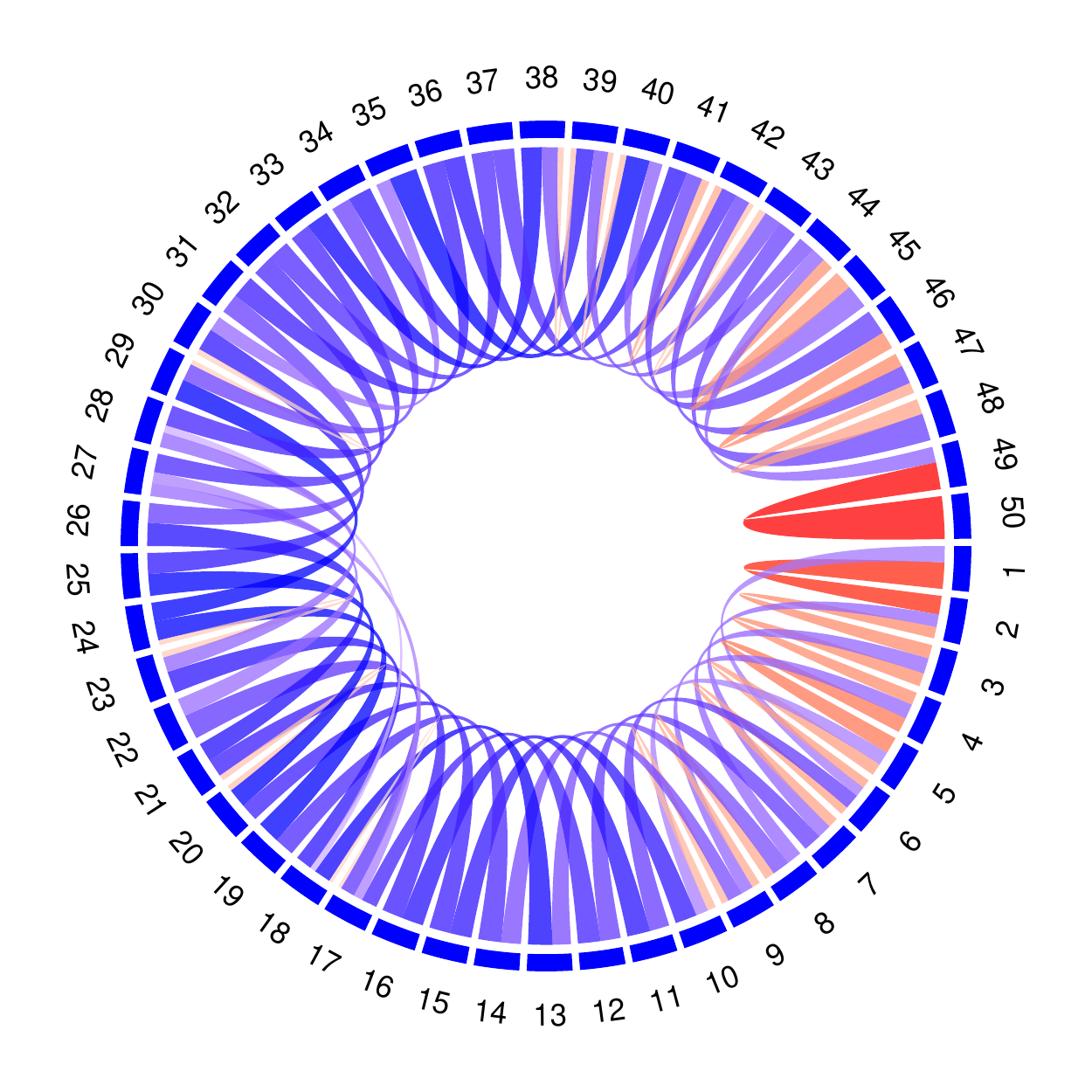} 
		\includegraphics[trim={0.5cm 0.5cm 0.5cm 0.5cm}, clip, width=0.325\linewidth]{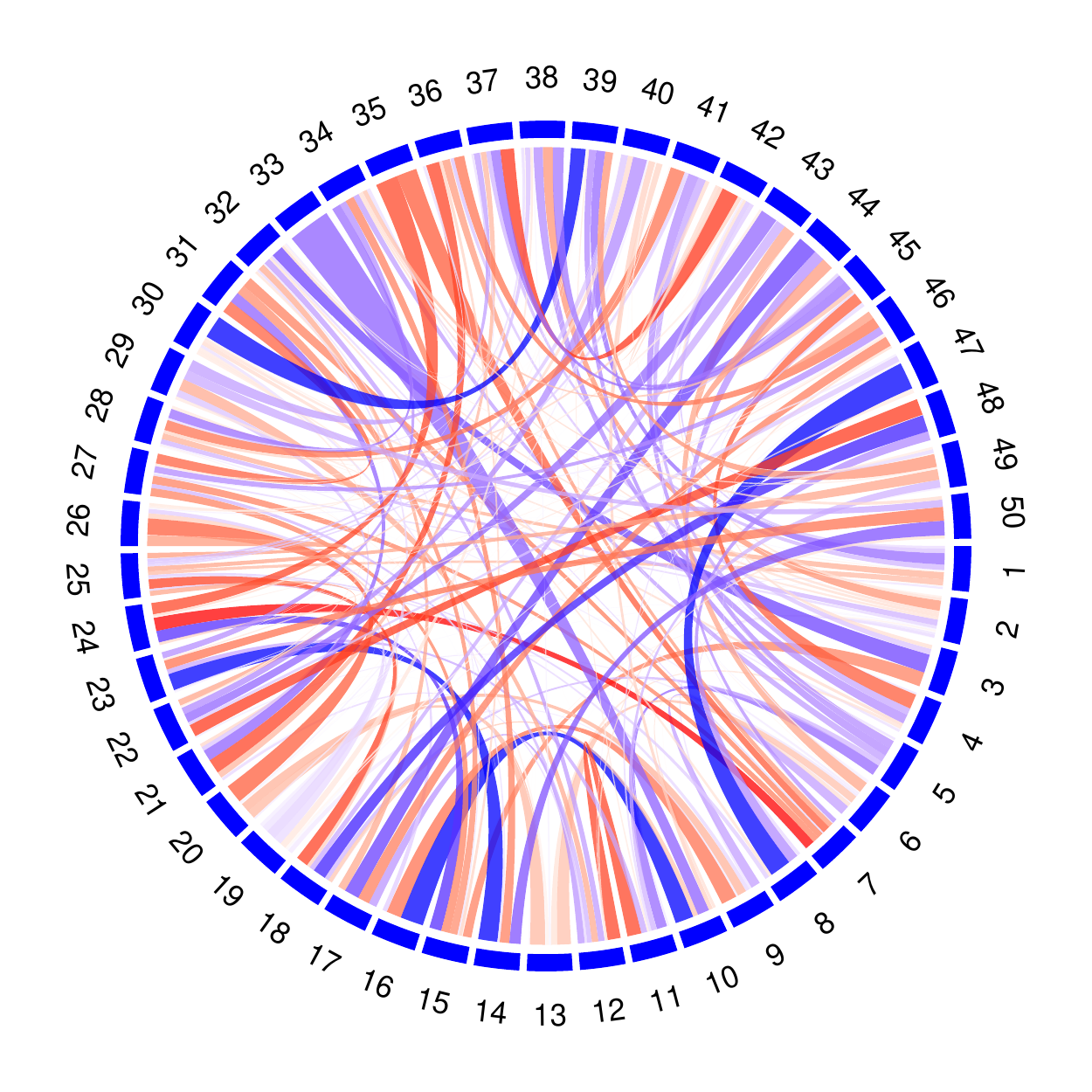} 
		\caption{Estimated graphs by Bagus.}
		\label{fig: graphs_bagus}
	\end{subfigure}	
	\begin{subfigure}[b]{.99\linewidth}
		\centering
		\includegraphics[trim={0.5cm 0.5cm 0.5cm 0.5cm}, clip, width=0.325\linewidth]{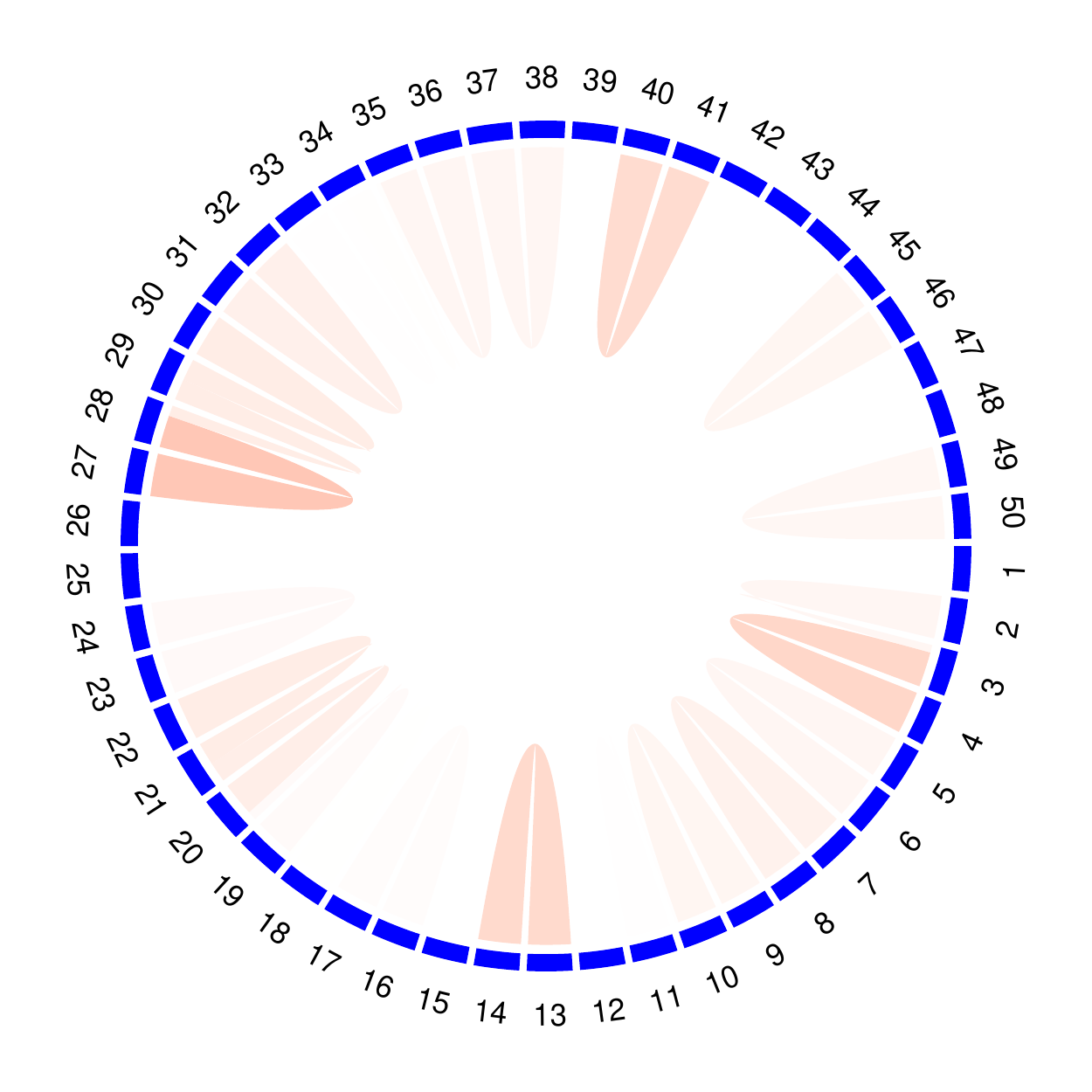} 
		\includegraphics[trim={0.5cm 0.5cm 0.5cm 0.5cm}, clip, width=0.325\linewidth]{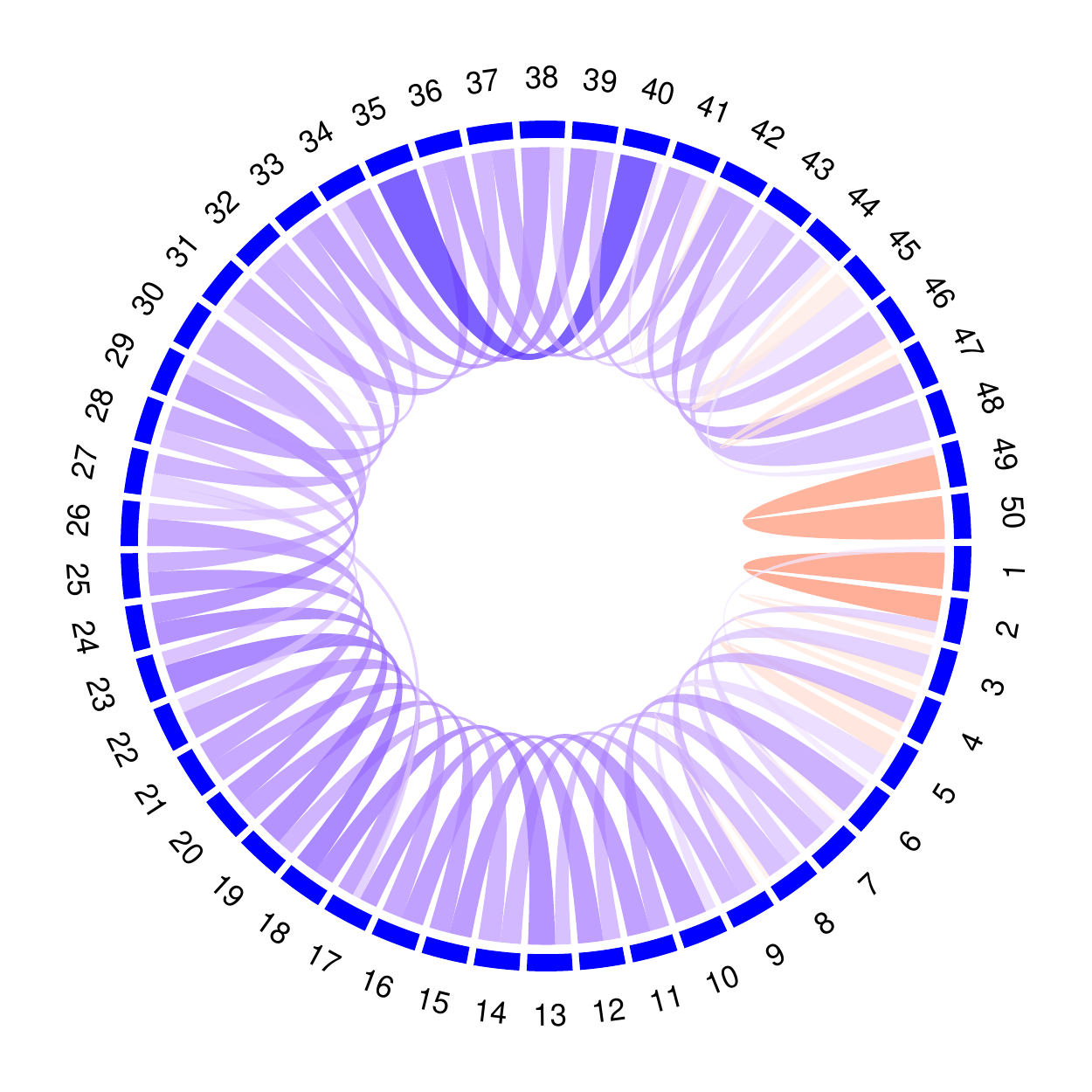} 
		\includegraphics[trim={0.5cm 0.5cm 0.5cm 0.5cm}, clip, width=0.325\linewidth]{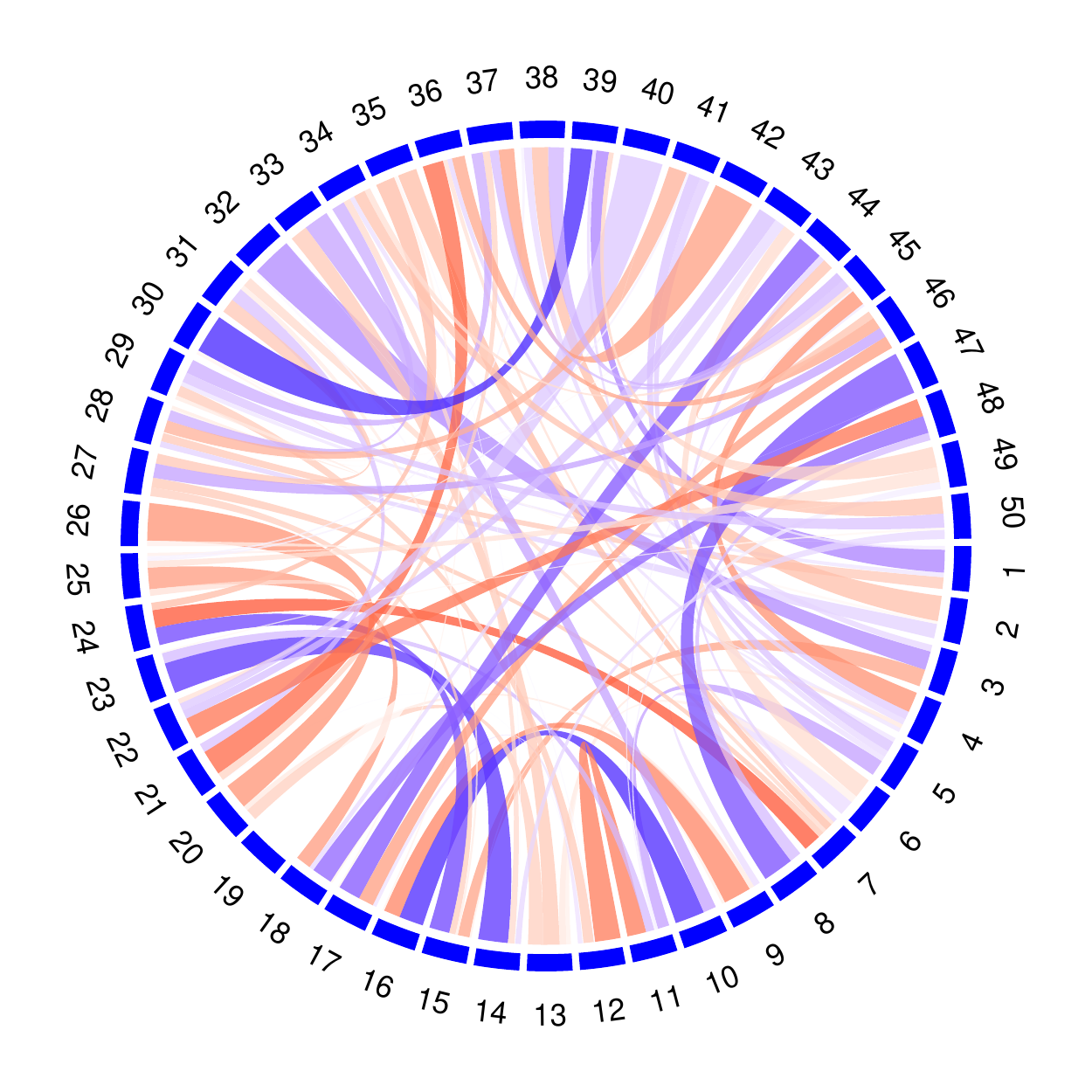} 
		\caption{Estimated graphs by M\&B.}
		\label{fig: graphs_mnb}
	\end{subfigure}	
\end{figure}
\clearpage\newpage
\begin{figure}[!ht]\ContinuedFloat
	\begin{subfigure}[b]{.99\linewidth}
		\centering
		\includegraphics[trim={0.5cm 0.5cm 0.5cm 0.5cm}, clip, width=0.325\linewidth]{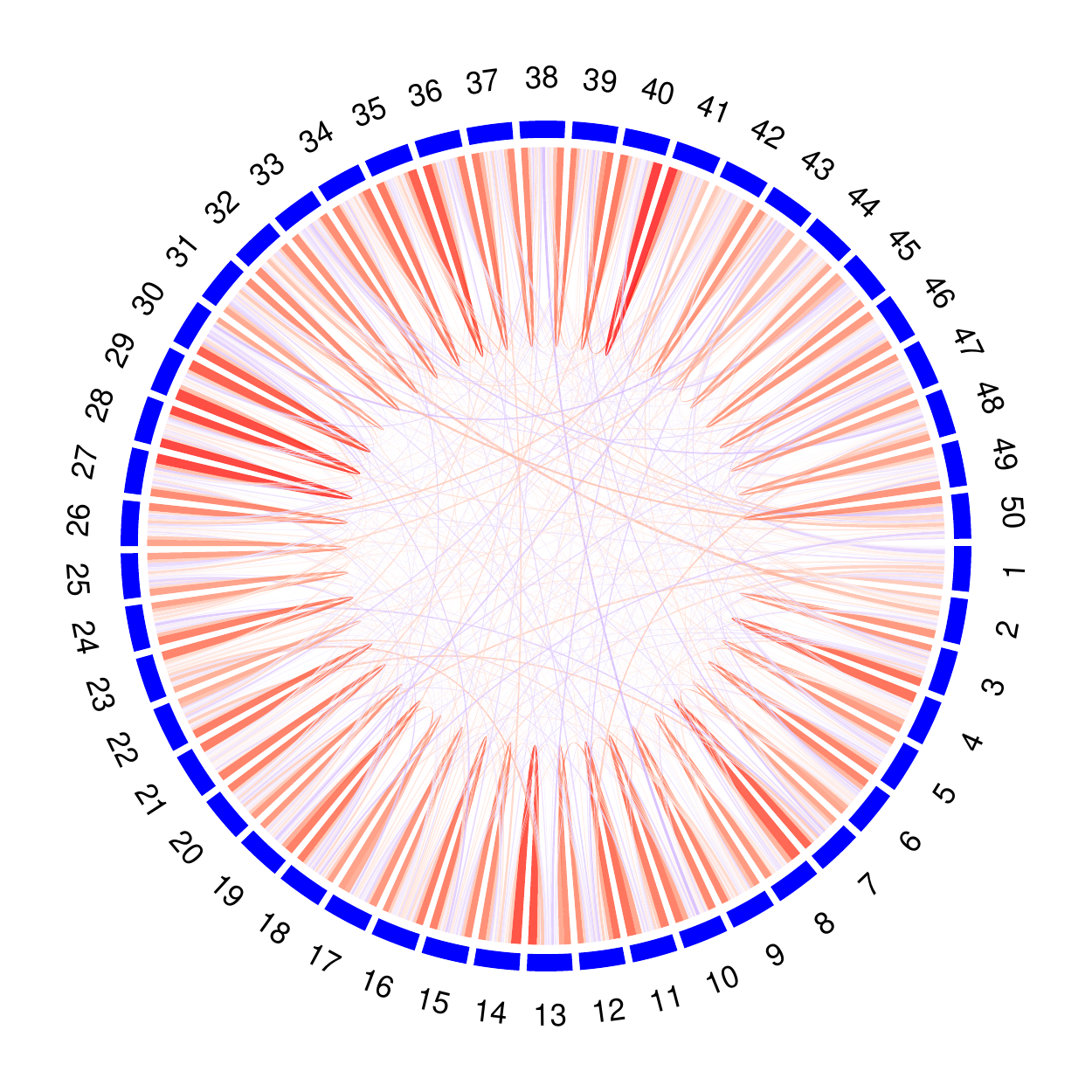} 
		\includegraphics[trim={0.5cm 0.5cm 0.5cm 0.5cm}, clip, width=0.325\linewidth]{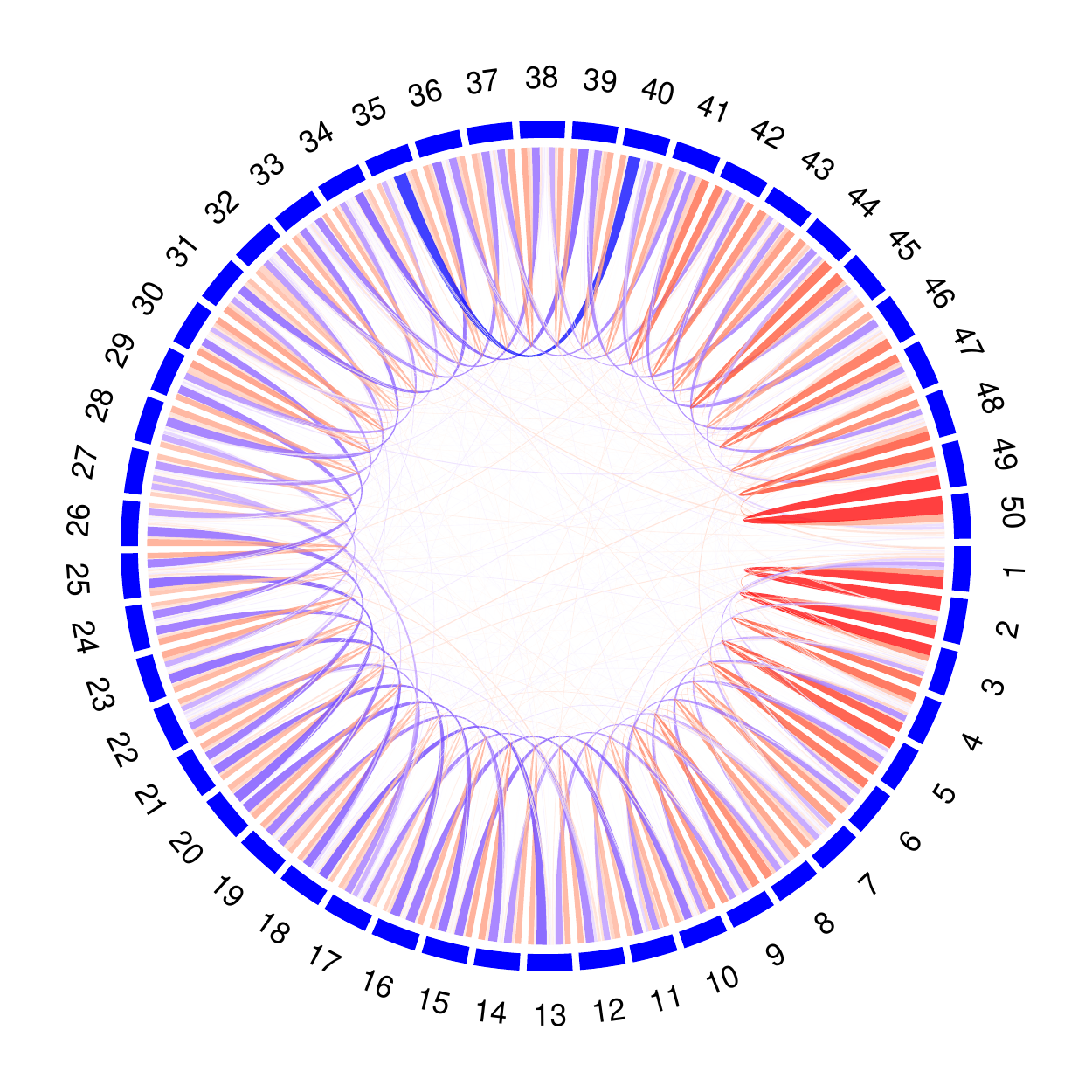} 
		\includegraphics[trim={0.5cm 0.5cm 0.5cm 0.5cm}, clip, width=0.325\linewidth]{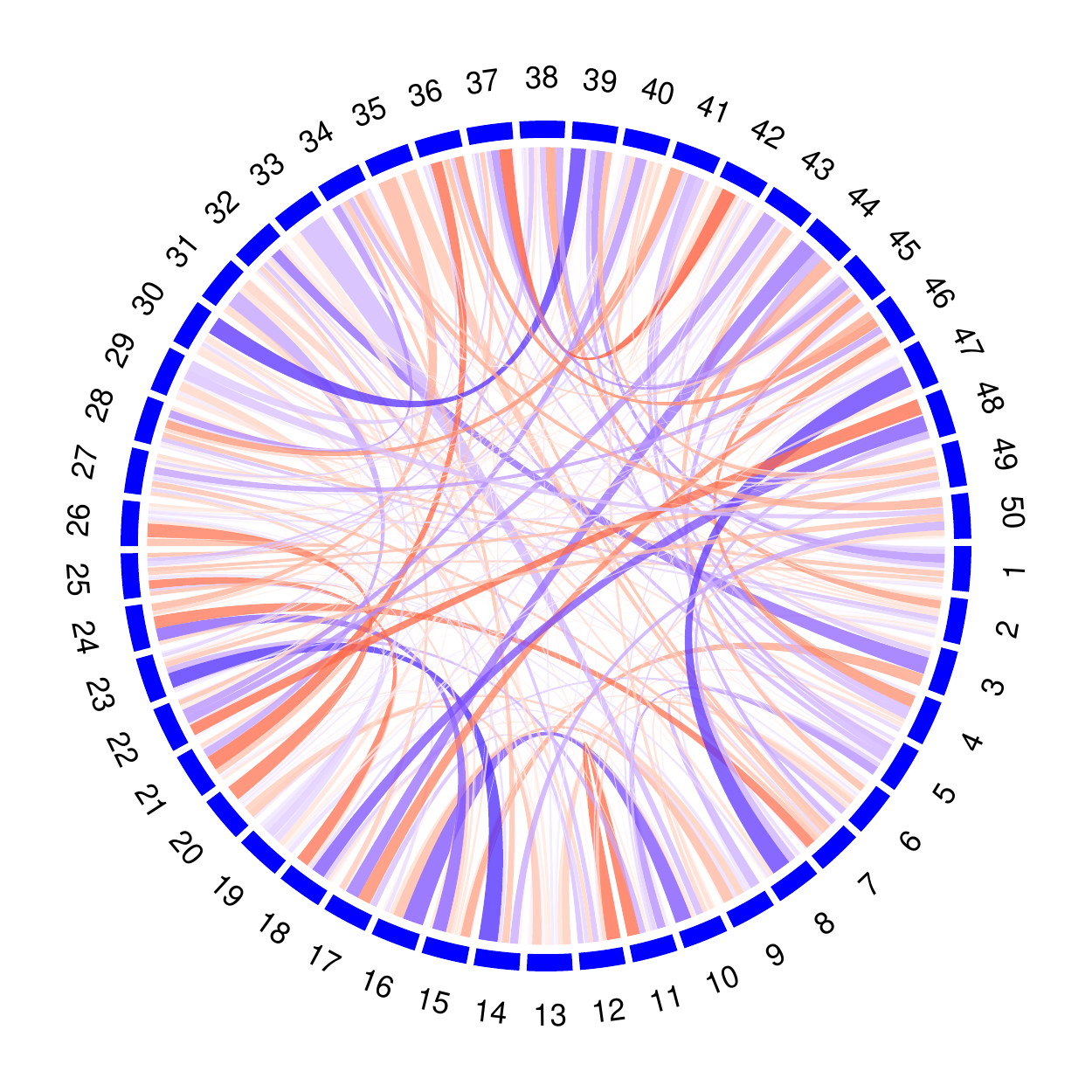} 
		\caption{Estimated graphs by Glasso.}
		\label{fig: graphs_glasso}
	\end{subfigure}
	\caption{Results of simulation experiments: Graph recovery: Panel (a) shows the true graphs for AR(2), banded and RSM structures from left to right; 
		panels (b), (c), (d) and (e) show the corresponding estimated graphs for our proposed PF and Bagus, M\&B and Glasso methods, respectively. 
		Positive (negative) associations are represented by blue (red) links, 
		their opacities being proportional to the corresponding association strengths. 
		The link widths are inversely proportional to the number of edges associated with the corresponding nodes. 
		Figure \ref{fig: sm graphs} in the supplementary materials provides a useful alternative graphical representation. 
	}
	\label{fig: graphs}
\end{figure}

\section{Application} \label{sec: applications}
We applied our proposed method to {two} real data sets from {two} different application domains, 
{namely genomics and finance}. 
To meet space constraints, the finance application is presented separately in Section \ref{subsec:finance} in the supplementary materials.

\label{subsec:gene}
Immune cells serve specialized roles in innate and adaptive bodily responses to eliminate antigens. 
To understand the cell biology of carcinogenic processes, study of immune cells and the genes therein are thus of immense importance. 

\begin{figure}[!ht]
	\centering
	\begin{subfigure}[b]{.44\linewidth}		
		\includegraphics[trim={1.1cm .8cm  1cm .9cm},clip, height=0.21\textheight]{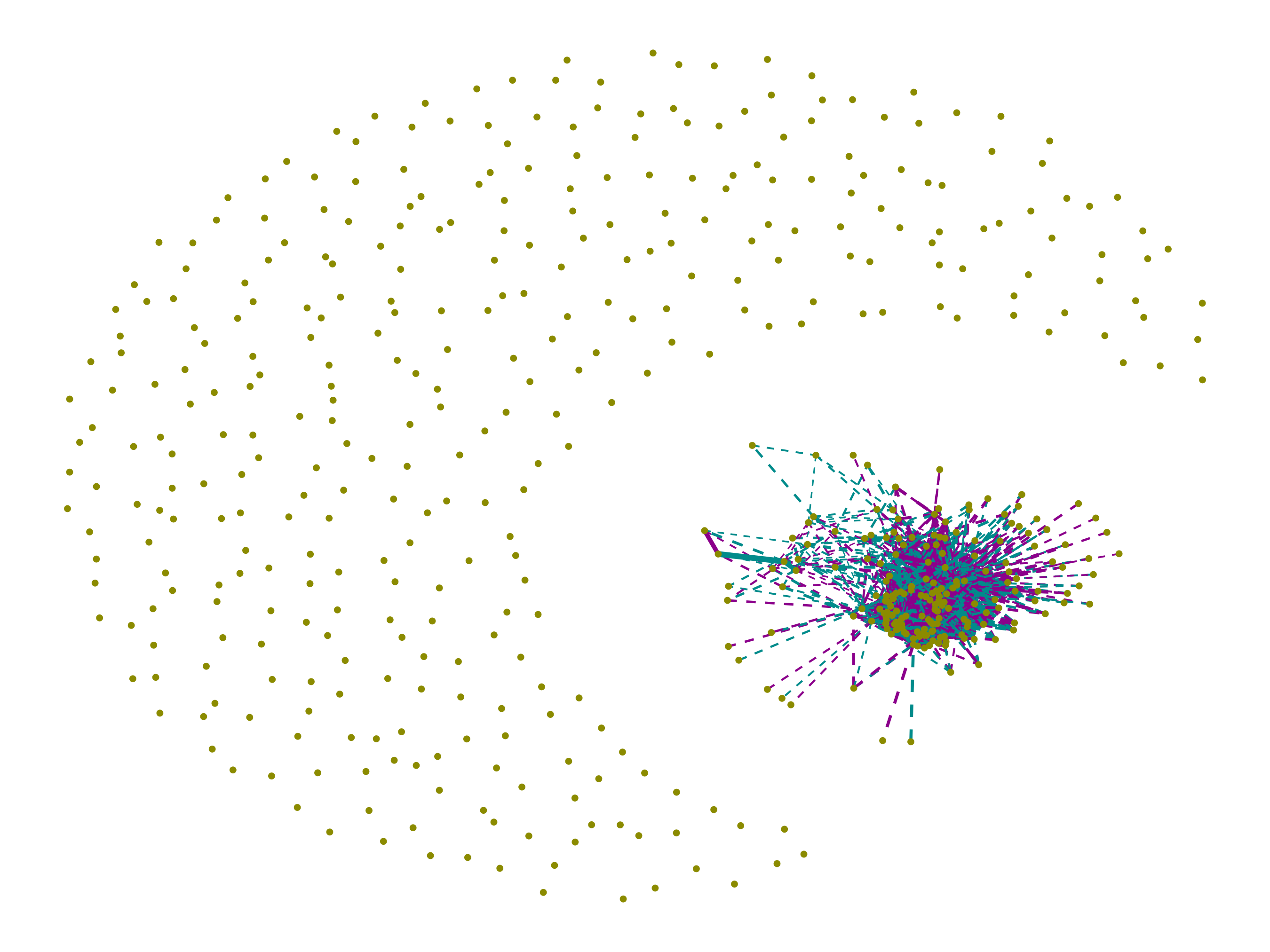}
		\caption{Graph of all genes.}
		\label{fig:marray_prec}
	\end{subfigure}
	\begin{subfigure}[b]{.44\linewidth}		
		\includegraphics[trim={.675cm .7cm  .4cm .35cm},clip, height=0.21\textheight]{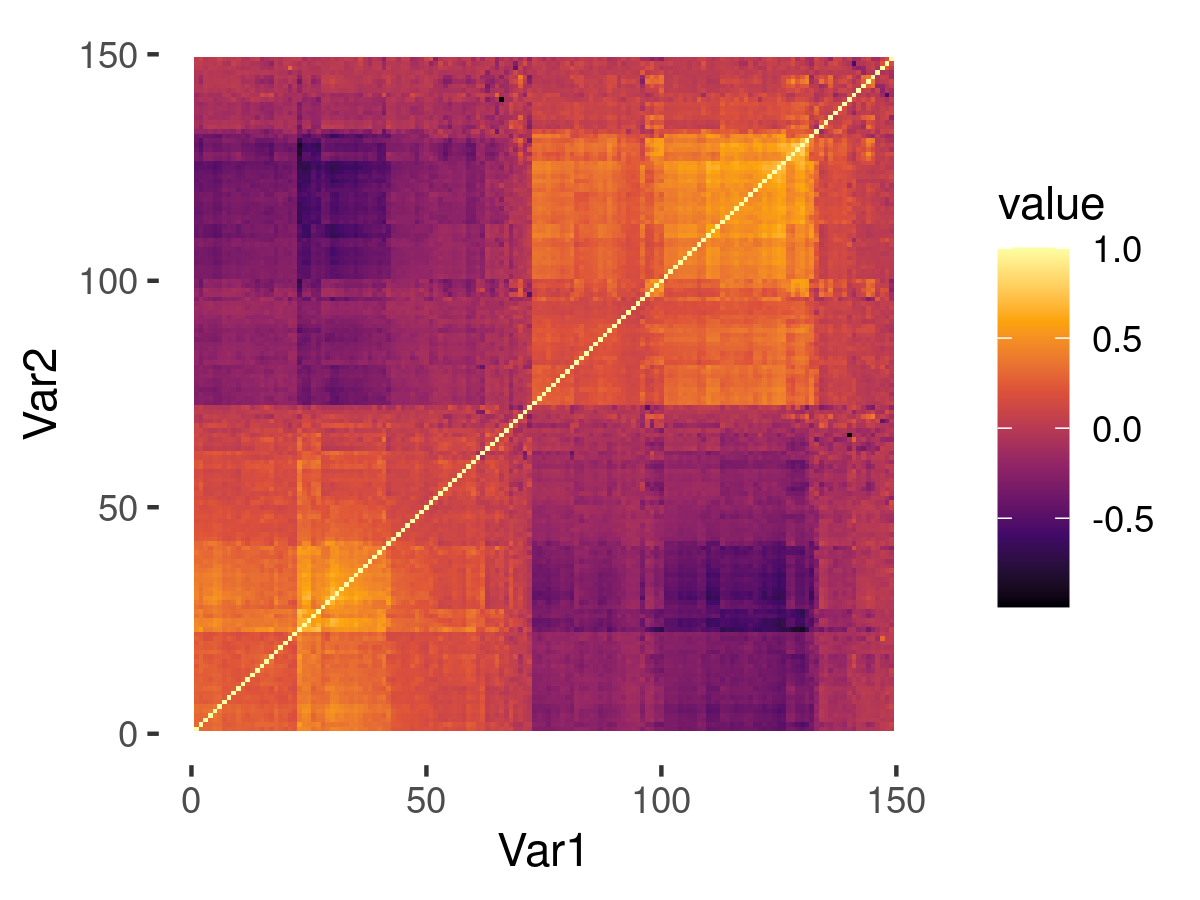}
		\caption{Heatmap of the connected genes.}
		\label{fig:marray_prec_z}
	\end{subfigure}
	\vskip 5pt
	\begin{subfigure}[b]{\linewidth}
		\centering
		\includegraphics[trim={2.4cm 1.85cm 2.4cm 2.5cm}, clip, width=0.65\linewidth]{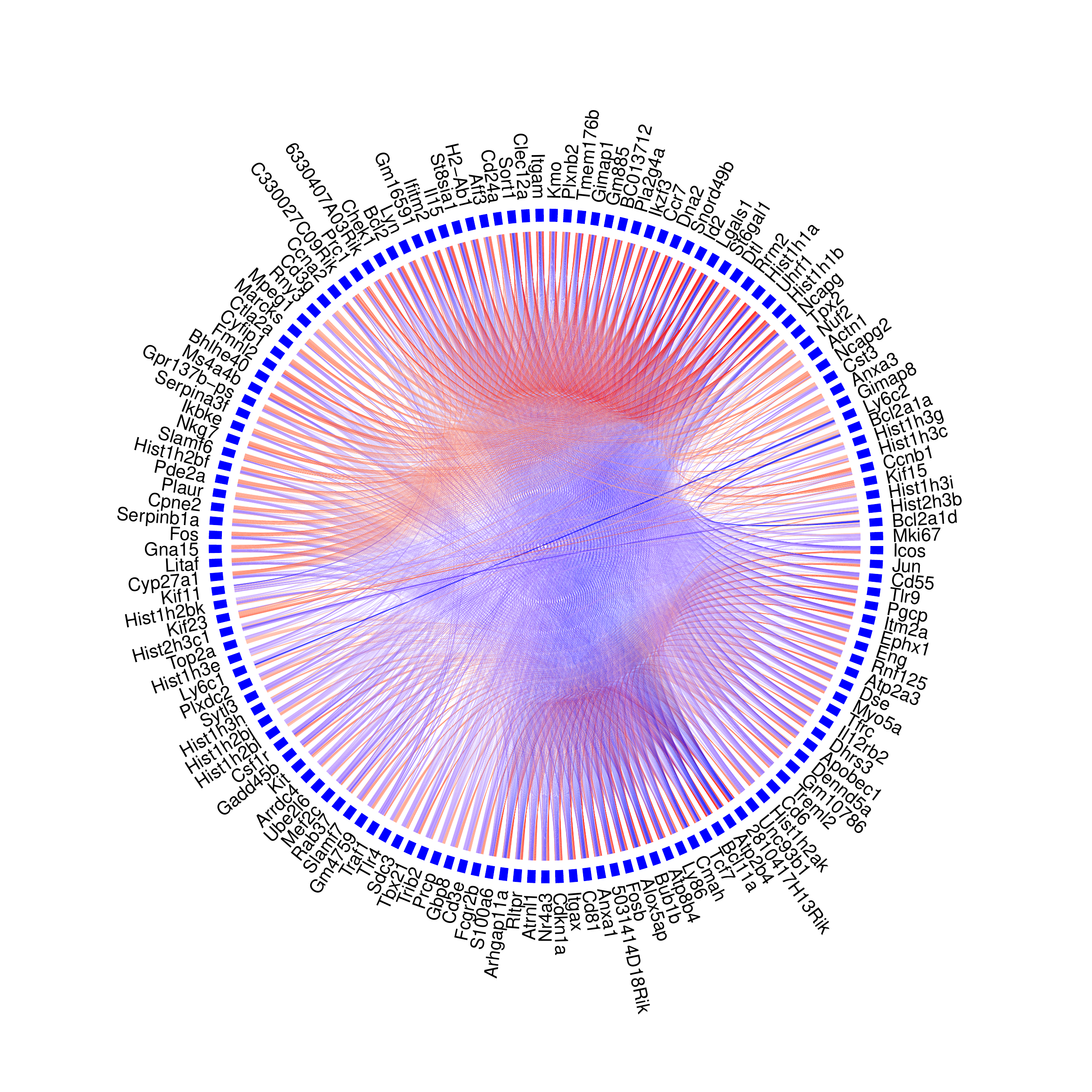}
		\caption{Graph of the connected genes.}
		\label{fig:marray_prec_zoom}
	\end{subfigure}
	\caption{Results for ImmGen microarray data: 
		Panel (a) shows the graph for all genes; 
		panel (b) shows the heatmap of the partial correlation matrix for the genes having at least 5 edges; 
		panel (c) zooms into the graph of these genes. 
		Positive (negative) associations are represented by blue (red) links, 
		their opacities being proportional to the corresponding association strengths. 
		The link widths are inversely proportional to the number of edges associated with the corresponding nodes. 
	}		
	\vskip -10pt
\end{figure}

We obtain the data from the Immunological Genome Project (ImmGen) \href{http://rstats.immgen.org/DataPage/}{data browser} \citep{immgen_paper}. 
ImmGen is a collaborative scientific research project that is currently building a gene-expression database for all characterized immune cells in mice. 
In particular, we use the GSE15907 microarray dataset \citep{micro1_train1,micro1_train2} comprising of multiple immune cell lineages 
which were isolated ex-vivo, primarily from young adult B6 male mice and double-sorted to $>99\%$ purity. 
The cell population includes all adaptive and innate lymphocytes (B, abT, gdT, Innate-Like Lymphocytes), 
myeloid cells (dendritic cells, macrophages, monocytes), mast cells and neutrophils. 
The already normalized dataset has more than $21,000$ gene expressions from $n=653$ immune cells. 
We made a $\log_2$ transformation of the data and filtered the top $2.5\%$ genes with highest variances using the \texttt{genefilter R} package \citep{genefilter}, 
resulting in a $d=544$ dimensional problem.  
Since different cell-types exhibit very different gene expression profiles, we centered the gene-expressions separately within each cell type.
The estimated graph corresponding to all genes is shown in Figure \ref{fig:marray_prec}. We can see clear evidence towards conditional independence between most of the genes, 
indicating that only a small subset of genes are functionally responsible for the variability.

For more insights, we zoom into the connected genes and plot the corresponding partial correlation matrix and graph, 
limited to the genes possessing at least 5 edges, in Figures \ref{fig:marray_prec_z} and \ref{fig:marray_prec_zoom} respectively.
{We notice overall positive partial correlations between the histone class of genes such as Hist2h3b, Hist1h3h, Hist2h3c1, Hist1h3c, etc \citep{histone}.
The positive correlations between these protein coding genes indicate that these genes are generally expressed together in the mechanism.}
%
{Bcl2a1a and Bcl2a1d, two functional isoforms of the B cell leukemia 2 family member A1 exerting important pro-survival functions, show strong positive association.}
We also observe strong positive association between Ly6c1 and Ly6c2 genes. 
\citet{ly6c} noted these genes to be located adjacent to each other in the mid‐section of the Ly6 complex. 
They share $>95\%$ similarity in their genomic and protein sequences and these two genes has been considered synonymous to each other and hence exhibit strong positive association. 
{Although positive correlations are often observed between the membrane-spanning 4A class of genes such as Ms4a4c, Ms4a4b, Ms4a6b, etc. \citep{ms4a}, we find the genes to be conditionally almost independent in our analysis.
It indicates that their expressions might be regulated by other genes.}

From Figure \ref{fig:marray_prec_z}, it thus seems that 
there exist two blocks where within each block the genes are positively associated whereas between the blocks the genes are negatively associated. 
This might again be an indication that either of these blocks of genes are generally expressed together in immune cells.

\section{Discussion} \label{sec: discussion}

In this article, we proposed a novel flexible statistical model for Gaussian precision matrices  
that relies on decomposing them 
into a low-rank and a diagonal component. 
The decomposition is theoretically and practically highly flexible and is thus often used in covariance matrix models, 
arising naturally in their computationally tractable latent factor based representations. 
The approach has, however, not been popular in precision matrix and related graph estimation problems as it poses daunting computational challenges when applied to such settings. 
We addressed this issue in this article by exploring a previously under-utilized latent variable construction, leading to a highly scalable Gibbs sampler that allows efficient posterior inference. 
The decomposition based strategy also allowed us to use sparsity inducing priors to shrink insignificant off-diagonal entries toward zero 
while also making the approach adaptable to high-dimensional sparse settings. 
We specifically adapted the Dirichlet-Laplace prior for sparse precision matrix estimation. 
We developed a novel posterior FDR control based method to perform graph selection that properly accommodates posterior uncertainty. 
We also established theoretical convergence guarantees for the proposed model in high-dimensional sparse settings. 
In synthetic experiments, the proposed method vastly outperformed its competitors in receiving the true underlying graphs. 
We illustrated the method's practical utility through real data examples from genomics and finance.

Aside from providing fundamentally new probabilistic perspectives on
Gaussian precision matrix and related graph estimation problems, 
our work also bridges the gap between frequentist penalized likelihood based strategies 
and Bayesian shrinkage prior based ideas for such models. 
The simple and highly scalable computational algorithms resulting from our latent factor representation 
should also free up Bayesians from having to put restrictive assumptions on the graph structures to achieve computational tractability, 
opening up new opportunities to adapt the basic models to more complex and more realistic data structures and study designs. 
A few such methodological extensions we are pursuing as topics of separate ongoing research include 
dynamic Gaussian graphical models \citep{huang2017learning}, 
covariate dependent Gaussian graphical models, 
nonparanormal \citep{liu2009nonparanormal} 
and Gaussian copula graphical models \citep{pitt2006efficient}, 
etc.

\section*{Supplementary Materials}
Supplementary materials discuss 
a general strategy for posterior computation under broad classes of generic priors, 
proofs of the theoretical results, 
some additional figures,
and an application to a NASDAQ-100 stock price dataset. 


\vspace*{-15pt}

\baselineskip=14pt
\bibliographystyle{natbib}
\bibliography{Graphical_Models}

\begin{thebibliography}{}

\bibitem[Armstrong {\em et~al.}(2009)Armstrong, Carter, Wong, and
  Kohn]{armstrong2009bayesian}
Armstrong, H., Carter, C.~K., Wong, K. F.~K., and Kohn, R. (2009).
\newblock Bayesian covariance matrix estimation using a mixture of decomposable
  graphical models.
\newblock {\em Statistics and Computing\/}, {\bf 19}, 303--316.

\bibitem[Atay-Kayis and Massam(2005)Atay-Kayis and Massam]{atay2005monte}
Atay-Kayis, A. and Massam, H. (2005).
\newblock A {M}onte {C}arlo method for computing the marginal likelihood in
  nondecomposable {G}aussian graphical models.
\newblock {\em Biometrika\/}, {\bf 92}, 317--335.

\bibitem[Baglama and Reichel(2005)Baglama and Reichel]{irlba}
Baglama, J. and Reichel, L. (2005).
\newblock Augmented implicitly restarted {L}anczos bidiagonalization methods.
\newblock {\em SIAM Journal on Scientific Computing\/}, {\bf 27}, 19--42.

\bibitem[Bai and Ng(2008)Bai and Ng]{bai2008factorest}
Bai, J. and Ng, S. (2008).
\newblock Large dimensional factor analysis.
\newblock {\em Foundations and Trends in Econometrics\/}, {\bf 3}, 89--163.

\bibitem[Banerjee {\em et~al.}(2008)Banerjee, {El Ghaoui}, and
  {d'Aspremont}]{banerjee2008}
Banerjee, O., {El Ghaoui}, L., and {d'Aspremont}, A. (2008).
\newblock Model selection through sparse maximum likelihood estimation for
  multivariate {G}aussian or binary data.
\newblock {\em JMLR\/}, {\bf 9}, 485--516.

\bibitem[Banerjee and Ghosal(2015)Banerjee and Ghosal]{banerjee2015bayesian}
Banerjee, S. and Ghosal, S. (2015).
\newblock Bayesian structure learning in graphical models.
\newblock {\em Journal of Multivariate Analysis\/}, {\bf 136}, 147--162.

\bibitem[Barvinok and Rudelson(2022)Barvinok and
  Rudelson]{BARVINOK2022quad_system}
Barvinok, A. and Rudelson, M. (2022).
\newblock When a system of real quadratic equations has a solution.
\newblock {\em Advances in Mathematics\/}, {\bf 403}, 108391.

\bibitem[Berger(1985)Berger]{berger_book}
Berger, J.~O. (1985).
\newblock {\em {Statistical decision theory and Bayesian analysis}\/}.
\newblock Springer series in statistics. Springer-Verlag, New York, 2nd
  edition.

\bibitem[Bhattacharya and Dunson(2011)Bhattacharya and
  Dunson]{bhattacharya2011sparse}
Bhattacharya, A. and Dunson, D.~B. (2011).
\newblock Sparse {B}ayesian infinite factor models.
\newblock {\em Biometrika\/}, {\bf 98}, 291--306.

\bibitem[Bhattacharya {\em et~al.}(2015)Bhattacharya, Pati, Pillai, and
  Dunson]{bhattacharya2015dirichlet}
Bhattacharya, A., Pati, D., Pillai, N.~S., and Dunson, D.~B. (2015).
\newblock Dirichlet-{L}aplace priors for optimal shrinkage.
\newblock {\em Journal of the American Statistical Association\/}, {\bf 110},
  1479--1490.

\bibitem[Bhattacharya {\em et~al.}(2016)Bhattacharya, Chakraborty, and
  Mallick]{bhattacharya2016fast}
Bhattacharya, A., Chakraborty, A., and Mallick, B.~K. (2016).
\newblock Fast sampling with {G}aussian scale mixture priors in
  high-dimensional regression.
\newblock {\em Biometrika\/}, {\bf 103}, 985--991.

\bibitem[Bien and Tibshirani(2011)Bien and Tibshirani]{bien2011covest}
Bien, J. and Tibshirani, R.~J. (2011).
\newblock Sparse estimation of a covariance matrix.
\newblock {\em Biometrika\/}, {\bf 98}, 807--820.

\bibitem[Carvalho and Scott(2009)Carvalho and Scott]{carvalho2009objective}
Carvalho, C.~M. and Scott, J.~G. (2009).
\newblock Objective {B}ayesian model selection in {G}aussian graphical models.
\newblock {\em Biometrika\/}, {\bf 96}, 497--512.

\bibitem[Carvalho {\em et~al.}(2007)Carvalho, Massam, and
  West]{carvalho2007simulation}
Carvalho, C.~M., Massam, H., and West, M. (2007).
\newblock Simulation of hyper-inverse {W}ishart distributions in graphical
  models.
\newblock {\em Biometrika\/}, {\bf 94}, 647--659.

\bibitem[Carvalho {\em et~al.}(2009)Carvalho, Polson, and
  Scott]{carvalho2009handling}
Carvalho, C.~M., Polson, N.~G., and Scott, J.~G. (2009).
\newblock Handling sparsity via the horseshoe.
\newblock In {\em Artificial Intelligence and Statistics\/}, pages 73--80.
  PMLR.

\bibitem[Chandra and Bhattacharya(2019)Chandra and
  Bhattacharya]{chandra2019non}
Chandra, N.~K. and Bhattacharya, S. (2019).
\newblock Non-marginal decisions: A novel {B}ayesian multiple testing
  procedure.
\newblock {\em Electronic Journal of Statistics\/}, {\bf 13}, 489--535.

\bibitem[Dallakyan and Pourahmadi(2020)Dallakyan and
  Pourahmadi]{dallakyan2020fused}
Dallakyan, A. and Pourahmadi, M. (2020).
\newblock Fused-lasso regularized {C}holesky factors of large nonstationary
  covariance matrices of longitudinal data.
\newblock {\em arXiv:2007.11168\/}.

\bibitem[Daniele {\em et~al.}(2019)Daniele, Pohlmeier, and
  Zagidullina]{daniele2019penalizingfactor}
Daniele, M., Pohlmeier, W., and Zagidullina, A. (2019).
\newblock Sparse approximate factor estimation for high-dimensional covariance
  matrices.
\newblock {\em arXiv:1906.05545\/}.

\bibitem[d'Aspremont {\em et~al.}(2008)d'Aspremont, Banerjee, and {El
  Ghaoui}]{d2008first}
d'Aspremont, A., Banerjee, O., and {El Ghaoui}, L. (2008).
\newblock First-order methods for sparse covariance selection.
\newblock {\em SIAM Journal on Matrix Analysis and Applications\/}, {\bf 30},
  56--66.

\bibitem[Dawid and Lauritzen(1993)Dawid and Lauritzen]{dawid1993hyper}
Dawid, A.~P. and Lauritzen, S.~L. (1993).
\newblock Hyper {M}arkov laws in the statistical analysis of decomposable
  graphical models.
\newblock {\em The Annals of Statistics\/}, {\bf 21}, 1272--1317.

\bibitem[Dellaportas {\em et~al.}(2003)Dellaportas, Giudici, and
  Roberts]{dellaportas2003bayesian}
Dellaportas, P., Giudici, P., and Roberts, G. (2003).
\newblock Bayesian inference for nondecomposable graphical {G}aussian models.
\newblock {\em Sankhy{\=a}: The Indian Journal of Statistics\/}, pages 43--55.

\bibitem[Desch {\em et~al.}(2011)Desch, Randolph, {\em et~al.}]{micro1_train2}
Desch, A.~N., Randolph, G.~J., {\em et~al.} (2011).
\newblock {CD103+ pulmonary dendritic cells preferentially acquire and present
  apoptotic cell--associated antigen}.
\newblock {\em Journal of Experimental Medicine\/}, {\bf 208}, 1789--1797.

\bibitem[Deshpande {\em et~al.}(2019)Deshpande, Ro{\v{c}}kov{\'a}, and
  George]{deshpande2019simultaneous}
Deshpande, S.~K., Ro{\v{c}}kov{\'a}, V., and George, E.~I. (2019).
\newblock Simultaneous variable and covariance selection with the multivariate
  spike-and-slab lasso.
\newblock {\em Journal of Computational and Graphical Statistics\/}, {\bf 28},
  921--931.

\bibitem[Dobra {\em et~al.}(2004)Dobra, Hans, Jones, Nevins, Yao, and
  West]{dobra2004sparse}
Dobra, A., Hans, C., Jones, B., Nevins, J.~R., Yao, G., and West, M. (2004).
\newblock Sparse graphical models for exploring gene expression data.
\newblock {\em Journal of Multivariate Analysis\/}, {\bf 90}, 196--212.

\bibitem[Dobra {\em et~al.}(2011)Dobra, Lenkoski, and
  Rodriguez]{dobra2011bayesian}
Dobra, A., Lenkoski, A., and Rodriguez, A. (2011).
\newblock Bayesian inference for general {G}aussian graphical models with
  application to multivariate lattice data.
\newblock {\em Journal of the American Statistical Association\/}, {\bf 106},
  1418--1433.

\bibitem[Eddelbuettel and Francois(2011)Eddelbuettel and Francois]{rcpp}
Eddelbuettel, D. and Francois, R. (2011).
\newblock {Rcpp: Seamless R and C++ integration}.
\newblock {\em Journal of Statistical Software\/}, {\bf 40}, 1--18.

\bibitem[Escobar and West(1995)Escobar and West]{escobar1995bayesian}
Escobar, M.~D. and West, M. (1995).
\newblock Bayesian density estimation and inference using mixtures.
\newblock {\em Journal of the American Statistical Association\/}, {\bf 90},
  577--588.

\bibitem[Fan {\em et~al.}(2011)Fan, Liao, and
  Mincheva]{fan2011penalizingfactor}
Fan, J., Liao, Y., and Mincheva, M. (2011).
\newblock {High-dimensional covariance matrix estimation in approximate factor
  models}.
\newblock {\em The Annals of Statistics\/}, {\bf 39}, 3320--3356.

\bibitem[Fan {\em et~al.}(2018)Fan, Liu, and Wang]{fan2018penalizingfactor}
Fan, J., Liu, H., and Wang, W. (2018).
\newblock {Large covariance estimation through elliptical factor models}.
\newblock {\em The Annals of Statistics\/}, {\bf 46}, 1383--1414.

\bibitem[Ferguson(1973)Ferguson]{ferguson73}
Ferguson, T.~S. (1973).
\newblock A {B}ayesian analysis of some nonparametric problems.
\newblock {\em The Annals of Statistics\/}, {\bf 1}, 209--230.

\bibitem[Friedman {\em et~al.}(2008)Friedman, Hastie, and
  Tibshirani]{friedman2008sparse}
Friedman, J., Hastie, T., and Tibshirani, R. (2008).
\newblock Sparse inverse covariance estimation with the graphical lasso.
\newblock {\em Biostatistics\/}, {\bf 9}, 432--441.

\bibitem[Gan {\em et~al.}(2019)Gan, Narisetty, and Liang]{gan2019bayesian}
Gan, L., Narisetty, N.~N., and Liang, F. (2019).
\newblock Bayesian regularization for graphical models with unequal shrinkage.
\newblock {\em Journal of the American Statistical Association\/}, {\bf 114},
  1218--1231.

\bibitem[Gentleman {\em et~al.}(2020)Gentleman, Carey, Huber, and
  Hahne]{genefilter}
Gentleman, R., Carey, V., Huber, W., and Hahne, F. (2020).
\newblock {\em genefilter: methods for filtering genes from high-throughput
  experiments\/}.
\newblock R package version 1.70.0.

\bibitem[Ghosal and van~der Vaart(2017)Ghosal and van~der Vaart]{ghosal_book}
Ghosal, S. and van~der Vaart, A. (2017).
\newblock {\em Fundamentals of nonparametric {B}ayesian inference\/}.
\newblock Cambridge Series in Statistical and Probabilistic Mathematics.
  Cambridge University Press.

\bibitem[Green(1995)Green]{rjmcmc}
Green, P.~J. (1995).
\newblock Reversible jump {M}arkov chain {M}onte {C}arlo computation and
  {B}ayesian model determination.
\newblock {\em Biometrika\/}, {\bf 82}, 711--732.

\bibitem[Green and Thomas(2013)Green and Thomas]{green2013sampling}
Green, P.~J. and Thomas, A. (2013).
\newblock Sampling decomposable graphs using a {M}arkov chain on junction
  trees.
\newblock {\em Biometrika\/}, {\bf 100}, 91--110.

\bibitem[Gu {\em et~al.}(2014)Gu, Gu, Eils, Schlesner, and
  Brors]{gu2014circlize}
Gu, Z., Gu, L., Eils, R., Schlesner, M., and Brors, B. (2014).
\newblock {\it circlize} implements and enhances circular visualization in {R}.
\newblock {\em Bioinformatics\/}, {\bf 30}, 2811--2812.

\bibitem[Heng {\em et~al.}(2008)Heng, Painter, Elpek, Lukacs-Kornek, Mauermann,
  Turley, Koller, Kim, Wagers, Asinovski, {\em et~al.}]{immgen_paper}
Heng, T.~S., Painter, M.~W., Elpek, K., Lukacs-Kornek, V., Mauermann, N.,
  Turley, S.~J., Koller, D., Kim, F.~S., Wagers, A.~J., Asinovski, N., {\em
  et~al.} (2008).
\newblock The immunological genome project: networks of gene expression in
  immune cells.
\newblock {\em Nature Immunology\/}, {\bf 9}, 1091--1094.

\bibitem[Huang and Chen(2017)Huang and Chen]{huang2017learning}
Huang, F. and Chen, S. (2017).
\newblock Learning dynamic conditional {G}aussian graphical models.
\newblock {\em IEEE Transactions on Knowledge and Data Engineering\/}, {\bf
  30}, 703--716.

\bibitem[Ishwaran and Rao(2005)Ishwaran and Rao]{Ishwaran2005spikeslab}
Ishwaran, H. and Rao, J.~S. (2005).
\newblock Spike and slab variable selection: {F}requentist and {B}ayesian
  strategies.
\newblock {\em The Annals of Statistics\/}, {\bf 33}, 730--773.

\bibitem[Jones {\em et~al.}(2005)Jones, Carvalho, Dobra, Hans, Carter, and
  West]{jones2005experiments}
Jones, B., Carvalho, C., Dobra, A., Hans, C., Carter, C., and West, M. (2005).
\newblock Experiments in stochastic computation for high-dimensional graphical
  models.
\newblock {\em Statistical Science\/}, {\bf 20}, 388--400.

\bibitem[Kang and Deng(2020)Kang and Deng]{Xiaoning2020}
Kang, X. and Deng, X. (2020).
\newblock An improved modified {C}holesky decomposition approach for precision
  matrix estimation.
\newblock {\em Journal of Statistical Computation and Simulation\/}, {\bf 90},
  443--464.

\bibitem[Kastner(2019)Kastner]{KASTNER2019}
Kastner, G. (2019).
\newblock Sparse {B}ayesian time-varying covariance estimation in many
  dimensions.
\newblock {\em Journal of Econometrics\/}, {\bf 210}, 98--115.

\bibitem[Khare {\em et~al.}(2018)Khare, Rajaratnam, and
  Saha]{khare2018bayesian}
Khare, K., Rajaratnam, B., and Saha, A. (2018).
\newblock Bayesian inference for {G}aussian graphical models beyond
  decomposable graphs.
\newblock {\em Journal of the Royal Statistical Society: Series B: Statistical
  Methodology\/}, {\bf 80}, 727--747.

\bibitem[Khondker {\em et~al.}(2013)Khondker, Zhu, Chu, Lin, and
  Ibrahim]{khondker2013bayesian}
Khondker, Z.~S., Zhu, H., Chu, H., Lin, W., and Ibrahim, J.~G. (2013).
\newblock The {B}ayesian covariance lasso.
\newblock {\em Statistics and its Interface\/}, {\bf 6}, 243.

\bibitem[Koller and Friedman(2009)Koller and Friedman]{koller2009probabilistic}
Koller, D. and Friedman, N. (2009).
\newblock {\em Probabilistic graphical models: {P}rinciples and techniques\/}.
\newblock MIT press.

\bibitem[Ksheera~Sagar {\em et~al.}(2021)Ksheera~Sagar, Banerjee, Datta, and
  Bhadra]{KsheeraSagar2021precision}
Ksheera~Sagar, K.~N., Banerjee, S., Datta, J., and Bhadra, A. (2021).
\newblock Precision matrix estimation under the horseshoe-like prior-penalty
  dual.
\newblock arXiv:2104.10750.

\bibitem[Lauritzen(1996)Lauritzen]{lauritzen1996graphical}
Lauritzen, S.~L. (1996).
\newblock {\em Graphical models\/}.
\newblock Clarendon Press.

\bibitem[Lee {\em et~al.}(2013)Lee, Wang, {\em et~al.}]{ly6c}
Lee, P.~Y., Wang, J.-X., {\em et~al.} (2013).
\newblock Ly6 family proteins in neutrophil biology.
\newblock {\em Journal of Leukocyte Biology\/}, {\bf 94}, 585--594.

\bibitem[Lenkoski(2013)Lenkoski]{lenkoski2013direct}
Lenkoski, A. (2013).
\newblock A direct sampler for {G-W}ishart variates.
\newblock {\em Stat\/}, {\bf 2}, 119--128.

\bibitem[Levina {\em et~al.}(2008)Levina, Rothman, and
  Zhu]{levina2008sparsecov}
Levina, E., Rothman, A., and Zhu, J. (2008).
\newblock {Sparse estimation of large covariance matrices via a nested Lasso
  penalty}.
\newblock {\em The Annals of Applied Statistics\/}, {\bf 2}, 245--263.

\bibitem[Li {\em et~al.}(2019a)Li, Craig, and Bhadra]{li2019graphical}
Li, Y., Craig, B.~A., and Bhadra, A. (2019a).
\newblock The graphical horseshoe estimator for inverse covariance matrices.
\newblock {\em Journal of Computational and Graphical Statistics\/}, {\bf 28},
  747--757.

\bibitem[Li {\em et~al.}(2019b)Li, Mccormick, and Clark]{li2019bayesianB}
Li, Z., Mccormick, T., and Clark, S. (2019b).
\newblock Bayesian joint spike-and-slab graphical lasso.
\newblock In {\em International Conference on Machine Learning\/}, pages
  3877--3885. PMLR.

\bibitem[Liang {\em et~al.}(2001)Liang, Buckley, {\em et~al.}]{ms4a}
Liang, Y., Buckley, T.~R., {\em et~al.} (2001).
\newblock Structural organization of the human {MS4A} gene cluster on
  chromosome 11q12.
\newblock {\em Immunogenetics\/}, {\bf 53}, 357--368.

\bibitem[Liu {\em et~al.}(2009)Liu, Lafferty, and
  Wasserman]{liu2009nonparanormal}
Liu, H., Lafferty, J., and Wasserman, L. (2009).
\newblock The nonparanormal: Semiparametric estimation of high dimensional
  undirected graphs.
\newblock {\em Journal of Machine Learning Research\/}, {\bf 10}, 2295--2328.

\bibitem[Mazumder and Hastie(2012)Mazumder and Hastie]{mazumder2012graphical}
Mazumder, R. and Hastie, T. (2012).
\newblock The graphical lasso: New insights and alternatives.
\newblock {\em Electronic Journal of Statistics\/}, {\bf 6}, 2125--2149.

\bibitem[Meinshausen and B{\"u}hlmann(2006)Meinshausen and
  B{\"u}hlmann]{meinshausen2006high}
Meinshausen, N. and B{\"u}hlmann, P. (2006).
\newblock {High-dimensional graphs and variable selection with the Lasso}.
\newblock {\em The Annals of Statistics\/}, {\bf 34}, 1436--1462.

\bibitem[Mitra {\em et~al.}(2013)Mitra, M{\"u}ller, Liang, Yue, and
  Ji]{mitra2013bayesian}
Mitra, R., M{\"u}ller, P., Liang, S., Yue, L., and Ji, Y. (2013).
\newblock A {B}ayesian graphical model for chip-seq data on histone
  modifications.
\newblock {\em Journal of the American Statistical Association\/}, {\bf 108},
  69--80.

\bibitem[Mohammadi and Wit(2015)Mohammadi and Wit]{mohammadi2015bayesian}
Mohammadi, A. and Wit, E.~C. (2015).
\newblock Bayesian structure learning in sparse {G}aussian graphical models.
\newblock {\em Bayesian Analysis\/}, {\bf 10}, 109--138.

\bibitem[Mohammadi {\em et~al.}(2021)Mohammadi, Massam, and
  Letac]{mohammadi2021jasa}
Mohammadi, R., Massam, H., and Letac, G. (2021).
\newblock Accelerating {B}ayesian structure learning in sparse {G}aussian
  graphical models.
\newblock {\em Journal of the American Statistical Association\/}.
\newblock To appear.

\bibitem[M{\"u}ller {\em et~al.}(2004)M{\"u}ller, Parmigiani, Robert, and
  Rousseau]{muller04}
M{\"u}ller, P., Parmigiani, G., Robert, C., and Rousseau, J. (2004).
\newblock Optimal sample size for multiple testing: {T}he case of gene
  expression microarrays.
\newblock {\em Journal of the American Statistical Association\/}, {\bf 99},
  990--1001.

\bibitem[Neal(2000)Neal]{neal2000markov}
Neal, R.~M. (2000).
\newblock Markov chain sampling methods for {D}irichlet process mixture models.
\newblock {\em Journal of Computational and Graphical Statistics\/}, {\bf 9},
  249--265.

\bibitem[Painter {\em et~al.}(2011)Painter, Davis, Hardy, Mathis, Benoist,
  Consortium, {\em et~al.}]{micro1_train1}
Painter, M.~W., Davis, S., Hardy, R.~R., Mathis, D., Benoist, C., Consortium,
  I. G.~P., {\em et~al.} (2011).
\newblock {Transcriptomes of the B and T lineages compared by multiplatform
  microarray profiling}.
\newblock {\em The Journal of Immunology\/}, {\bf 186}, 3047--3057.

\bibitem[Pati {\em et~al.}(2014)Pati, Bhattacharya, Pillai, and
  Dunson]{pati2014}
Pati, D., Bhattacharya, A., Pillai, N.~S., and Dunson, D. (2014).
\newblock Posterior contraction in sparse {B}ayesian factor models for massive
  covariance matrices.
\newblock {\em The Annals of Statistics\/}, {\bf 42}, 1102--1130.

\bibitem[Peng {\em et~al.}(2009)Peng, Wang, Zhou, and Zhu]{peng2009partial}
Peng, J., Wang, P., Zhou, N., and Zhu, J. (2009).
\newblock Partial correlation estimation by joint sparse regression models.
\newblock {\em Journal of the American Statistical Association\/}, {\bf 104},
  735--746.

\bibitem[Pitt {\em et~al.}(2006)Pitt, Chan, and Kohn]{pitt2006efficient}
Pitt, M., Chan, D., and Kohn, R. (2006).
\newblock Efficient {B}ayesian inference for {G}aussian copula regression
  models.
\newblock {\em Biometrika\/}, {\bf 93}, 537--554.

\bibitem[Polson and Scott(2010)Polson and Scott]{polson2010shrink}
Polson, N.~G. and Scott, J.~G. (2010).
\newblock Shrink globally, act locally: Sparse {B}ayesian regularization and
  prediction.
\newblock {\em Bayesian Statistics\/}, {\bf 9}, 1--24.

\bibitem[Pourahmadi(2013)Pourahmadi]{pourahmadi2013high}
Pourahmadi, M. (2013).
\newblock {\em High-dimensional covariance estimation\/}.
\newblock John Wiley \& Sons.

\bibitem[Rothman {\em et~al.}(2008)Rothman, Bickel, Levina, Zhu, {\em
  et~al.}]{rothman2008sparse}
Rothman, A.~J., Bickel, P.~J., Levina, E., Zhu, J., {\em et~al.} (2008).
\newblock Sparse permutation invariant covariance estimation.
\newblock {\em Electronic Journal of Statistics\/}, {\bf 2}, 494--515.

\bibitem[Roverato(2002)Roverato]{roverato2002hyper}
Roverato, A. (2002).
\newblock Hyper inverse {W}ishart distribution for non-decomposable graphs and
  its application to {B}ayesian inference for {G}aussian graphical models.
\newblock {\em Scandinavian Journal of Statistics\/}, {\bf 29}, 391--411.

\bibitem[Sabnis {\em et~al.}(2016)Sabnis, Pati, Engelhardt, and
  Pillai]{sabnis2016divide}
Sabnis, G., Pati, D., Engelhardt, B., and Pillai, N. (2016).
\newblock A divide and conquer strategy for high dimensional {B}ayesian factor
  models.
\newblock {\em arXiv preprint arXiv:1612.02875\/}.

\bibitem[Scott and Berger(2010)Scott and Berger]{scott10}
Scott, J.~G. and Berger, J.~O. (2010).
\newblock {Bayes and empirical-Bayes multiplicity adjustment in the
  variable-selection problem}.
\newblock {\em The Annals of Statistics\/}, {\bf 38}, 2587--2619.

\bibitem[Shalizi(2009)Shalizi]{shalizi2009}
Shalizi, C.~R. (2009).
\newblock {Dynamics of Bayesian updating with dependent data and misspecified
  models}.
\newblock {\em Electronic Journal of Statistics\/}, {\bf 3}, 1039--1074.

\bibitem[Wang(2012)Wang]{wang2012bayesian}
Wang, H. (2012).
\newblock Bayesian graphical lasso models and efficient posterior computation.
\newblock {\em Bayesian Analysis\/}, {\bf 7}, 867--886.

\bibitem[West(1992)West]{west1992hyperparameter}
West, M. (1992).
\newblock {\em Hyperparameter estimation in {D}irichlet process mixture
  models\/}.
\newblock Duke University ISDS Discussion Paper\# 92-A03.

\bibitem[Witten {\em et~al.}(2011)Witten, Friedman, and Simon]{witten2011new}
Witten, D.~M., Friedman, J.~H., and Simon, N. (2011).
\newblock New insights and faster computations for the graphical lasso.
\newblock {\em Journal of Computational and Graphical Statistics\/}, {\bf 20},
  892--900.

\bibitem[Wolffe(2001)Wolffe]{histone}
Wolffe, A. (2001).
\newblock Histone genes.
\newblock In {\em Encyclopedia of Genetics\/}, pages 948--952. Academic Press,
  New York.

\bibitem[Yoshida and West(2010)Yoshida and West]{yoshida_west}
Yoshida, R. and West, M. (2010).
\newblock Bayesian learning in sparse graphical factor models via variational
  mean-field annealing.
\newblock {\em Journal of Machine Learning Research\/}, {\bf 11}, 1771--1798.

\bibitem[Yuan and Lin(2007)Yuan and Lin]{yuan2007model}
Yuan, M. and Lin, Y. (2007).
\newblock Model selection and estimation in the {G}aussian graphical model.
\newblock {\em Biometrika\/}, {\bf 94}, 19--35.

\bibitem[Zhang and Zou(2014)Zhang and Zou]{zhang2014}
Zhang, T. and Zou, H. (2014).
\newblock {Sparse precision matrix estimation via lasso penalized D-trace
  loss}.
\newblock {\em Biometrika\/}, {\bf 101}, 103--120.

\bibitem[Zhu {\em et~al.}(2014)Zhu, Khondker, Lu, and Ibrahim]{zhu2014}
Zhu, H., Khondker, Z., Lu, Z., and Ibrahim, J.~G. (2014).
\newblock Bayesian generalized low rank regression models for neuroimaging
  phenotypes and genetic markers.
\newblock {\em Journal of the American Statistical Association\/}, {\bf 109},
  977--990.

\end{thebibliography}


\begin{thebibliography}{}

\bibitem[Banerjee and Ghosal(2015)Banerjee and Ghosal]{banerjee2015bayesian}
Banerjee, S. and Ghosal, S. (2015).
\newblock Bayesian structure learning in graphical models.
\newblock {\em Journal of Multivariate Analysis\/}, {\bf 136}, 147--162.

\bibitem[Bhattacharya {\em et~al.}(2015)Bhattacharya, Pati, Pillai, and
  Dunson]{bhattacharya2015dirichlet}
Bhattacharya, A., Pati, D., Pillai, N.~S., and Dunson, D.~B. (2015).
\newblock Dirichlet-{L}aplace priors for optimal shrinkage.
\newblock {\em Journal of the American Statistical Association\/}, {\bf 110},
  1479--1490.

\bibitem[Castillo and van~der Vaart(2012)Castillo and van~der
  Vaart]{castillo2012needles}
Castillo, I. and van~der Vaart, A. (2012).
\newblock Needles and straw in a haystack: Posterior concentration for possibly
  sparse sequences.
\newblock {\em The Annals of Statistics\/}, {\bf 40}, 2069--2101.

\bibitem[Ghosal and van~der Vaart(2017)Ghosal and van~der Vaart]{ghosal_book}
Ghosal, S. and van~der Vaart, A. (2017).
\newblock {\em Fundamentals of nonparametric {B}ayesian inference\/}.
\newblock Cambridge Series in Statistical and Probabilistic Mathematics.
  Cambridge University Press.

\bibitem[Pati {\em et~al.}(2014)Pati, Bhattacharya, Pillai, and
  Dunson]{pati2014}
Pati, D., Bhattacharya, A., Pillai, N.~S., and Dunson, D. (2014).
\newblock Posterior contraction in sparse {B}ayesian factor models for massive
  covariance matrices.
\newblock {\em The Annals of Statistics\/}, {\bf 42}, 1102--1130.

\bibitem[Schloerke {\em et~al.}(2020)Schloerke, Cook, Larmarange, Briatte,
  Marbach, Thoen, Elberg, and Crowley]{ggally}
Schloerke, B., Cook, D., Larmarange, J., Briatte, F., Marbach, M., Thoen, E.,
  Elberg, A., and Crowley, J. (2020).
\newblock {\em {GGally: Extension to `ggplot2'}\/}.
\newblock R package version 2.0.0.

\bibitem[Vershynin(2012)Vershynin]{vershynin_book}
Vershynin, R. (2012).
\newblock {\em Introduction to the non-asymptotic analysis of random
  matrices\/}, page 210–268.
\newblock Cambridge University Press.

\bibitem[Yu(1997)Yu]{yu1997}
Yu, B. (1997).
\newblock Assouad, fano, and le cam.
\newblock In {\em Festschrift for Lucien Le Cam\/}, pages 423--435. Springer.

\end{thebibliography}


\clearpage\pagebreak\newpage
\newgeometry{textheight=9in, textwidth=6.5in,}
\pagestyle{fancy}
\fancyhf{}
\rhead{\bfseries\thepage}
\lhead{\bfseries SUPPLEMENTARY MATERIALS}

\baselineskip 20pt
\begin{center}
{\LARGE{Supplementary Materials for\\}} 
{\LARGE{\bf Bayesian Scalable Precision Factor\\
\vskip -9pt
Analysis 
for Massive Sparse \\ 
\vskip 9pt
Gaussian Graphical Models
}}
\end{center}

\setcounter{equation}{0}
\setcounter{page}{1}
\setcounter{table}{1}
\setcounter{figure}{0}
\setcounter{section}{0}
\numberwithin{table}{section}
\renewcommand{\theequation}{S.\arabic{equation}}
\renewcommand{\thesubsection}{S.\arabic{section}.\arabic{subsection}}
\renewcommand{\thesection}{S.\arabic{section}}
\renewcommand{\thepage}{S.\arabic{page}}
\renewcommand{\thetable}{S.\arabic{table}}
\renewcommand{\thefigure}{S.\arabic{figure}}
\baselineskip=15pt

\vspace{0cm}

\authors

\vskip 10mm
Supplementary materials discuss 
a general strategy for posterior computation under a broad class of generic priors, 
proofs of the theoretical results, 
some additional figures,
and an application to a NASDAQ-100 stock price dataset. 

\baselineskip=16pt
\newpage

\section{General Strategy for Posterior Computation} \label{sec: sm post computation}
In this section, we sketch out a general strategy for posterior inference for the precision factor model not limited to the prior specifications in Section \ref{sec: priors}.

\begin{description}
	\item[Sampling model:]  
	The data generating mechanism, we recall, is
	\vskip-5ex
	\bse 
		\by_{1:n} \simiid \mn_{d}(\bzero,\bOmega^{-1})~~~\text{ with }~~~ \bOmega=\bLambda\bLambda\trans+\bDelta, 
	\ese
	\vskip-2ex
	\noindent where $\bLambda$ is a $d\times q$ dimensional matrix with $q \leq d$ and $\bDelta = \diag(\delta_{1}^{2},\dots,\delta_{d}^{2})$. 
	\item[Generic priors:] We consider the following broad classes of hierarchical priors on $\bLambda$ and $\bDelta$ 
	with $\bxi$ and $\bzeta$ 
	being the respective associated hyperparameters
	\vskip-5ex
	\bse
		& \vect(\bLambda)\lvert \bxi \sim \mn_{dq}\{\bzero, g(\bxi)\}, 
		\quad \bxi\sim \Pi_{\bxi},
		\quad\quad \bDelta \lvert \bzeta \sim \Pi_{\bDelta \lvert \bzeta}(\bDelta \lvert \bzeta),\quad
		\bzeta \sim \Pi_{\bzeta}.
	\ese
	Varying choices of $g(\bxi), \Pi_{\bxi}, \Pi_{\bDelta \lvert \bzeta}, \Pi_{\bzeta}$ then produce a broad class of priors for the model parameters.
\end{description}

The Gibbs sampler then iterates through the following steps.
\begin{description}
	\item[Step 1. Sample the latent variables:]  Generate $\bu_{1:n} \simiid \mn_{q}(\bzero,\bP)$ with $\bP = (\bI_{q}+\bLambda\trans\bDelta^{-1}\bLambda)$ independently from $\by_{1:n}$ and let $\bv_{1:n} = \by_{1:n}+\bDelta^{-1}\bLambda \bP^{-1}\bu_{1:n}$. 
	\item[Step 2. Sample the rows of $\bLambda$:] We have $\bu_{i}=\sum_{r=1}^{d} \blambda_{r} v_{r,i}+\bvarepsilon_{i}$, 
	where $\blambda_{r} = (\lambda_{r,1},\dots,\lambda_{r,q})$ is the $r\th$ row of $\bLambda$ and $\bv_{i}=(v_{1,i},\dots,v_{d,i})\trans$. 
	Define $\bu_{i}^{(j)}=\bu_{i}-\sum_{r\neq j}\blambda_{r} v_{r,i}$. 
	Then $\bu_{i}^{(j)}=\blambda_{j} v_{j,i}+\bvarepsilon_{i}$. 
	Conditioned on $\bu_{i}^{(j)}$, $\bv_{i}$ {and the associated hyper-parameters}, $\blambda_{j}$'s can be updated sequentially $j=1,\dots,d$ from the following posterior distribution
	\[
	\blambda_{j} \sim \mn_{q}\{(\bD_{j}^{-1} + \norm{\bv^{(j)}}^{2} \bI_q )^{-1} \bw_{j},  (\bD_{j}^{-1} + \norm{\bv^{(j)}}^{2} \bI_q )^{-1}\},
	\]
	where $\bD_{j}$ is the prior covariance matrix of $\blambda_{j}$, 
	$\bv^{(j)}=(v_{j,1},\dots,v_{j,n})\trans$ 
	and $\bw_{j}=\sum_{i=1}^n v_{j,i} \bu_{i}^{(j)}$.
	\item[Step 3. Sample $\bDelta$:] Sample $\bDelta$ from $\Pi_{\bDelta\lvert \bv^{(1:d)}, \bzeta}$.
	\item[Step 4. Sample the hyperparameters of $\bLambda$:] Sample $\bxi$ from  $\bxi\sim \Pi_{\bxi\lvert \bLambda}$.
	\item[Step 5. Sample the hyperparameters of $\bDelta$:] Sample $\bzeta$ from  $\bzeta\sim \Pi_{\bzeta\lvert \bDelta}$.
\end{description}

For non-conjugate priors, appropriate Metropolis-within-Gibbs type MCMC algorithms can be employed in Steps 2, 3, 4 or 5.

\section{Proofs of Theoretical Results} \label{sec: sm proofs}

\textbf{Notations:} 
In what follows, for an operation `$*$', $\overline{a*b}$ is sometimes used for $(a*b)$.
Also, for two sequences $a_{n},b_{n} \geq 0$, $a_{n}\precsim b_{n}$ implies that $a_{n}\leq C b_{n}$ for some constant $C>0$;
$\an\asymp b_{n}$ implies that $0< \liminf \abs{{\an}/{b_{n}}} \leq \limsup \abs{{\an}/{b_{n}}}<\infty$.
$\abs{\bA}$ denotes the determinant of the square matrix $\bA$.
For a set $S$, $\abs{S}$ denotes its cardinality.
Let $\norm{\bx}$ is the Euclidean norm of a vector $\bx$.
We denote the Kullback-Leibler (KL) divergence between two mean zero Gaussian distributions with precision matrices $\bOmega$ and $\bOmega'$ by $\kl{\bOmega}{\bOmega'}$.
We borrow the following result from \citelatex{pati2014} (Lemma 1.1. from the supplement).
{For brevity of notations, we reuse $C$, $\wt{C}$, $C''$, etc. in the proofs to denote constants and their values may not be the same throughout the same proof.
	Nevertheless, we were careful to make sure that these quantities are indeed constants.}
\begin{lemma}
	\label{lemma:pati}
	For any two matrices $\bA$ and $\bB$,
	\begin{enumerate}[label=({{\roman*}})]
		\item $\smin{\bA} \fnorm{\bB}\leq \fnorm{\bA \bB} \leq \specnorm{\bA} \fnorm{\bB}$.
		\item $\smin{\bA} \specnorm{\bB}\leq \specnorm{\bA \bB} \leq \specnorm{\bA} \specnorm{\bB}$.
		\item \label{lemma:pati3} $\smin{\bA} \smin{\bB}\leq \smin{\bA \bB} \leq \specnorm{\bA} \smin{\bB}$.
	\end{enumerate}
\end{lemma}

\vskip 5pt
\begin{proof}[{\bf Theorem \ref{thm: posterior concentration}}]
	Let $\Pn=\Pnon \cup \Pntw$ be an arbitrary partition.
	Note that,
	\bse
	\pin \left(\specnorm{\bOmegan-\bOmegann}>M_{n}\epsilon_{n}\lvert \by_{1:n} \right) \leq \pin \left(\bOmegan\in \Pnon:\specnorm{\bOmegan-\bOmegann}>M_{n}\epsilon_{n}\lvert \by_{1:n} \right)+ \pin\left( \Pntw \lvert \by_{1:n} \right).
	\ese
	To prove the theorem, we show that, as $n\to\infty$, the posterior distribution on $\Pnon$ concentrates around $\bOmegann$ 
	while the remaining mass assigned to $\Pntw$ diminishes to zero.
	We adopt a variation of Theorem 8.22 by \citetlatex{ghosal_book} on posterior contraction rates.
	Define
	$B_{n,0}(\bOmegann,\epsilon)=\left\{\bOmegan \in \Pn: \kl{\bOmegann}{\bOmegan}\leq n\epsilon^{2} \right\} $.
	We now formally state the result tailored towards our application.	
	
	\begin{theorem}
		\label{th:ghosal}
		Let $\Pn$ be a class of distributions prametrized by $\bOmegan$, $\pP_{\bOmegan}$ the corresponding distribution and $\bOmegann$ be the true data-generating value of the parameter.
		Let $\en$ be a metric on $\Pn$ and $\Pnon\subset\Pn$. 
		For constants $\tau_{n}$ with $n\tau_{n}^{2}\geq1$ and every sufficiently large $j\in \nN$, assume the following conditions hold. 
		\begin{enumerate}[label=({{\roman*}})]
			\item\label{cond1} For some $C>0$, $\pin\left\{B_{n,0}(\bOmegann,\tau_{n})\right\}\geq e^{-Cn\tau_{n}^{2}}$;
			\item\label{cond2} Define the set $G_{j,n}=\left\{\bOmegan\in \Pnon: j\tau_{n}< \en \left( \bOmegan,\bOmegann\right)\leq 2j\tau_{n} \right\}$. 
			There exists tests $\phi_{n}$ such that, for some $K>0$,
			\bse
			\lim_{n\to \infty} \eE_{\bOmegann} \phi_{n}= 0;\qquad \sup_{\bOmegan\in G_{j,n}} \eE_{\bOmegan} (1-\phi_{n})\leq \exp(- Knj^{2}\tau_{n}^{2} ).
			\ese
		\end{enumerate}			
		Then, $\pin \left\{\bOmegan\in\Pnon:  \en\left(\bOmegan, \bOmegann\right) >M_{n}\tau_{n}\lvert \by_{1:n} \right\} \to 0$ in $\pP_{\bOmegann}$-probability for any $M_{n}\to\infty$. 	
	\end{theorem}
	We define 
	$\Pnon=\{\bOmegan: \specnorm{\bOmegan}  <  C''(\qnn\sn)^{4}\log(\dn\qn)\}$, $\Pntw=\Pn \backslash \Pnon $, $\en\left(\bOmegan, \bOmegann\right)= \frac{\specnorm{\bOmegan - \bOmegann}}{C''(\qnn\sn)^{4}\log(\dn\qn)}$  and $n \tau_{n}^{2}=\qnn\sn\log(\dn\qn)$ for some large enough constant $C''>0$.	
	We verify \ref{cond1} in Lemma \ref{lemma:kl_support} 
	and show the existence of a sequence of test functions satisfying \ref{cond2} in Lemma \ref{lemma:testfn}.
	Hence, 
	\bse
	\pin \left\{\bOmegan\in\Pnon:  \en\left(\bOmegan, \bOmegann\right) > M_{n}\tau_{n}\lvert \by_{1:n} \right\} \to 0 ~ \text{in}~\pP_{\bOmegann}\text{-probability for every}~M_{n}\to\infty.
	\ese 
	Subsequently applying the dominated convergence theorem (DCT), we get $\lim_{n\to \infty} \eE_{\bOmegann} \pin (\bOmegan\in\Pnon:  \specnorm{\bOmegan-\bOmegann}>M_{n}\tau_{n}\lvert \by_{1:n} ) = 0$. 
	To conclude the proof, we show that the remaining mass assigned to $\Pntw$ goes to 0 in the following theorem.
	
	\begin{theorem}
		\label{th:remainingmass}
		$\lim_{n\to \infty} \eE_{\bOmegann} \pin \left(\Pntw \lvert \by_{1:n} \right) = 0$.
	\end{theorem}
	\vskip-15pt
	
	\vskip -10pt
\end{proof}

\begin{proof}[{\bf Theorem \ref{th:ghosal}}]
	The theorem closely resembles Theorem 8.22 from \citetlatex{ghosal_book}.
	In the discussion following equation (8.22) in the book, 
	the authors noted that for the iid case, simpler theorems are obtained by using an absolute lower bound on the prior mass 
	and by replacing the local entropy by the global entropy. 
	In particular, \ref{cond1} implies Theorem 8.19(i) in the book.
	Similarly, \ref{cond2} is the same condition in Theorem 8.20 and hence the proof.
\end{proof}

\begin{lemma} (Kullback-Leibler upper bound)
	\label{lemma:kl_bound}
	Let $\bOmegan$ and $\bOmegann$ be $\dn\times \dn$ order positive definite matrices. 
	Then,
	\bse
	2\kl{\bOmegann}{\bOmegan}&=-\log \abs{\bOmegann^{-1}\bOmegan} +\trace(\bOmegann^{-1}\bOmegan-\bI_{\dn}) 
	\leq \fnorm{\bOmegann-\bOmegan}^{2}\specnorm{\bOmegann} /2\delmin^{6}.
	\ese
\end{lemma}

\begin{proof}
	Let $\bH=\bOmegann ^{-\half} \bOmegan \bOmegann^{-\half}$. 
	Letting $\psi_{1},\dots,\psi_{\dn}$ be the eigenvalues of $\bH$, we note that 
	\be
	2\kl{\bOmegann }{\bOmegan} =\sum_{j=1}^{\dn} \left\{(\psi_{j}-1)-\log\psi_{j} \right\}. \label{eq:klbound1}
	\ee
	\noindent Consider the function $h_{\beta}(x)=\log x -(x-1)+\beta(x-1)^{2}/2$ on $(0,\infty)$ for $\beta>1$. 
	Note that $h_{\beta}(x)\geq 0$ for all $x\geq 1/\beta$. 
	From \ref{ass2} and using Lemma \ref{lemma:pati},
	$\psi_{j}\geq\smin{\bH}\geq \smin{\bOmegan} \smin{\bOmegann^{-1}}= \delmin^{2} /\specnorm{\bOmegann}
	$. 
	Noting that $\delmin^{2} /\specnorm{\bOmegann}<1$, we set $\beta=\specnorm{\bOmegann}/\delmin^{2}$. 
	Therefore, 
	\be
	\sum_{j=1}^{\dn} \left\{\log\psi_{j} - (\psi_{j}-1) \right\} \geq -\frac{\specnorm{\bOmegann}}{2\delmin^{2}} \sum_{j=1}^{\dn}(\psi_{j}-1)^{2}=-\frac{\specnorm{\bOmegann}}{2\delmin^{2}} \fnorm{\bH-\bI_{\dn}}^{2}. \label{eq:klbound2}
	\ee
	\noindent Again, using Lemma \ref{lemma:pati} and \ref{ass2}, we have
	\be
	\fnorm{\bH-\bI_{\dn}} \leq \fnorm{\bOmegann-\bOmegan}\specnorm{\bOmegann^{-1}}\leq \fnorm{\bOmegann-\bOmegan}/\delmin^{2}. \label{eq:klbound3}
	\ee
	\noindent Combing \eqref{eq:klbound1}-\eqref{eq:klbound3}, we conclude the lemma.
\end{proof}

\vskip 5pt
\begin{lemma}
	Define the set $B_{n,0}(\bOmegann,\tau)=\{\bOmegan \in \Pn: \kl{\bOmegann }{\bOmegan}\leq n\tau^{2} \} $. 
	Then, $\pin\left\{B_{n,0}(\bOmegann,\tau_{n})\right\}\geq e^{-Cn\tau_{n}^{2}}$ for $n\tau_{n}^{2}=\sn\qnn\log(\dn\qn)$ and some absolute constant $C>0$.
	\label{lemma:kl_support}
\end{lemma}
\begin{proof}
	From Lemma \ref{lemma:kl_bound}, we note that
	\be
	&\pin\left\{B_{n,0}(\bOmegann,\tau_{n})\right\} \geq \pin\left(\fnorm{\bOmegann-\bOmegan}^{2}\specnorm{\bOmegann} /4 \delmin^{6} \leq n\tau_{n}^{2}\right)\notag \\
	\geq& \pin\left(\fnorm{\bLambdan\bLambdan\trans -\bLambdann\bLambdann\trans}\leq  \frac{\sqrt{n}\tau_{n} \delmin^{3}}{\specnorm{\bOmegann}^{\half}} \right) \pin\left(\fnorm{\bDeltan-\bDeltann}\leq  \frac{\sqrt{n}\tau_{n} \delmin^{3}}{\specnorm{\bOmegann}^{\half}} \right) \label{eq:productprior}.
	\ee
	\noindent We handle the $\bLambda$ and $\bDelta$ parts in \eqref{eq:productprior} separately and conclude the proof by showing that each of the terms individually exceeds $e^{-Cn\tau_{n}^{2}}$ for some constant $C>0$.
	
	\noindent\textbf{The $\bLambda$ part in \eqref{eq:productprior}: }	
	We define $\tLambdann=[\bLambdann\quad \bzero^{\dn\times\overline{ \qn-\qnn} }]$, after augmenting $\qn-\qnn$ null columns to $\bLambdann$. 	
	Also $\tLambdann\tLambdann\trans=\bLambdann\bLambdann\trans$. 
	Invoking \ref{ass3} and Lemma \ref{lemma:pati}, we have that $\fnorm{\bLambdan - \tLambdann}<\epsilon$ implies that  $ \fnorm{\bLambdan\bLambdan\trans -\tLambdann\tLambdann\trans}\leq C\epsilon^{2}$ for some $C>0$. 
	Note that, from \ref{ass3} and \ref{ass2}, $\specnorm{\bOmegann}=O(1)$. 
	Let $\bA^{S}$ denote the vector of elements of $\bA$ corresponding to an index set $S$, 
	and $S_{0}$ denote the set of nonzero elements in $\tLambdann$ .
	Then, for some absolute constant $\wt{C}$, we have 	
	\be
	&\pin\left(\fnorm{\bLambdan\bLambdan\trans -\bLambdann\bLambdann\trans}\leq  \frac{\sqrt{n}\tau_{n} \delmin^{3}}{\specnorm{\bOmegann}^{\half}} \right)\geq \pin\left(\fnorm{\bLambdan - \tLambdann}\leq  \wt{C} n^{1/4} \tau_{n}^{\half} \right)\notag\\
	&\geq \pin\left(\norm{\bLambdan^{\Snt}- \bLambdann^{\Snt}}\leq \wt{C} n^{1/4} \tau_{n}^{\half}/2 \right) \pin\left(\norm{\bLambdan^{\Snt^{C}}}\leq \wt{C} n^{1/4} \tau_{n}^{\half}/2 \right) \label{eq:factorpart}
	\ee
	\noindent The first term in \eqref{eq:factorpart} ensures that the Bayesian model assigns requisite prior mass around the \textit{signals} whereas the second term ensures that appropriate \textit{shrinkage} is envisaged through the prior.
	We show that each of these products exceeds $e^{-Cn\tau_{n}^{2}}$.	
	\begin{description}
		\item[The shrinkage part in \eqref{eq:factorpart}:] From \citetlatex[equation (10)]{bhattacharya2015dirichlet}, we get that 
		$\bLambdan  =\frac{2}{\bnn}((\lambda_{n,j,h}^{*})) $ where $\lambda_{n,j,h}^{*}$'s are marginally iid random variables with a priori with pdf $\Pi(\lambda_{n,j,h})=\abs{\lambda_{n,j,h}}^{(\an-1)/2} K_{1-\an}(\sqrt{2\abs{\lambda_{n,j,h}}})/\{2^{(1+\an)/2}\Gamma(\an)\}$ where $K_{\nu}(x)=\frac{\Gamma(\nu+1/2)(2x)^{\nu}}{\sqrt{\pi}}\int_{0}^{\infty} \frac{\cos t}{(t^{2}+x^{2})^{\nu+1/2}} \de t$ is the modified Bessel function of the second kind.
		Using Lemma 3.2 of the same paper and letting $\Hn=\bnn n^{1/4} \tau_{n}^{\half}$, we have
		\bse
		\pin\left(\norm{\bLambdan^{\Snt^{C}}}\leq \wt{C} n^{1/4} \tau_{n}^{\half}/2 \right)\geq \left\{\Pr(\abs{\lambda_{n,j,h}^{*}}\leq \frac{C\Hn}{\abs{\Snt^{C}}})\right\}^{\abs{\Snt^{C}}}\geq \left\{1-\frac{C}{\Gamma(\an)}\log\frac{\abs{\Snt^{C}}}{\Hn}\right\}^{\abs{\Snt^{C}}}.
		\ese	
		\noindent Now, $\abs{\Snt^{C}}=\dn\qn-\sn\qnn$ and $\Gamma(\an)=\Gamma(1+\an)/\an$.
		Since $\an=1/(\dn\qn)$ and $\Hn$ is an increasing sequence, $\frac{1}{\Gamma(\an)}\log\frac{\abs{\Snt^{C}}}{\Hn}=\frac{1}{\dn\qn\Gamma(1+\an)}\log\frac{\dn\qn-\sn\qnn}{\Hn}\leq 1$ and hence
		\vskip-5ex
		\bse
		\begin{split}
			\pin\left(\norm{\bLambdan^{\Snt^{C}}}\leq \wt{C} n^{1/4} \tau_{n}^{\half}/2 \right)
			& \geq \left\{1-\frac{C}{\Gamma(\an)}\log\frac{\abs{\Snt^{C}}}{\Hn}\right\}^{\abs{\Snt^{C}}}\geq \left\{1-\frac{C\log (\dn\qn)}{\dn\qn\Gamma(1+\an)} \right\}^{\dn\qn}\\ 
			& \geq e^{-C\log(\dn\qn)}\left\{1-\frac{C^{2} \log^{2}(\dn\qn)}{\dn\qn}\right\}\geq e^{-C'\qnn\sn\log(\dn\qn)}.
		\end{split}
		\ese	
		\vskip-2ex
		\noindent The second last inequality in the previous equation follows since, for any $n\geq 1$ and $\abs{x}\leq n$, 
		we have $\left(1+\frac{x}{n}\right)^{n}\geq e^{x}\left(1-\frac{x^{2}}{n}\right)$.
		\item[The signal part in \eqref{eq:factorpart}:] Let us define $v_{q}(r)$ to be the $q$-dimensional Euclidean ball of radius $r$ centered at zero and $\abs{v_{q}(r)}$ denotes its volume.
		For the sake of brevity, denote $v_{q}=\abs{v_{q}(1)}$, so that $\abs{v_{q}(r)}=r^{q}v_{q}$. 
		Letting $\bLambdan^{*}=((\lambda^{*}_{n,j,h}))^{\dn\times\qn}$ ($\lambda^{*}_{n,j,h}$'s defined earlier)
		and $\tn=\bnn\left(\wt{C} n^{1/4} \tau_{n}^{\half}/2+\norm{\bLambdann^{\Snt}}\right)$, we have
		\be
		\begin{split}
			\pin\left(\norm{\bLambdan^{\Snt}- \bLambdann^{\Snt}}\leq \wt{C} n^{1/4} \tau_{n}^{\half}/2 \right) 
			&= \pin\left(\norm{\bLambdan^{*\Snt}- \bnn\bLambdann^{\Snt}}\leq \wt{C}\bnn n^{1/4} \tau_{n}^{\half}/2 \right)\\
			& \geq \abs{v_{\abs{\Snt}}(\tn)}   \inf_{v_{\abs{\Snt}}(\tn)} \pin(\bLambdan^{*\Snt}). \label{eq:signalpart}
		\end{split}
		\ee
		\noindent Note that, $\abs{\Snt}=\sn\qnn$ and $\norm{\bLambdann^{\Snt}}=\fnorm{\bLambdann}$
		and using \citetlatex[Lemma 5.3]{castillo2012needles}, $v_{q}\asymp (2\pi e)^{q/2}q^{-\overline{q+1}/2}$.
		Note also that for $x> 0$, $\frac{\log(1+x)}{x}\leq \frac{1}{\sqrt{1+x}}$ which implies that $x^{x}\leq e^{x^{3/2}}$. 
		Hence,
		\bse
		\abs{v_{\abs{\Snt}}(\tn)}=\tn^{\sn\qnn}v_{\sn\qnn}\geq\exp\{\sn\qnn\log\tn-C'(\qnn\sn)^{3/2}/2\}\geq \exp\{-C\sn\qnn\log (\dn\qn)\}.
		\ese
		\noindent The last inequality follows from \ref{ass1} and the prior specifications.	
		From \citetlatex[Lemma 3.1]{bhattacharya2015dirichlet} we have
		\bse
		&& \inf_{v_{\abs{\Snt}}(\tn)} \pin(\bLambdan^{*\Snt}) =\exp[-C\left\{\sn\qnn\log(1/\an)+ (\sn\qnn)^{3/4}\sqrt{\tn} \right\}]\\
		&& =\exp\left[-C\left\{\sn\qnn\log(\dn\qn) 
		+ (\sn\qnn)^{3/4}\bnn^{\half}\left(\wt{C} n^{1/4} \tau_{n}^{\half}/2+\fnorm{\bLambdann}\right)^{\half} \right\}\right]\\
		&& \geq \exp(-C \sn\qnn\log\overline{\dn\qn}),			
		\ese
		\noindent since, $\bnn =\log^{3/2}(\dn\qn)$, $n\tau_{n}^{2}=\sn\qnn\log\overline{\dn\qn}$ and $\fnorm{\bLambdann}^{2}=O(\qnn)$ from \ref{ass3}.
		Hence from \eqref{eq:signalpart}, 
		$\pin\left(\norm{\bLambdan^{\Snt}- \bLambdann^{\Snt}}\leq \wt{C} n^{1/4} \tau_{n}^{\half}/2 \right) \geq \exp(-C\sn\qnn\log\overline{\dn\qn})$.
	\end{description}	
	
	\noindent\textbf{The $\bDelta$ part in \eqref{eq:productprior}:} Note that, $\pin\left(\fnorm{\bDeltan-\bDeltann}\leq  \frac{\sqrt{n}\tau_{n} \delmin^{3}}{\specnorm{\bOmegann}^{\half}} \right) \geq \pin\left( \abs{\deltann^{2}-\deltan^{2}} \leq  \frac{\sqrt{n}\tau_{n} \delmin^{3}}{\sqrt{\dn \specnorm{\bOmegann}}} \right)$. 
	Now, $\pin ( \abs{\deltann^{2}-\deltan^{2}} \leq  x )\geq e^{-\deltann^{2}}(1- e^{-2x}) $. 	
	Using $1-e^{-2x}\geq x$ for $x\in(0,\half)$ and since $\frac{\sqrt{n}\tau_{n} \delmin^{3}}{\sqrt{\dn \specnorm{\bOmegann}}}\to 0$ as $n\to\infty$, 
	we have for some $C'>0$
	\be
	\hspace*{-10pt} \pin\left( \abs{\deltann^{2}-\deltan^{2}} \leq  \frac{\sqrt{n}\tau_{n} \delmin^{3}}{\sqrt{\dn \specnorm{\bOmegann}}} \right)
	\geq \exp(-\delta_{0}^{2}+\log \frac{\sqrt{n}\tau_{n} \delmin^{3}}{\sqrt{\dn \specnorm{\bOmegann}}}) 
	\geq \exp(-C'\sn\qnn\log\overline{\dn\qn}). \notag 
	\ee
\end{proof}

\begin{corollary}
	\label{cor:arbit_tau}
	For arbitrary $\wt{\tau}_{n}> n^{-\half}$, we have
	$\pin\left\{B_{n,0}(\bOmegann,\wt{\tau}_{n})\right\}\geq e^{-C\qnn\sn\log(\dn\qn) }$ for some constant $C>0$.	
\end{corollary}
\noindent Following the proof of Lemma \ref{lemma:kl_support}, it can be seen that the above corollary holds.

\begin{lemma}
	\label{lemma:testfn}
	Let $\bOmegann\in\Pnn $ and the set $G_{j,n}=\left\{\bOmegan\in \Pnon: j\tau_{n}<\frac{1}{\zeta_{n} }\specnorm{\bOmegan-\bOmegann}\leq 2j\tau_{n} \right\}$ denote an annulus of inner radius $j\tau_{n} \zeta_{n} $ and outer radius $(j+1)\tau_{n} \zeta_{n} $, where $\zeta_{n}=C''(\qnn\sn)^{4}\log(\dn\qn)$, in operator norm around $\bOmegann$ for some integer $j>1$. 
	Based on iid samples $\by_{1:n}$ from $\mn_{\dn}(\bzero, \bOmegann^{-1})$, consider the following hypothesis testing problem
	\bse
	H_{0}:\bOmegan=\bOmegann\text{ versus }H_{1}:\bOmegan\in G_{j,n}.
	\ese
	\noindent We simulate $\bu_{1:n}\simiid \mn_{\qnn}(\bzero, \bP_{0n})$, with $\bP_{0n} = (\bI_{\qnn}+\bLambdann\trans\bDeltann^{-1}\bLambdann)$ independently from $\by_{1:n}$ and define $\bv_{i} = \by_{i}+\bDeltann^{-1}\bLambdann \bP_{0n}^{-1}\bu_{i}$. 
	Letting $\bV\trans=(\bv_{1},\dots,\bv_{n})$, 
	we define the following test function $\phi_{n}=\mathbbm{1}\left\{\specnorm{ \bLambdann \trans \left(\frac{1}{n} \bV\trans \bV -\bDeltann^{-1}\right) \bLambdann}> \tau_{n}  \right\}$.
	Then 
	\bse
	\lim_{n\to \infty} \eE_{\bOmegann} \phi_{n}= 0;\qquad \sup_{\bOmegan\in G_{j,n}} \eE_{\bOmegan} (1-\phi_{n})\leq \exp(- Knj^{2}\tau_{n}^{2} ),
	\ese
	\noindent	for some absolute constant $K>0$.
\end{lemma}
\begin{proof}
	\begin{description}
		\item[Type-I error:] 
		Under $H_{0}$, we have $\bv_{1:n}\simiid\mn_{\dn}(\bzero,\bDeltann^{-1})$. 
		Letting $\bXinn=\cov \left({\bLambdann\trans\bv_{i}}\right)= \bLambdann\trans\bDeltann^{-1}\bLambdann$, we have,
		$\phi_{n}\leq \mathbbm{1}\left(\specnorm{\bXinn} \specnorm{\frac{1}{n}\sum_{i=1}^{n}\bz_{i}\bz_{i}\trans -\bI_{\qnn}}>\tau_{n}  \right)$, where $\specnorm{\bXinn}=O(1)$ from \ref{ass3}.
		From \citetlatex[Corollary 5.50]{vershynin_book},  $\eE_{H_{0}}\phi_{n} \leq 2\exp(-\wt{C}\qnn\tn^{2} )$ for a universal constant $\wt{C}>0$ and any sequence $\tn$ satisfying $\qnn \tn^{2}\leq n\tau_{n}^{2}$. 
		Since from \ref{ass1}, $\qnn\to\infty$ and $n\tau_{n}^{2}/\qnn \to \infty$ as $n\to\infty$, $\tn$ can be constructed such that $\tn\succsim1$ and hence $\lim_{n\to \infty} \eE_{H_{0}}\phi_{n} =0$.
		\item[Type-II error:] 
		If $\by_{1:n}\simiid\mn_{\dn}(\bzero,\bOmegan)$ with $\bOmegan=\bLambdan\bLambdan\trans+\bDeltan$, then $\bv_{1:n}\simiid\mn_{\dn}(\bzero,\bDeltan^{-1} +  \bnablan)$ where $\bnablan=\bDeltann^{-1} \bLambdann(\bI_{\qnn}+\bLambdann\trans\bDeltann^{-1}\bLambdann)^{-1}\bLambdann\trans\bDeltann^{-1} -\bDeltan^{-1} \bLambdan(\bI_{\qn}+\bLambdan\trans\bDeltan^{-1}\bLambdan)	^{-1}\bLambdan\trans\bDeltan^{-1}$. 
		Notably, $\bDeltan^{-1} -\bDeltann^{-1}+\bnablan=\bOmegan-\bOmegann$. 
		Now,
		\vskip-5ex
		\be
		&& \hskip -25pt 1-\phi_{n} = \mathbbm{1}\left\{\specnorm{ \bLambdann \trans \left(\frac{1}{n} \bV\trans \bV -\bDeltann^{-1}\right) \bLambdann}\leq \tau_{n}  \right\}\notag \\
		&=& \mathbbm{1}\left\{\specnorm{ \bLambdann \trans \left\{\left(\frac{1}{n} \bV\trans \bV -\bDeltan^{-1}-\bnablan\right) + (\bDeltan^{-1} -\bDeltann^{-1}+\bnablan)\right\} \bLambdann}\leq \tau_{n}  \right\}\notag\\
		&\leq& \mathbbm{1}\left\{\specnorm{ \bLambdann \trans \left(\frac{1}{n} \bV\trans \bV -\bDeltan^{-1}-\bnablan\right)\bLambdann }>  \specnorm{\bLambdann\trans (\bOmegan-\bOmegann)\bLambdann }- \tau_{n}  \right\}\notag\\
		&\leq& \mathbbm{1}\left\{\specnorm{  \left\{\frac{1}{n} \bLambdann\trans\bV\trans \bV\bLambdann -\bLambdann\trans(\bDeltan^{-1}+\bnablan)\bLambdann\right\} }> {\smin{\bLambdann\trans \bLambdann}} \specnorm{\bOmegan-\bOmegann} - \tau_{n}  \right\}\notag\\
		&\leq& \mathbbm{1}\left\{ \specnorm{ \frac{1}{n}\sum_{i=1}^{n}\bz_{i}\bz_{i}\trans -\bI_{\qn} }> \frac{\smin{\bLambdann\trans \bLambdann}}{\specnorm{\bDeltan^{-1}+\bnablan} \specnorm{\bLambdann}^{2}}  \left(\specnorm{\bOmegan-\bOmegann} - \frac{\tau_{n}}{\smin{\bLambdann\trans \bLambdann}}\right)  \right\}, ~~~~~\label{eq:type2error}
		\ee
		\noindent where $\bz_{1:n}\simiid\mn_{\qn}(\bzero,\bI_{\qn})$. 
		By construction of $\Pnon$, $\specnorm{\bDeltan^{-1}+\bnablan}\leq \zeta_{n}$ for $\bOmegan\in \Pnon$. 
		Hence, for $\bOmegan\in G_{j,n}$, the RHS of \eqref{eq:type2error} is bounded below by $j\tau_{n}$ for sufficiently large $j$.
		Using \citetlatex[Eqn 5.26]{vershynin_book} again, we have for all $\bOmegan \in G_{j,n} $, $\eE_{\bOmegan}(1-\phi_{n})\leq e^{-Knj^{2}\tau_{n}^{2}}$.		
	\end{description}
	Hence the proof.
\end{proof}
\begin{proof}[{\bf Theorem \ref{th:remainingmass}}]
	For densities $p$ and $q$, define $\klv{p}{q}=\int \left\{\log\frac{p}{q}-\kl{p}{q}\right\}^{2}\de p$.
	When $p$ and $q$ are mean zero multivariate Gaussian densities with precision matrices $\bOmega$ and $\bOmega'$, we simply denote $\klv{\bOmega}{\bOmega'}$. We define the set $B_{n,2}(\bOmegann, \epsilon)=\{\bOmegan\in\Pn: \kl{\bOmegann}{\bOmegan}\leq n\epsilon^{2},~ 
	\klv{\bOmegann}{\bOmegan} \leq n\epsilon^{2}\}$.
	
	From \citetlatex[Theorem 3.1]{banerjee2015bayesian} and the proof of Lemma \ref{lemma:kl_bound}, we have $\klv{\bOmegann}{\bOmegan}=\frac{1}{2\delmin^{4}}\fnorm{\bOmegann-\bOmegan}^{2}$ and therefore $B_{n,2}(\bOmegann, \epsilon)= \{\bOmegan\in\Pn: \fnorm{\bOmegann-\bOmegan}^{2}\leq C\epsilon \}$ for some constant $C>0$.
	
	Let $\{\tn\}_{n=1}^{\infty}$ be an increasing sequence of positive numbers.	
	Then, $ \pin(\specnorm{\bLambdan}\geq \tn )\leq \pin(\fnorm{\bLambdan}\geq \tn )\leq\pin(\norm{\vect(\bLambdan^{*})}_{\ell_{1}}\geq \bnn\tn )\leq 2e^{-C\sqrt{\bnn\tn}}$ for some constant $C>0$ where $\norm{\cdot}_{\ell_{1}}$ is the $\ell_{1}$ norm.
	The last inequality follows from \citetlatex[Lemma 7.4]{pati2014}. 
	Hence, for $\zeta_{n}=C''(\qnn\sn)^{4} \log(\dn\qn)$,	
	\bse
	\pin(\Pntw)\leq \pin \left\{\delta_{\max}^{2} + \specnorm{\bLambdan}^{2}\geq \zeta_{n} \right\} \geq \pin \left\{ \specnorm{\bLambdan^{*}}^{2}\geq K'C''(\qnn\sn\log\overline{\dn\qn})^{4}  \right\}  \leq e^{-CC''\qnn\sn\log (\dn\qn) }.
	\ese

	From Corollary \ref{cor:arbit_tau}, for arbitrary $\wt{\tau}_{n}>n^{-\half}$, $\pin\left\{B_{n,2}(\bOmegann, \wt{\tau}_{n})\right\}\geq e^{-C\qnn\sn\log (\dn\qn)}$.		
	Hence, $\frac{\pin(\Pntw)}{\pin\left\{B_{n,2}(\bOmegann, \wt{\tau}_{n})\right\}}\leq e^{-2C\qnn\sn\log (\dn\qn)}=o(e^{-2n\wt{\tau}_{n}^{2}})$ for $\wt{\tau}_{n}<\tau_{n}$.
	The last display holds for suitable large enough choice of $C''$, which, however, can be chosen independent of $n$.
	Using \citetlatex[Theorem 8.20]{ghosal_book} and subsequent application of DCT we conclude the proof.
\end{proof}

\begin{theorem}[minimax rate] \label{thm: sm minimax}
	If $\widehat{\bOmega}_{n}$ is a sequence of estimators of $\bOmegann\in\Pnn$ with $\qnn=O(1)$, 
	then for some absolute constant $C>0$
		\vskip-6ex
	\bse
	\inf_{\widehat{\bOmega}_{n}} \sup_{\bOmegann\in\Pnn} \specnorm{\widehat{\bOmega}_{n}- \bOmegann}\geq C\sqrt{\frac{\sn\log\dn}{n}}.
	\ese	
\end{theorem}
\vspace*{-3ex}
\begin{proof}
	We use Fano's lemma to derive a lower bound for the minimax risk.
	Let $\F=\{\bOmega_{(1)},\dots,\bOmega_{(m_{n})}\}$, $m_{n}\geq 2$ be a finite subset of $\Pnn$ and let $\widehat{\bOmega}_{n}$ be an estimate of $\bOmegann$ based on iid samples $\by_{1:n}$.
	Suppose for all $j\neq j'$, we have $e(\bOmega_{(j)},\bOmega_{(j')})\geq e_{m_{n}}$ for some metric $e$ and $\kl{\bOmega_{(j)}}{\bOmega_{(j')}}\leq K_{m_{n}}$. Letting $\eE_{j}$ denote the expectation under $\mn_{\dn}(\bzero,\bOmega_{(j)}^{-1})$, Fano's lemma \citeplatex{yu1997} implies
	\be
	\max_{1\leq j\leq m_{n}} \eE_{j} e(\bOmega_{(j)},\bOmega_{(j')}) \leq \frac{e_{m_{n}}}{2}\left(1- \frac{K_{m_{n}}+\log2}{\log m_{n}}\right).\label{eq:minmaxbound}
	\ee
	
	We now construct the finite class $\F$.
	Let $r_{n}=\dn-1$. Define $\M=\{\bx\in \rR^{r_{n}}: x_{j}\in\{0,1\}\text{ for all } j,~ \sum_{j}x_{j}=\sn \}$ to be the collection of all binary vectors of length $r_{n}$ with exactly $\sn$ ones and
	let $\hamm{\bx}{\by}$ be the Hamming distance between two binary vectors $\bx$ and $\by$.
	Let $\bg_{j}=(\bx_{j},0)$ denote the $\dn$-dimensional vector obtained by appending zero at the end of $\bx_{j}$ with $\bx_{j}\in \M$.
	With this definitions, set
	$\bOmega_{(j)}=\beta_{n}\bI_{\dn}+\gamma_{n} \bg_{j}\bg_{j}\trans+ \kappa_{n} \bvarepsilon_{\dn}\bvarepsilon_{\dn}\trans$ where $\bvarepsilon_{\dn}$ is the vector with 1 in the $\dn\th$ coordinate and zero elsewhere, and $\gamma_{n}\leq\beta_{n}\leq \kappa_{n}$ are positive sequences to be chosen below.
	
	Observe that if $\hamm{\bg_{j}}{\bg_{j'}}=\sn-p_{n}$ for $j\neq j$, then $\bg_{j}\trans \bg_{j'}=p_{n}$.
	Note also that $\bOmega_{(j)}-\bOmega_{(j')}=\gamma_{n}(\bg_{j}\bg_{j}\trans - \bg_{j'}\bg_{j'}\trans)$.
	The nonzero eigenvalues of the matrix $\bB=\bg_{j}\bg_{j}\trans - \bg_{j'}\bg_{j'}\trans$ are $(\sqrt{\sn^{2}-p_{n}^{2}},-\sqrt{\sn^{2}-p_{n}^{2}})$, since $\mathrm{rank}(\bB)=2$, $\trace(\bB)=0$ and $\trace(\bB^{2})=2(\sn^{2}-p_{n}^{2})$.
	This implies that $\specnorm{\bOmega_{(j)}-\bOmega_{(j')}}=\gamma_{n} (\sn^{2}-p_{n}^{2})$.
	Since $\bg_{j}\in\M$ for all $j$, by symmetry $\abs{\bOmega_{(j)}}=\abs{\bOmega_{(j')}}$ for all $j\neq j'$.
	Hence  $\kl{\bOmega_{(j)}}{\bOmega_{(j')}}=\frac{1}{2}\trace\left\{\bOmega_{(j)}\bOmega_{(j')}^{-1}-\dn\right\}$.
	Let $\bA=\beta_{n}(\bA+t\bg_{j}\bg_{j}\trans)$, where $\bA$ is a diagonal matrix with the first $(\dn-1)$ diagonals equaling one and the $\dn\th$ entry being $(1+\kappa_{n}/\beta_{n})$.
	Subsequently, applying the Woodbury matrix inversion identity, we get $(\bA+\tn\bg_{j'}\bg_{j'}\trans)^{-1}=\bA^{-1}-\frac{\tn}{1+\tn \sn}\bg_{j'}\bg_{j'}\trans$ so that
	$\bOmega_{(j)}\bOmega_{(j')}^{-1}=\bI_{\dn}-\frac{\tn}{1+\tn \sn}\bg_{j}\bg_{j}\trans+ \tn\bg_{j'}\bg_{j'}\trans-\frac{\tn^{2}p_{n}}{1+\tn\sn}\bg_{j}\bg_{j'}\trans$ with $\tn=\gamma_{n}/\beta_{n}$.
	Observing that $\trace(\bg_{j}\bg_{j}\trans)=\sn$ and $\trace(\bg_{j}\bg_{j'}\trans)=p_{n}$, we get $\kl{\bOmega_{(j)}}{\bOmega_{(j')}}=\frac{1}{2}\frac{\tn^{2}}{\tn\sn+1}(\sn^{2}-p_{n}^{2})$.
	
	Now, from \citetlatex[Lemma 5.6]{pati2014}, given $\sn\geq 6$, there exists a subset $\M_{0}=\{\bx_{1},\dots,\bx_{m_{n}}\}$ of $\M$ with $m_{n}\asymp \exp(C\sn\log\dn)$ and $\hamm{\bx_{j}}{\bx_{j'}}\geq \sn/3$ for all $1\leq j\neq j'\leq m_{n}$, where $C$ is a positive constant independent of $\dn$. 
	We set $\F=\{\bOmega_{(j)}:\bx_{j}\in\M_{0}\}$.
	Using the aforementioned lemma and preceding discussions, we have that $p_{n}$ is bounded above by $2\sn/3$ fo all pairs $j\neq j' \in\M_{0}$. 
	Hence, we can choose $e_{m_{n}}\geq c_{1}\gamma_{n}\sn$ and $K_{m_{n}}=(t\sn)^{2}=(\gamma_{n}\sn/\beta_{n})^{2}$ in \eqref{eq:minmaxbound}.
	To obtain $e_{m_{n}}$ as a lower bound to the minimax risk up to a constant, we set $K_{m_{n}}/\log m_{n}=C'$ for some $C'\in(0,1)$.
	Since $\log m_{n}\asymp C\sn \log\dn$, we obtain by choosing $\beta_{n}, \kappa_{n}\asymp 1 $, that $e_{m_{n}}^{2}=C(\gamma_{n}\sn)^{2}=C\frac{\sn\log\dn}{n}$.	
\end{proof}

\section{Additional Figures}
\subsection{Uncertainty Quantification}
\begin{figure}[H]
	\begin{subfigure}[b]{.99\linewidth}
		\includegraphics[width=\linewidth]{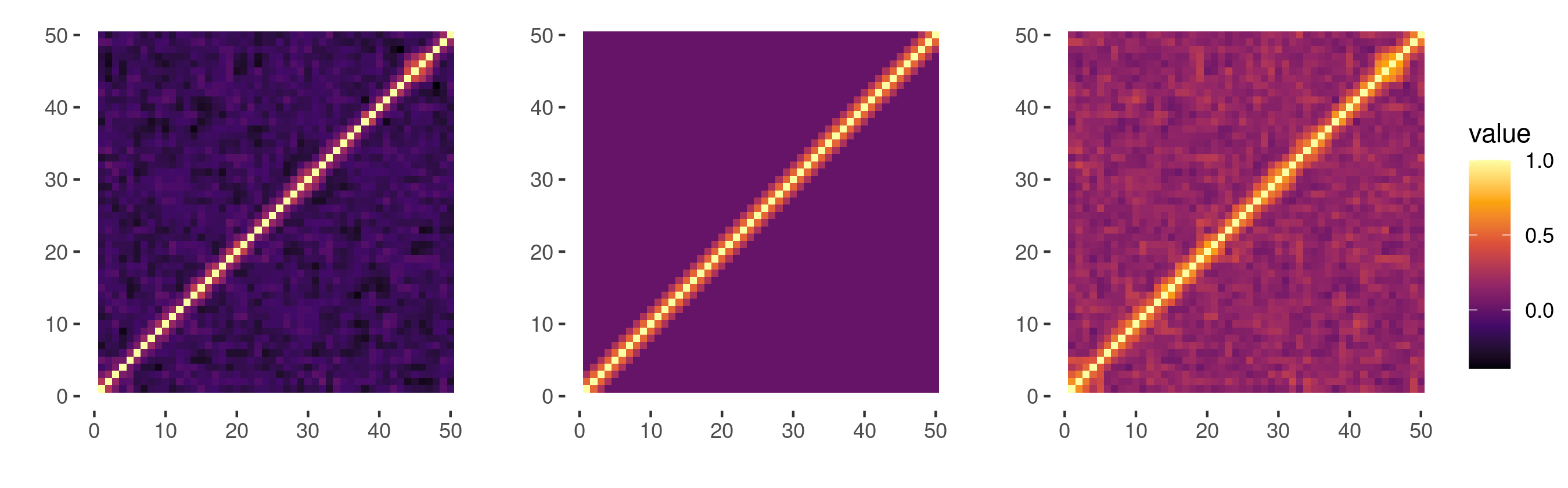}
		\caption{Quantiles for the AR(2) case}
		\label{fig:AR_quant}
	\end{subfigure}
	\vskip 10pt
	\begin{subfigure}[b]{.99\linewidth}
		\includegraphics[width=\linewidth]{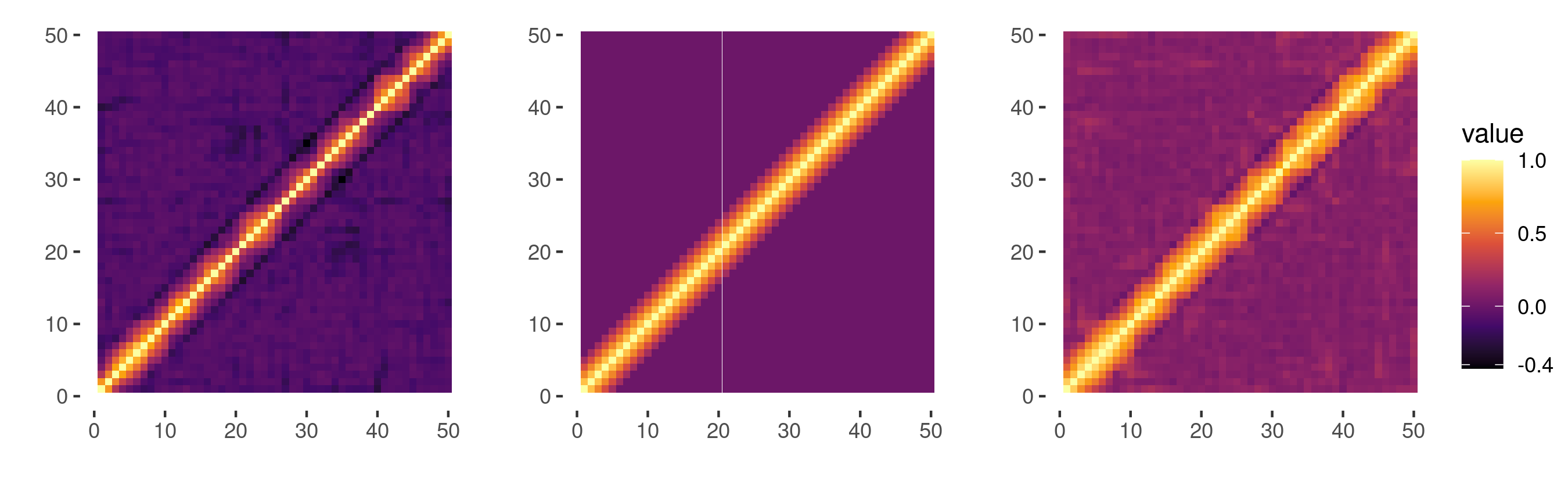}
		\caption{Quantiles for the banded case}
		\label{fig:banded_quant}
	\end{subfigure}
	\vskip 10pt
	\begin{subfigure}[b]{.99\linewidth}
		\includegraphics[width=\linewidth]{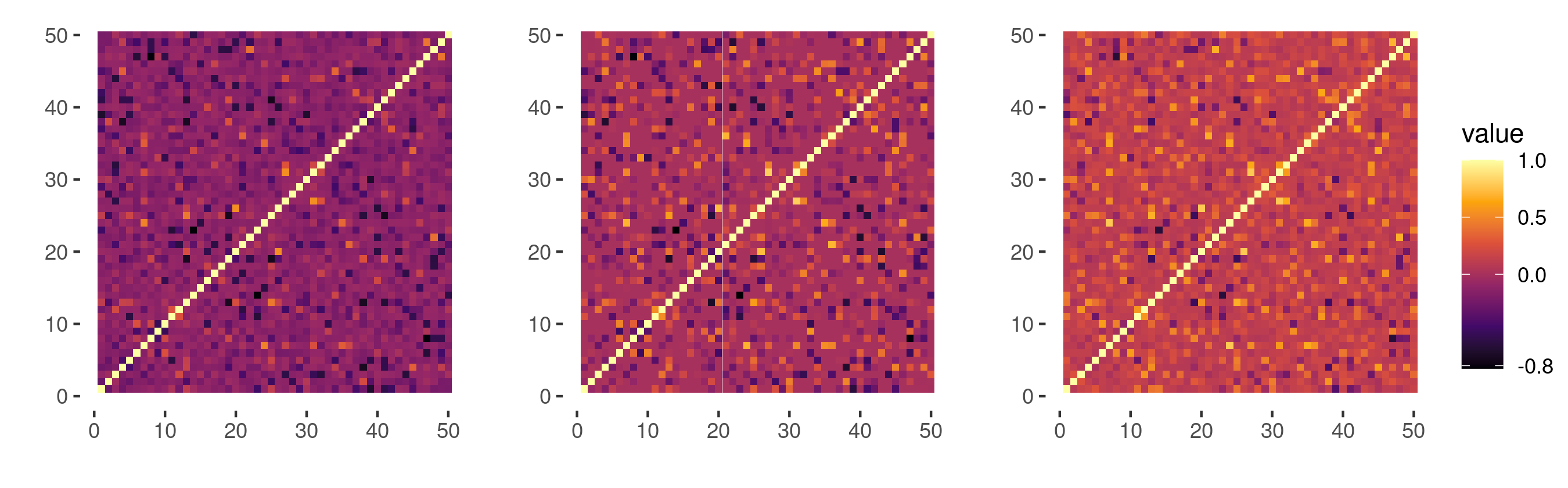}
		\caption{Quantiles for the RSM case}
		\label{fig:RSM_quant}
	\end{subfigure}
	\caption{Results of simulation experiments: Uncertainty quantification: In each panel, the middle heatmap is the simulation truth {of  $\bR=\diag(\bOmega)^{-\half} ~\bOmega~ \diag(\bOmega)^{-\half}$}, the left and the right heatmaps show the lower $2.5\%$ and the upper $97.5\%$ quantiles of $\bR$, respectively, estimated from the MCMC samples.}
	\label{fig:quantiles}
\end{figure}

\clearpage\newpage
\subsection{Graph Estimation}
{True and estimated graphs derived from the respective precision matrices as described in Section \ref{sec: sim studies} of the main paper using the  \texttt{GGally} package \citeplatex{ggally} in \texttt{R}.}
\begin{figure}[H]
	\begin{subfigure}[b]{.99\linewidth}
		\includegraphics[width=0.325\linewidth]{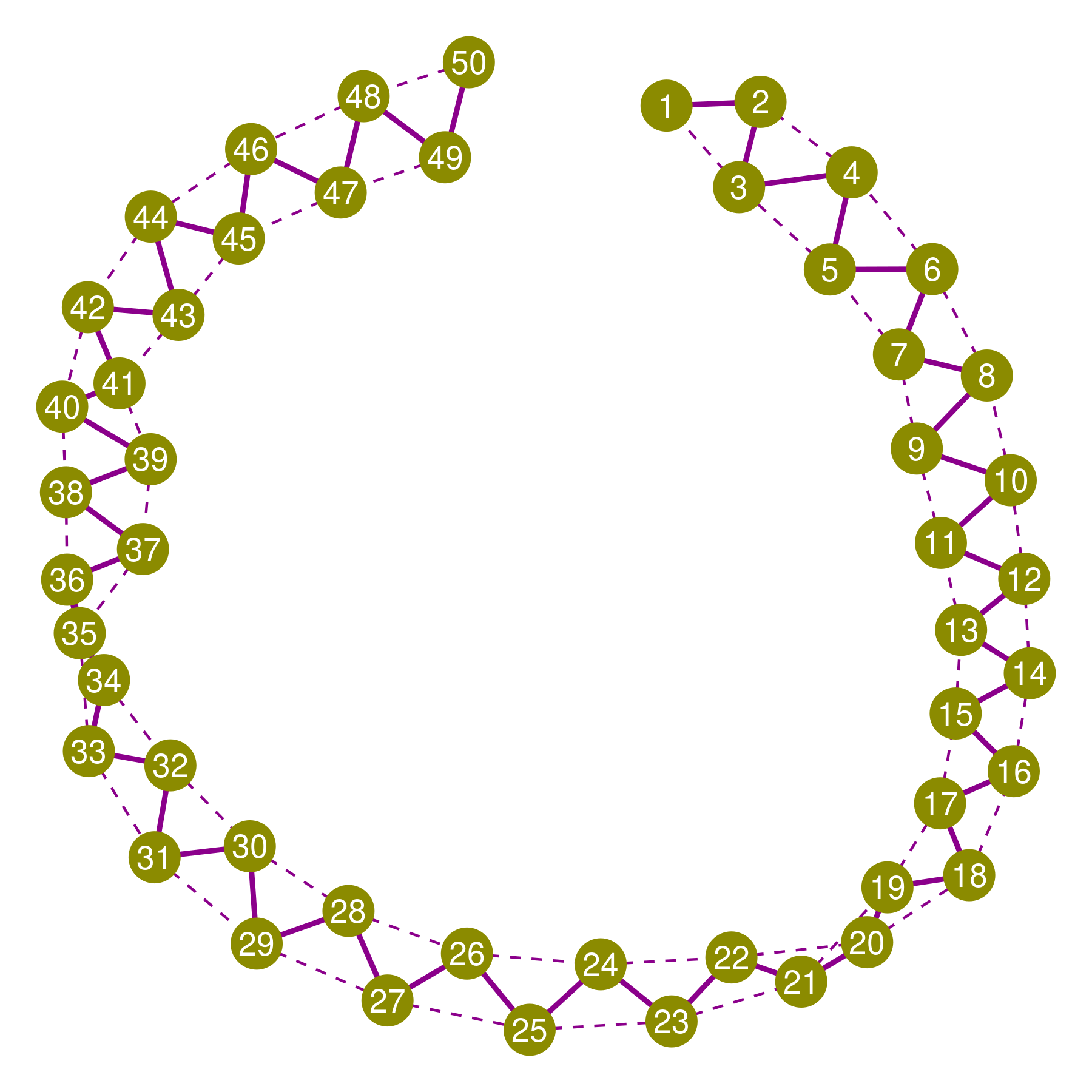}
		\includegraphics[width=0.325\linewidth]{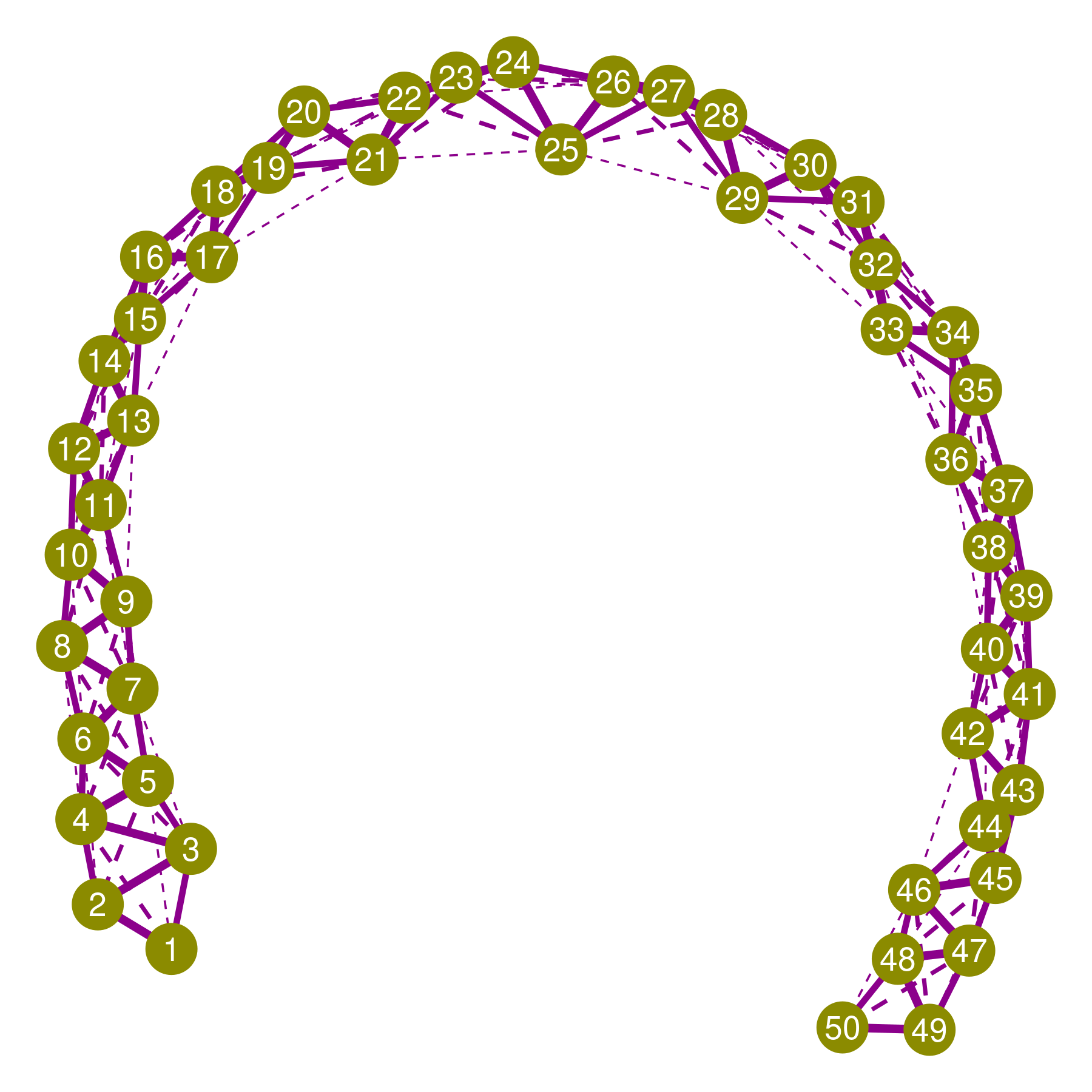}
		\includegraphics[width=0.325\linewidth]{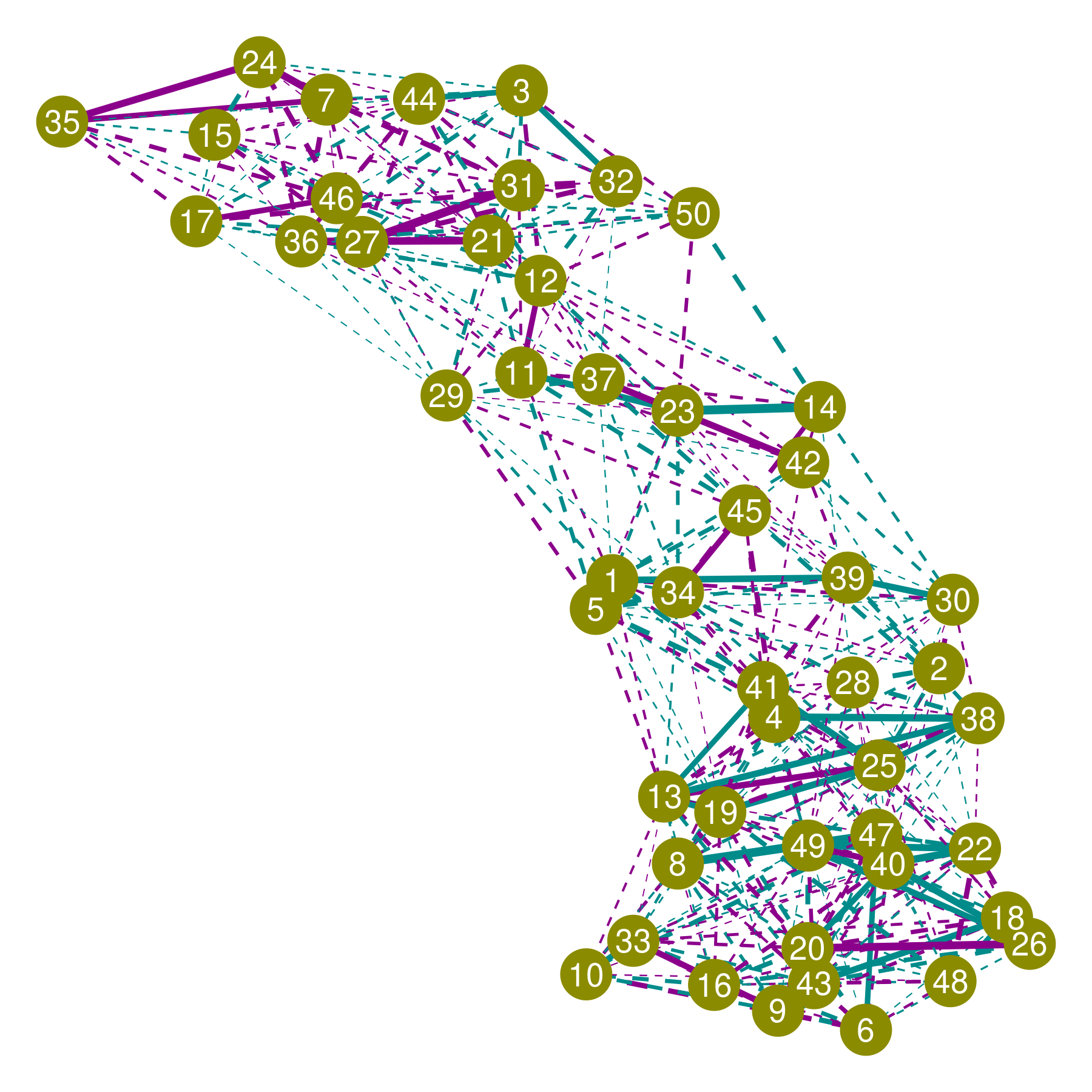}
		\caption{True graphs.}
		\label{fig: sm truegraphs}
	\end{subfigure}
	\vskip 20pt
	\begin{subfigure}[b]{.99\linewidth}
		\includegraphics[width=0.325\linewidth]{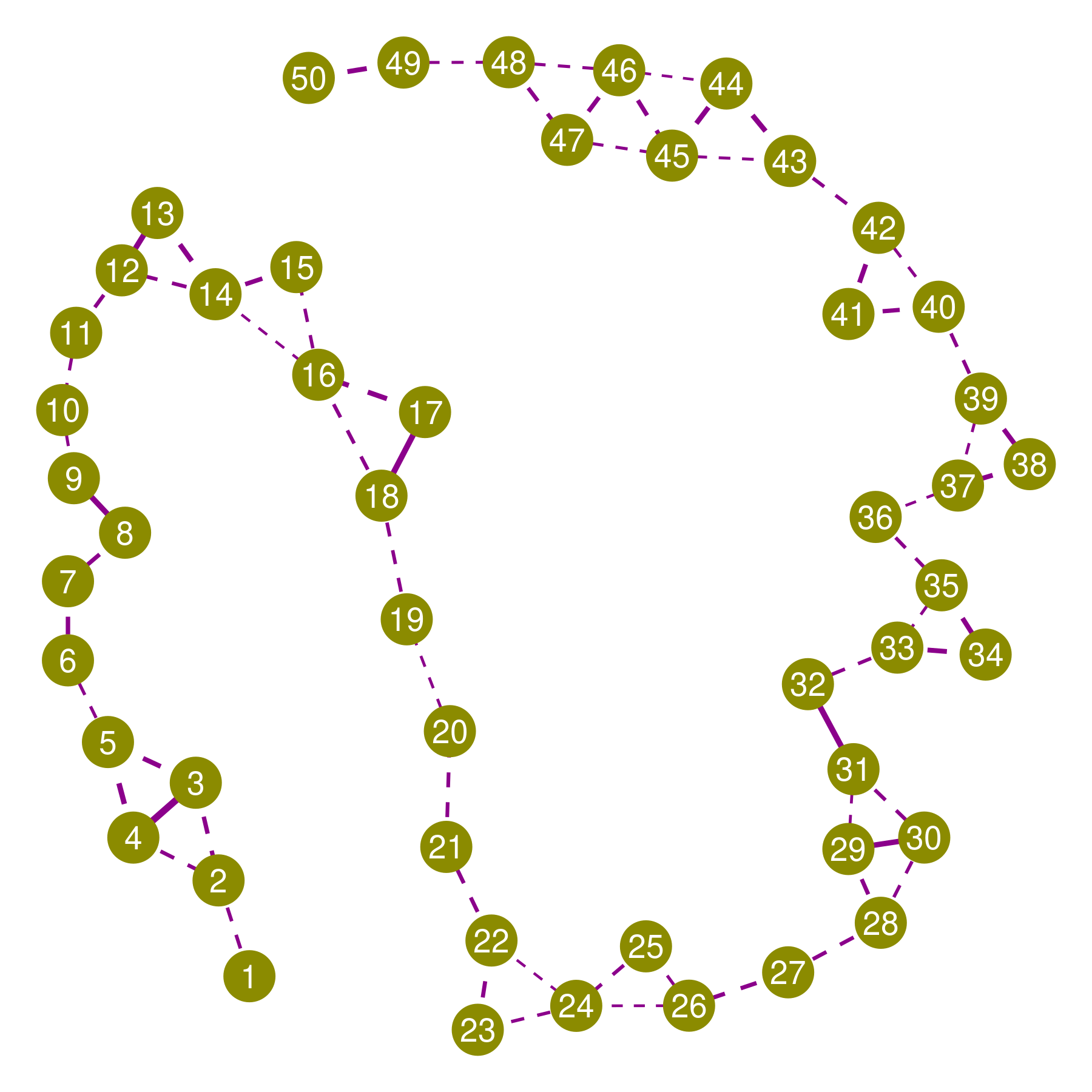}
		\includegraphics[width=0.325\linewidth]{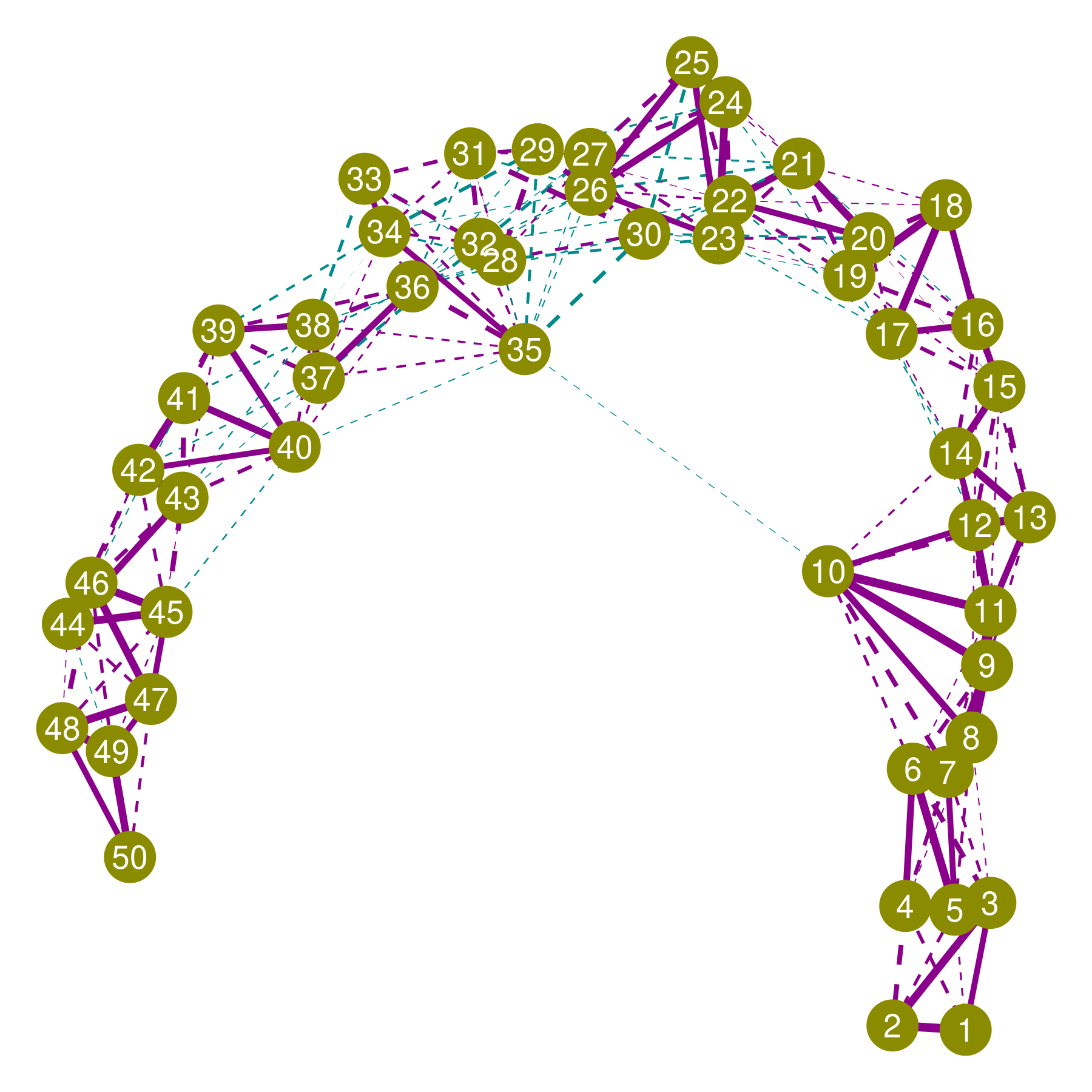}
		\includegraphics[width=0.325\linewidth]{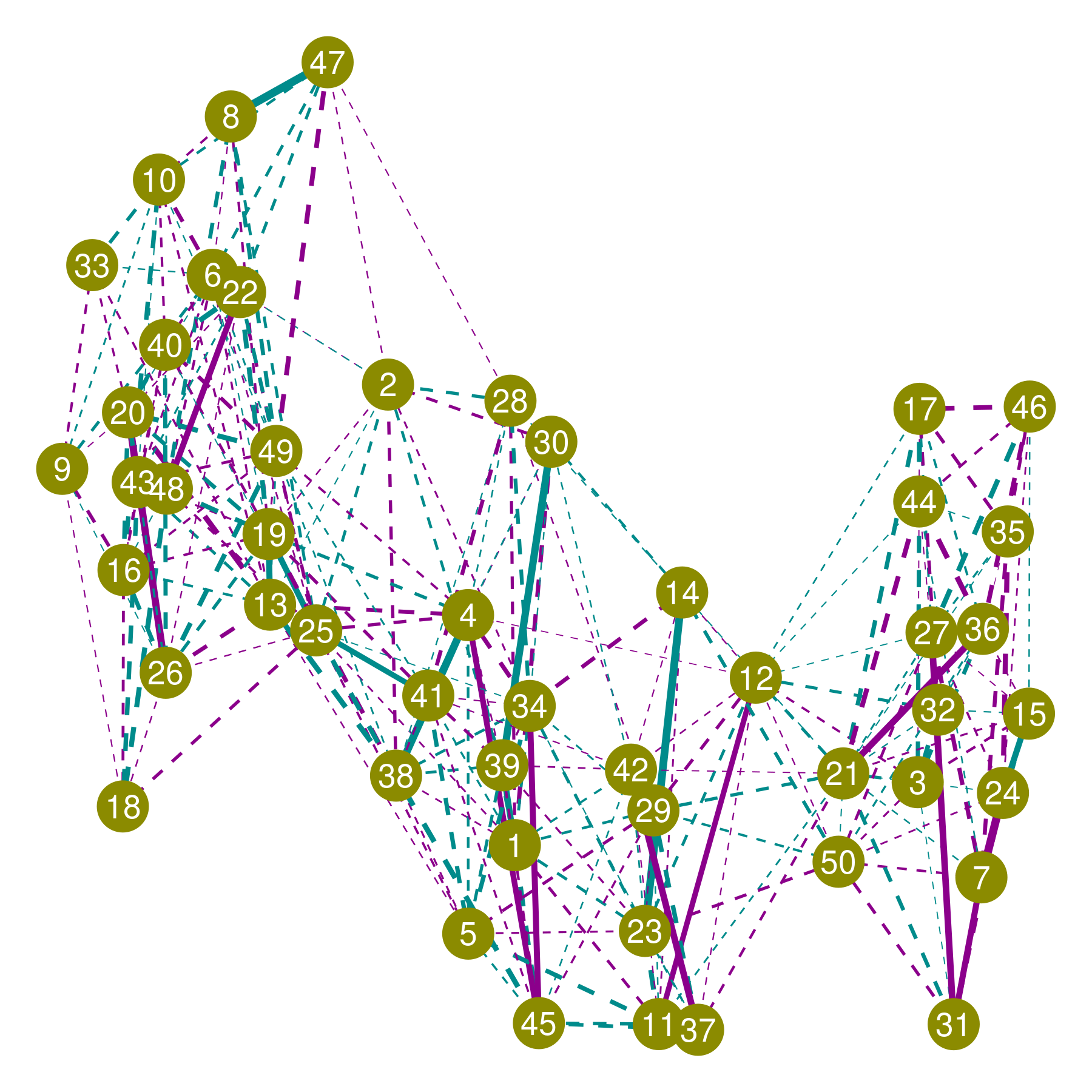}
		\caption{Estimated graphs by the PF method.}
		\label{fig: sm graphs_PF}
	\end{subfigure}
	\vskip 20pt
	\begin{subfigure}[b]{.99\linewidth}
		\includegraphics[width=0.325\linewidth]{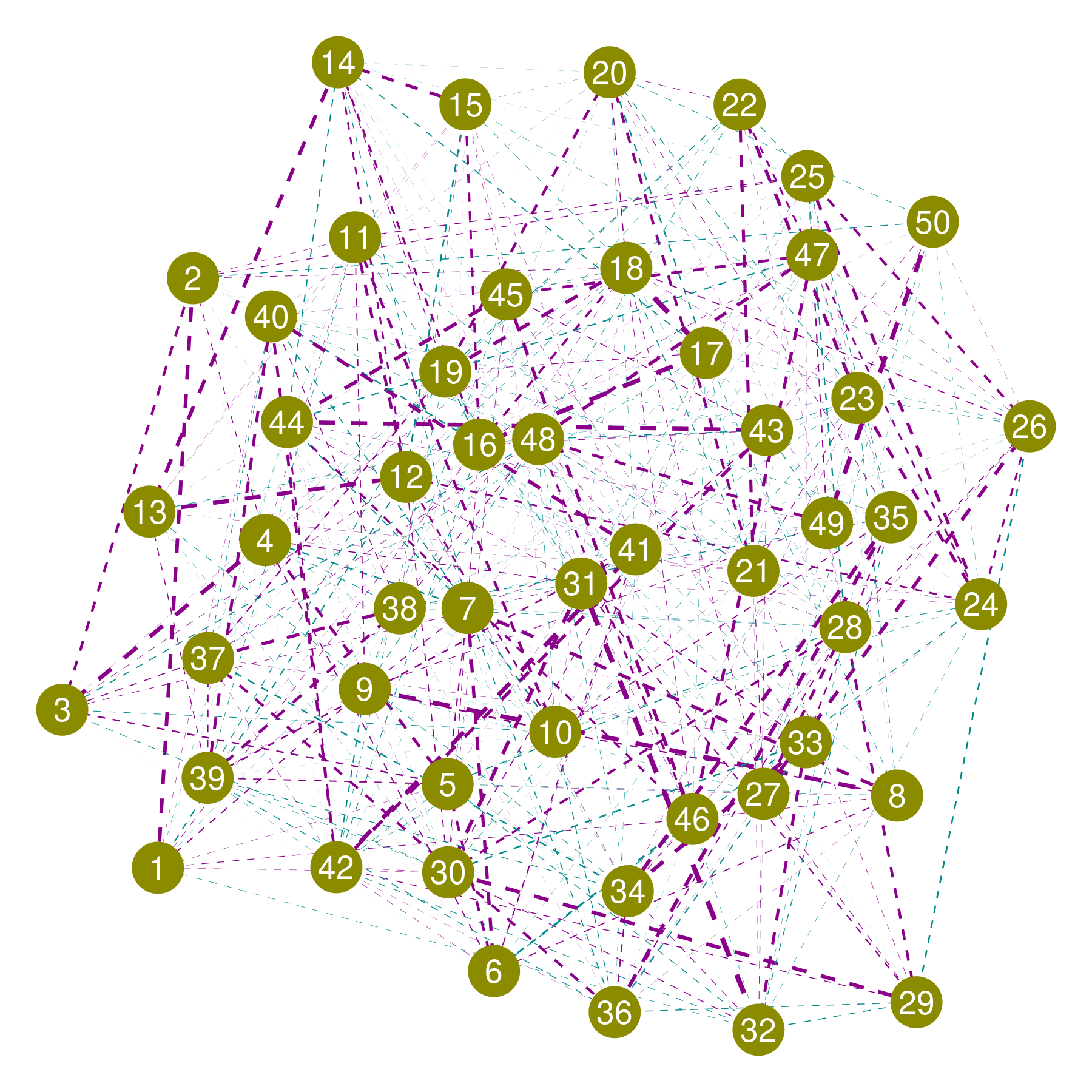}
		\includegraphics[width=0.325\linewidth]{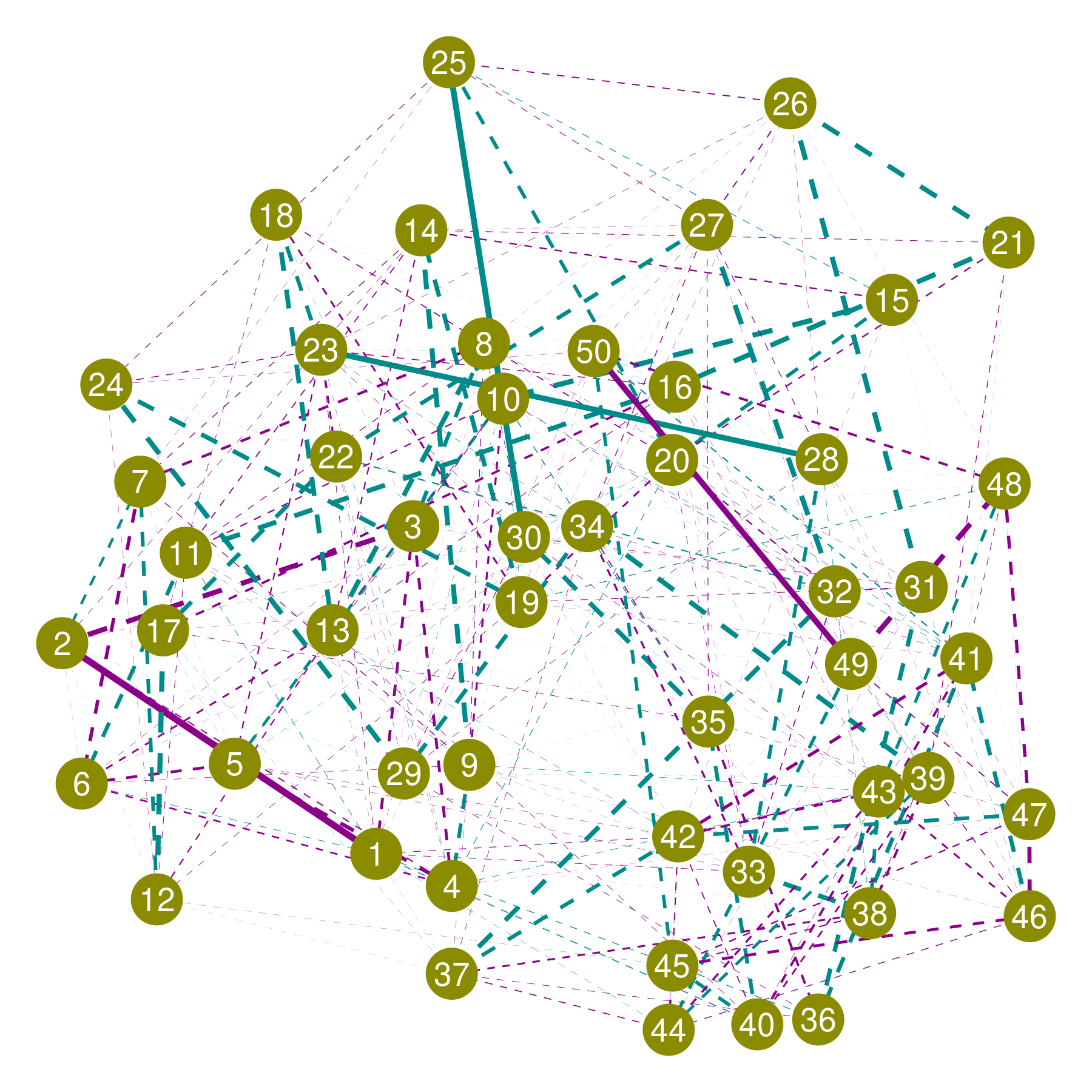}
		\includegraphics[width=0.325\linewidth]{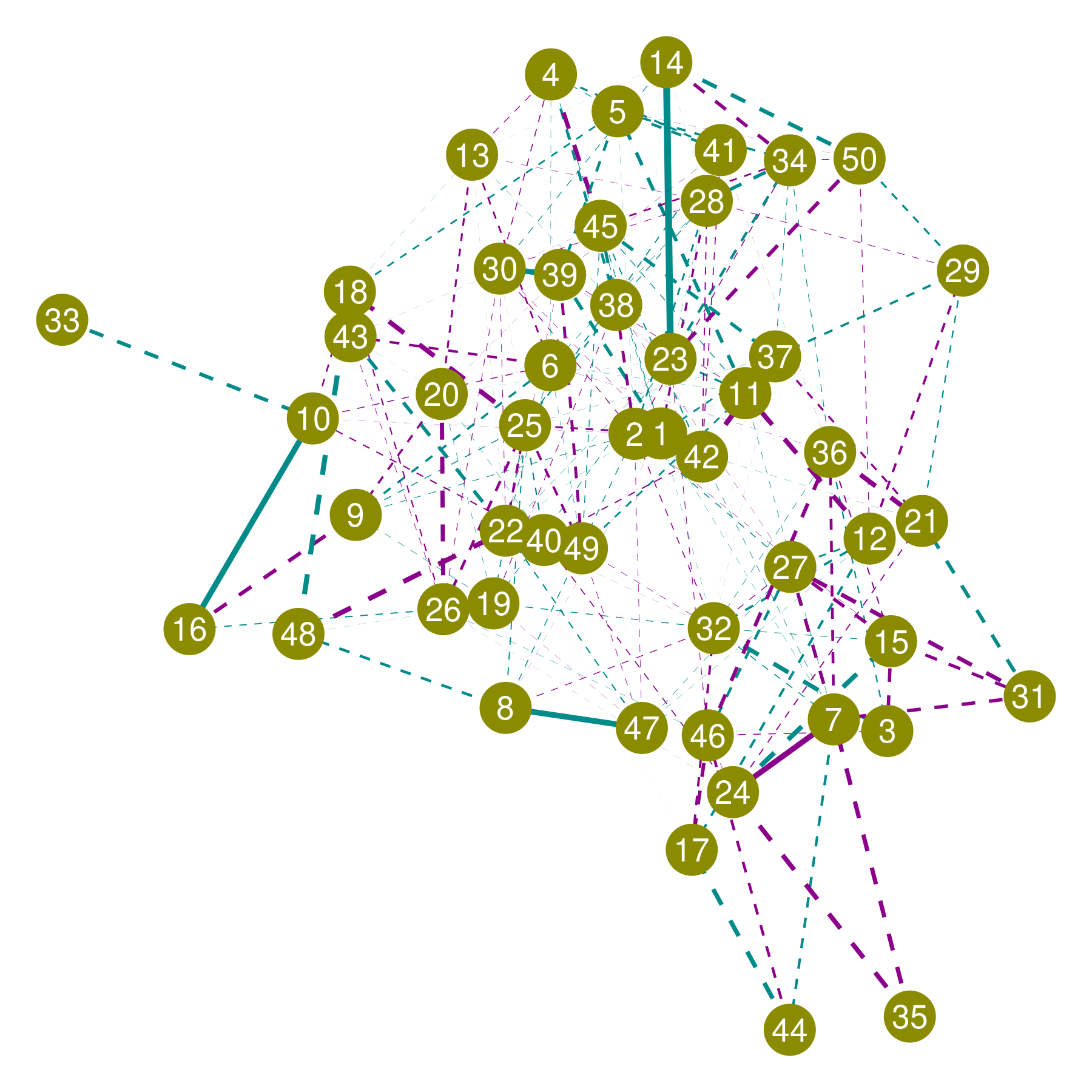}
		\caption{Estimated graphs by Bagus.}
		\label{fig: sm graphs_bagus}
	\end{subfigure}	
\end{figure}
\begin{figure}[H]\ContinuedFloat
	\begin{subfigure}[b]{.99\linewidth}
		\includegraphics[width=0.325\linewidth]{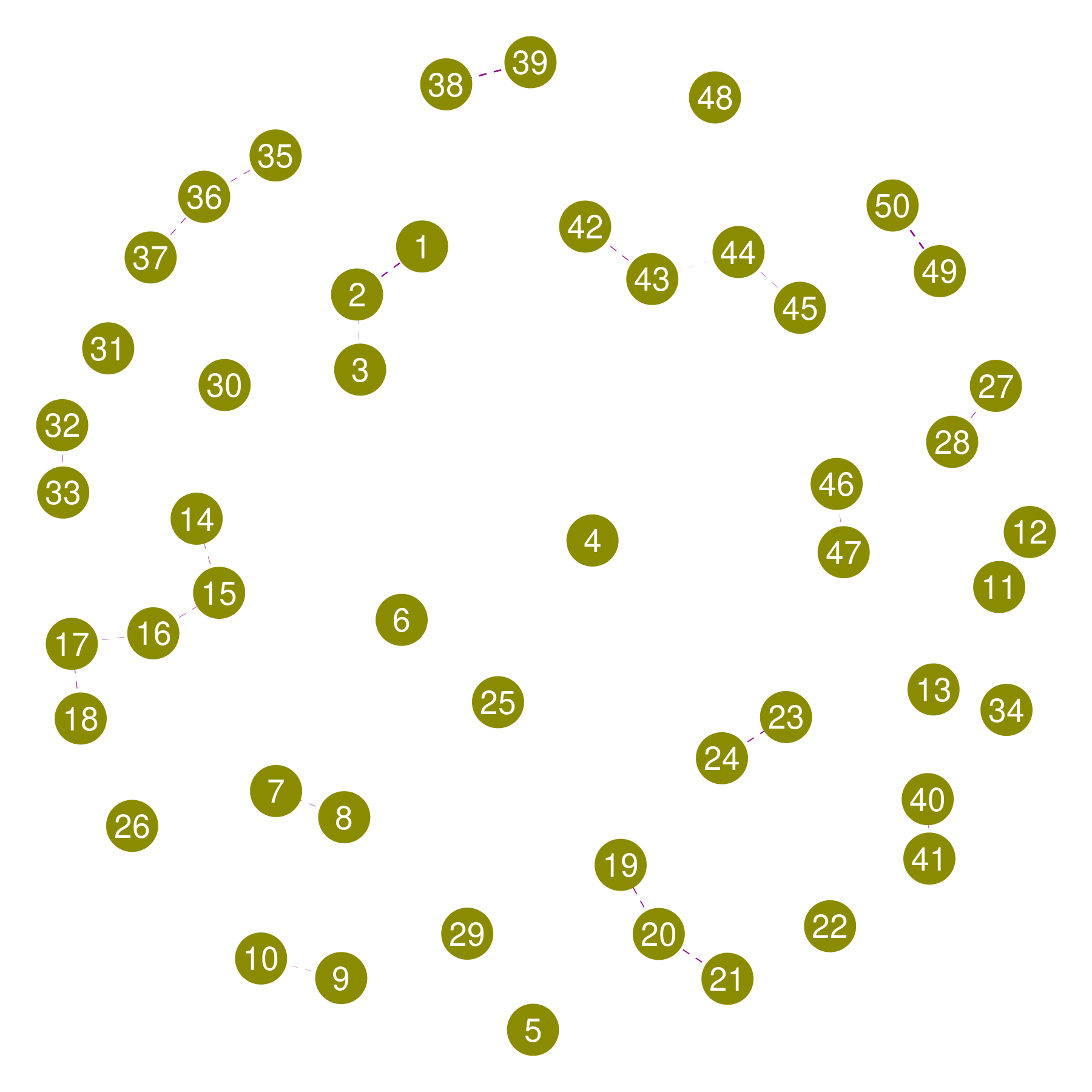}
		\includegraphics[width=0.325\linewidth]{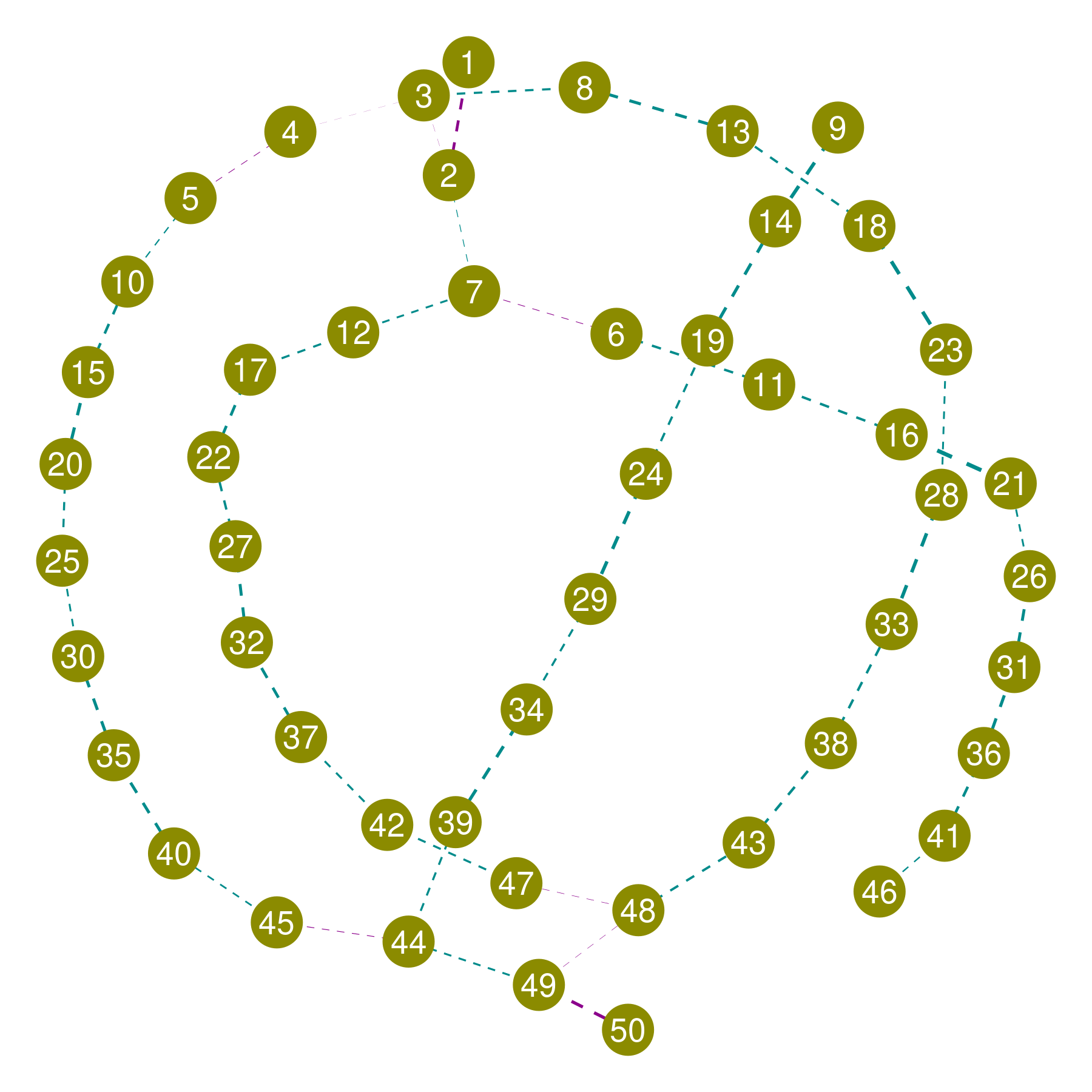}
		\includegraphics[width=0.325\linewidth]{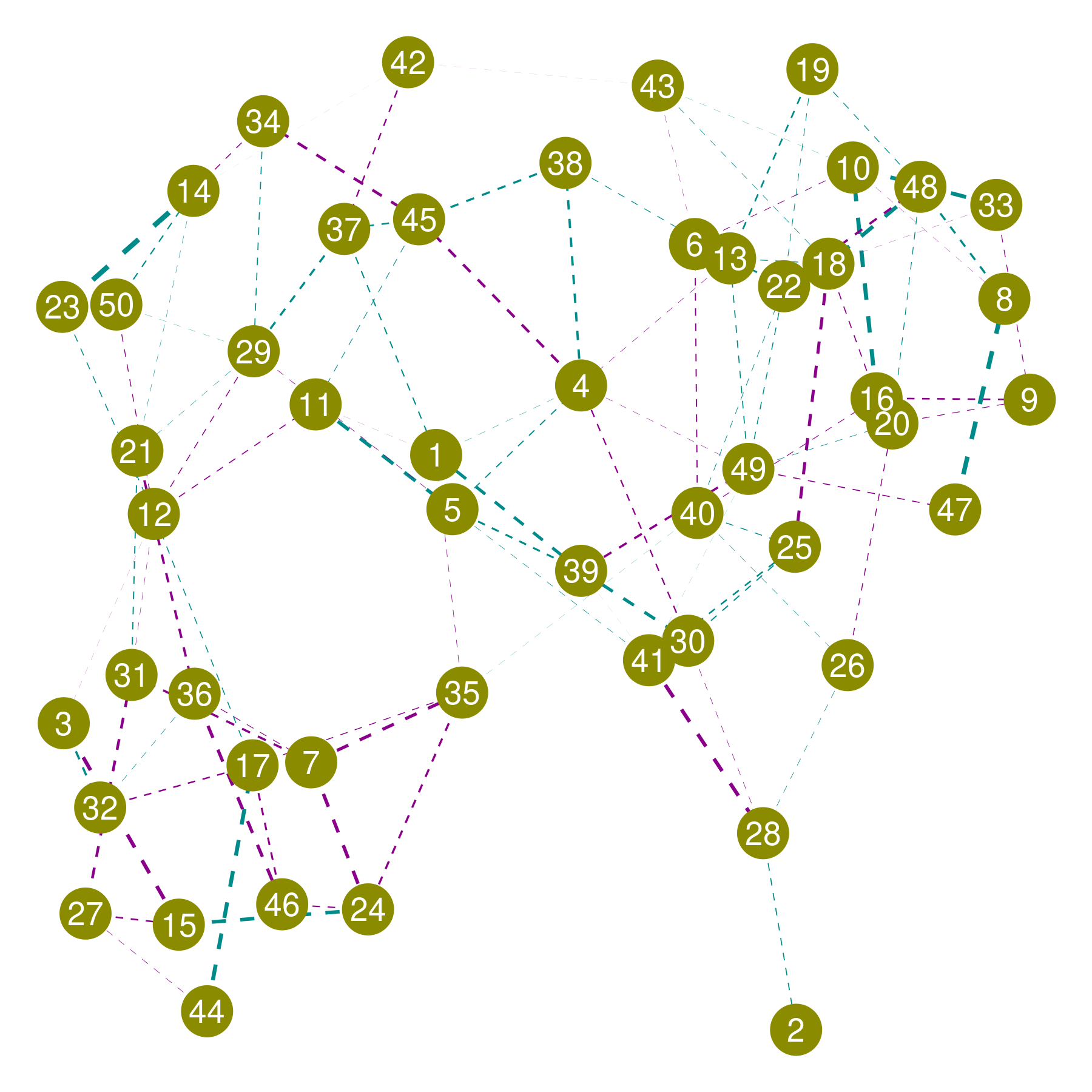}
		\caption{Estimated graphs by M\&B.}
		\label{sm fig: graphs_mnb}
	\end{subfigure}	
	\vskip 20pt
	\begin{subfigure}[b]{.99\linewidth}
		\includegraphics[width=0.325\linewidth]{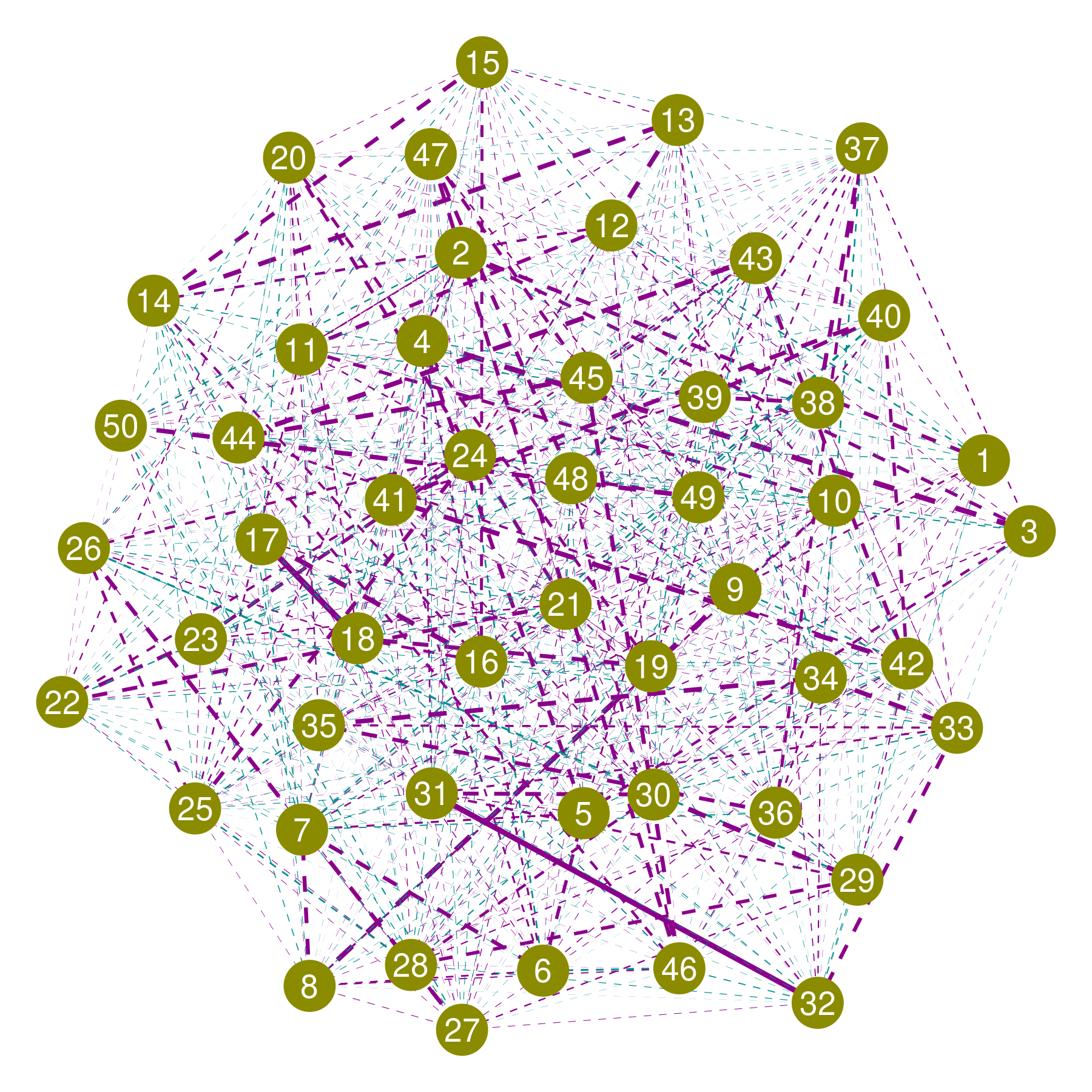}
		\includegraphics[width=0.325\linewidth]{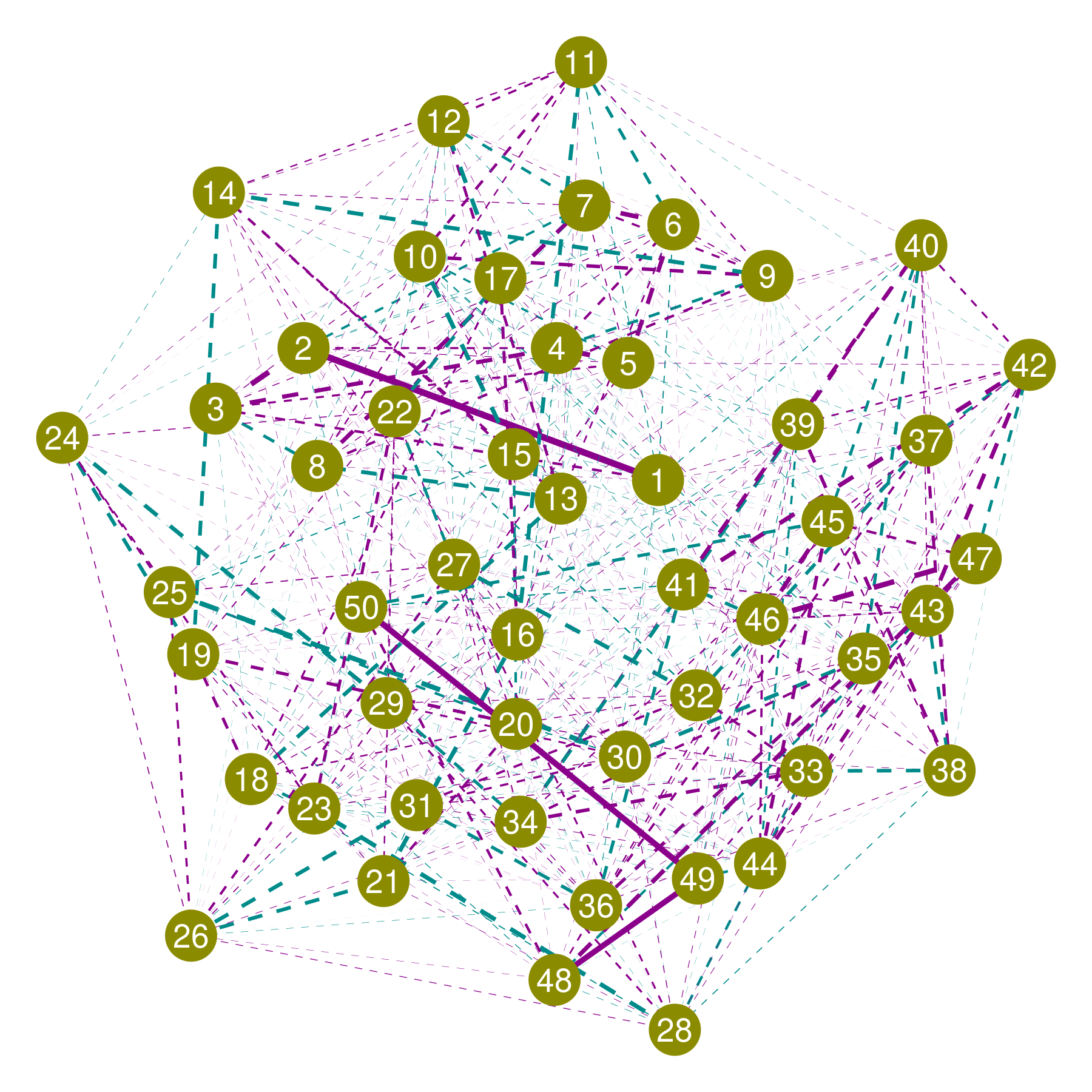}
		\includegraphics[width=0.325\linewidth]{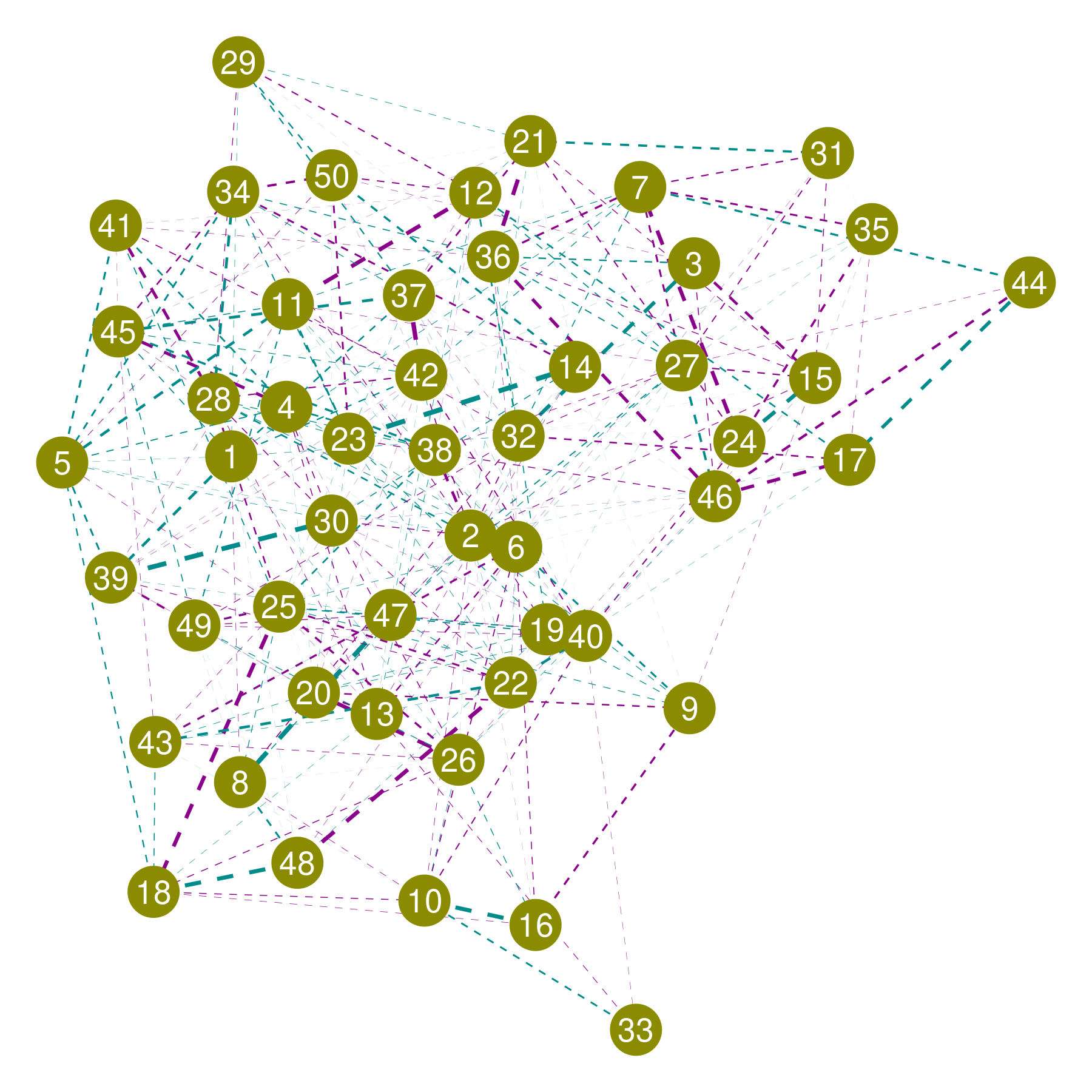}
		\caption{Estimated graphs by Glasso.}
		\label{sm fig: graphs_glasso}
	\end{subfigure}
	\caption{Results of simulation experiments: Graph recovery: Panel (a) shows the true graphs for AR(2), banded and RSM structures from left to right; 
		panels (b), (c), (d) and (e) show the corresponding estimated graphs for the our proposed PF and Bagus, M\&B and Glasso methods, respectively. 
		Positive (negative) associations are represented by blue (magenta) edges and edge-widths are proportional to the association strength. 
		If the absolute value of a partial correlation coefficient is greater (less) than $0.5$, the corresponding edge is represented by a solid (dotted) line.
		}
	\label{fig: sm graphs}
\end{figure}

\newpage
\subsection{Precision Matrix Estimation}
\begin{figure}[H]
	\centering
	\begin{subfigure}[b]{\linewidth}
		\includegraphics[width=.31\linewidth]{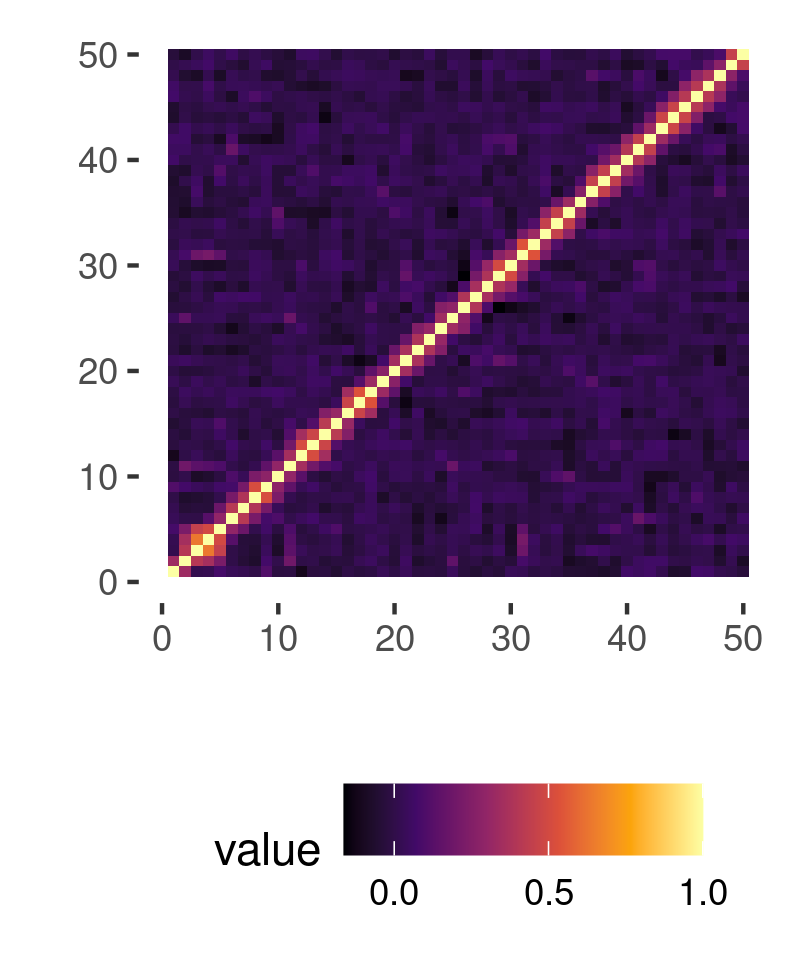}
		\includegraphics[width=.31\linewidth]{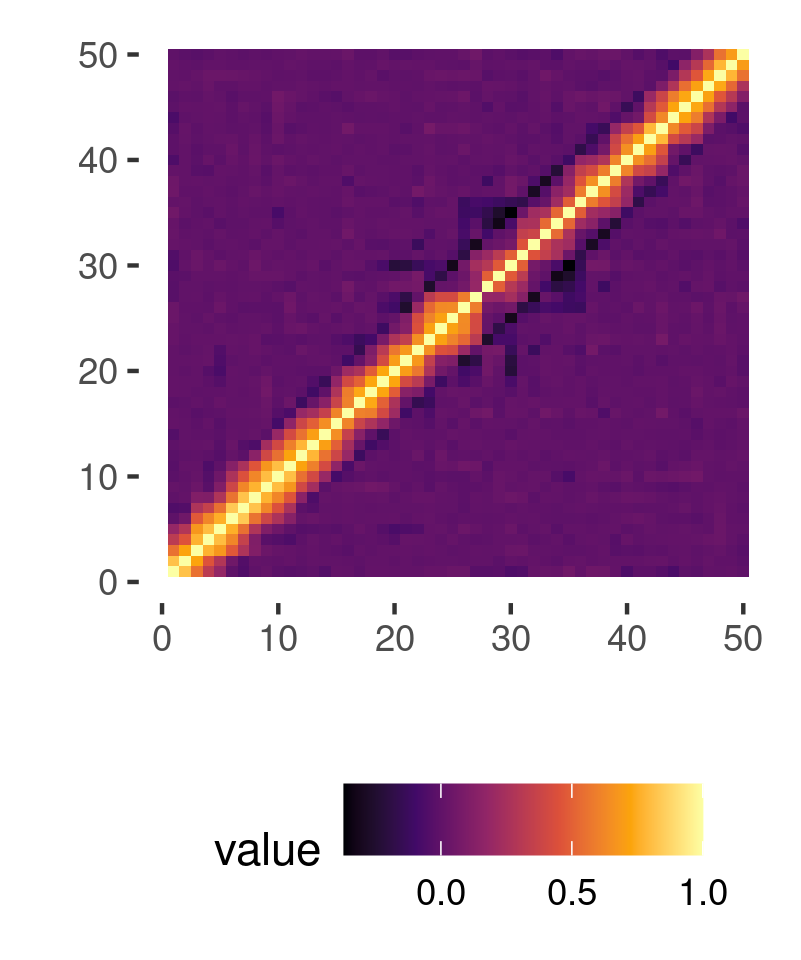}
		\includegraphics[width=.31\linewidth]{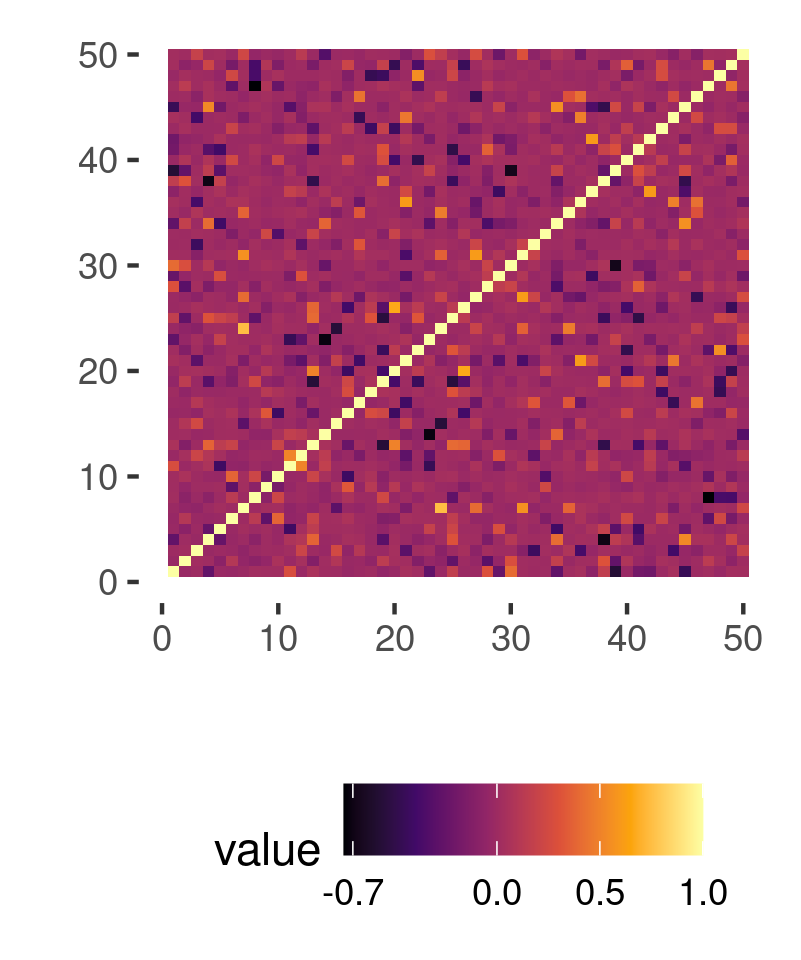}		
		\caption{Results obtained by the proposed PF method. }
		\label{fig:precision_PF}
	\end{subfigure}
	\vskip 10pt
	\begin{subfigure}[b]{.99\linewidth}
		\includegraphics[width=.31\linewidth]{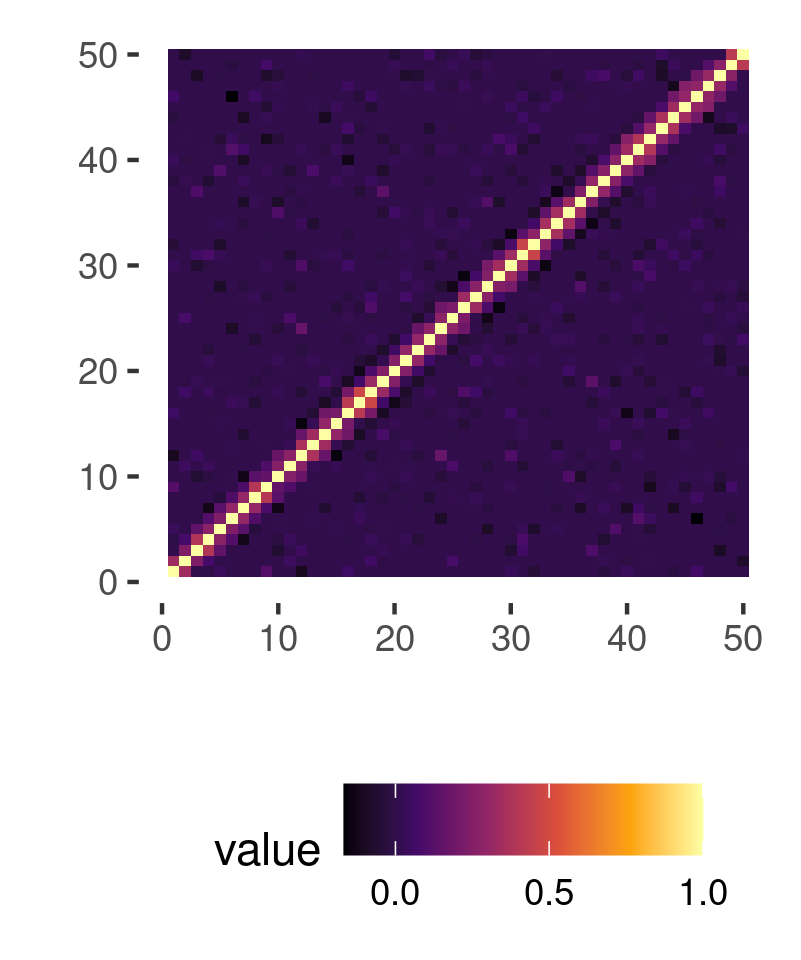}
		\includegraphics[width=.31\linewidth]{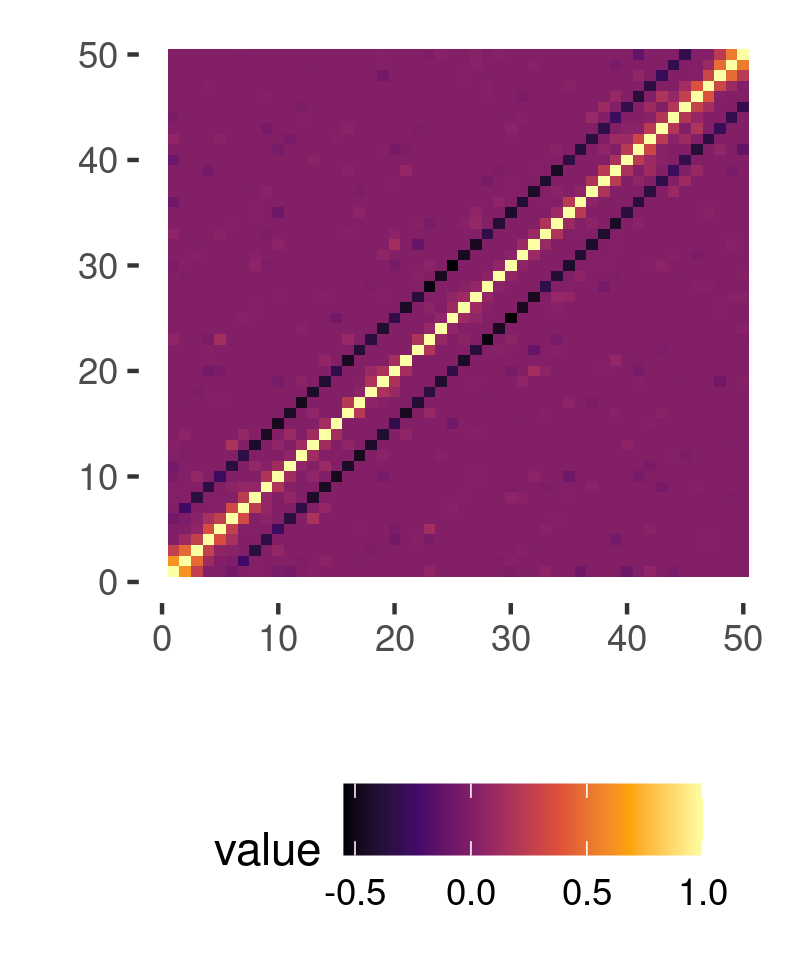}
		\includegraphics[width=.31\linewidth]{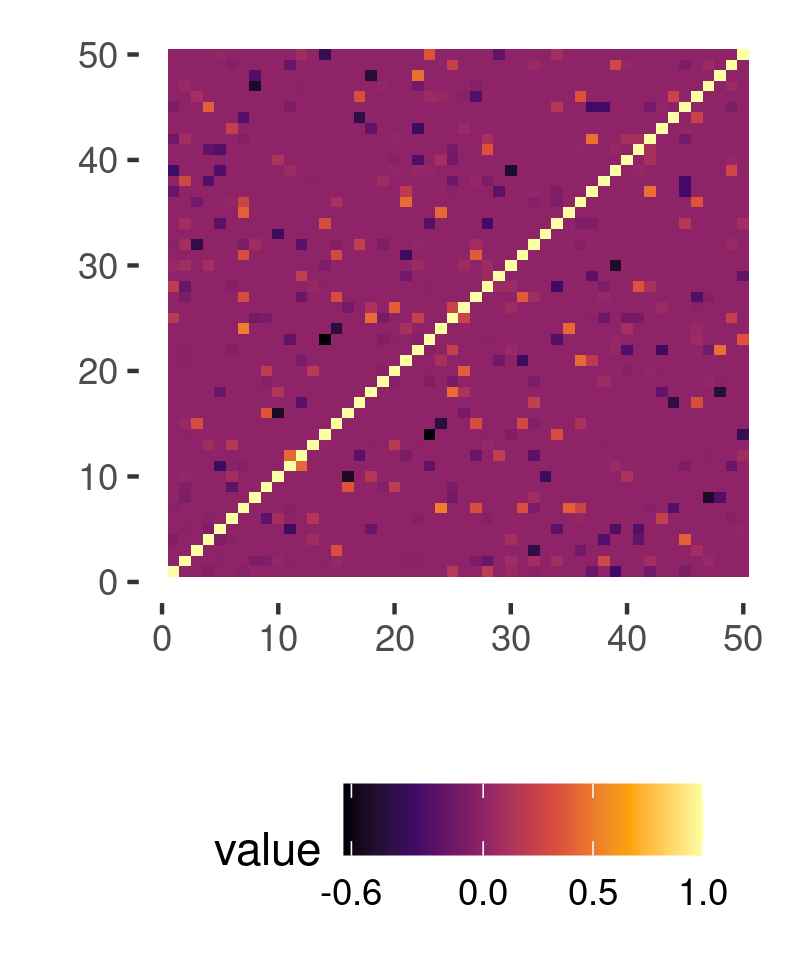}		
		\caption{Results obtained by Bagus.}
		\label{fig:precision_bagus}
	\end{subfigure}
	\vskip 10pt
	\begin{subfigure}[b]{.99\linewidth}
		\includegraphics[width=.31\linewidth]{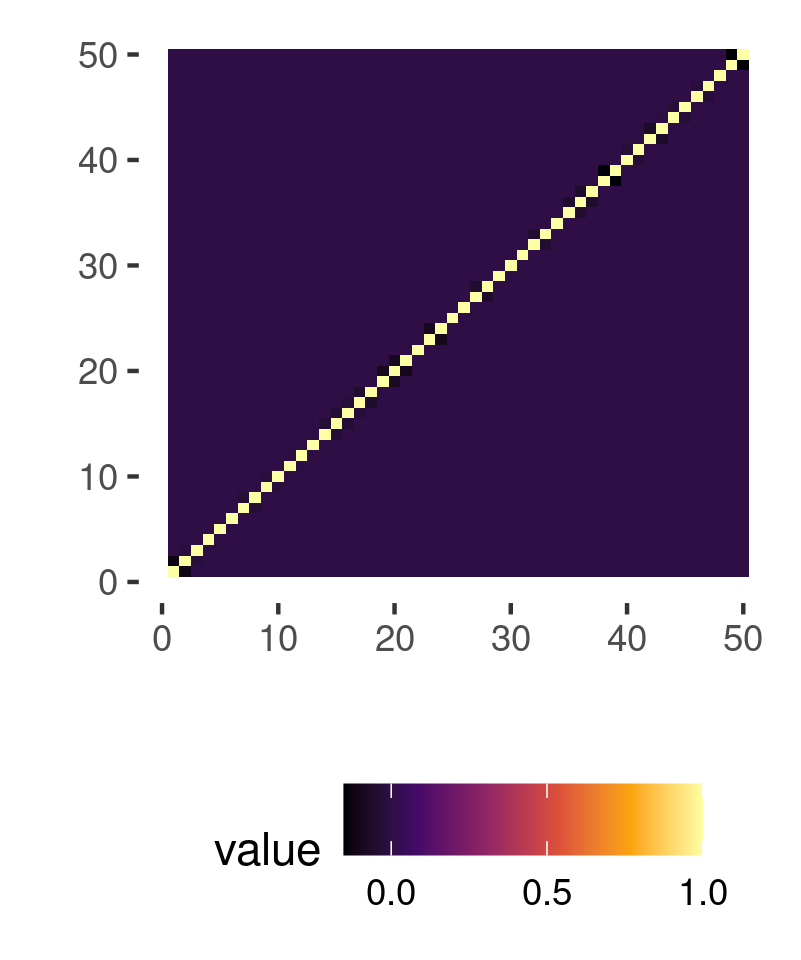}
		\includegraphics[width=.31\linewidth]{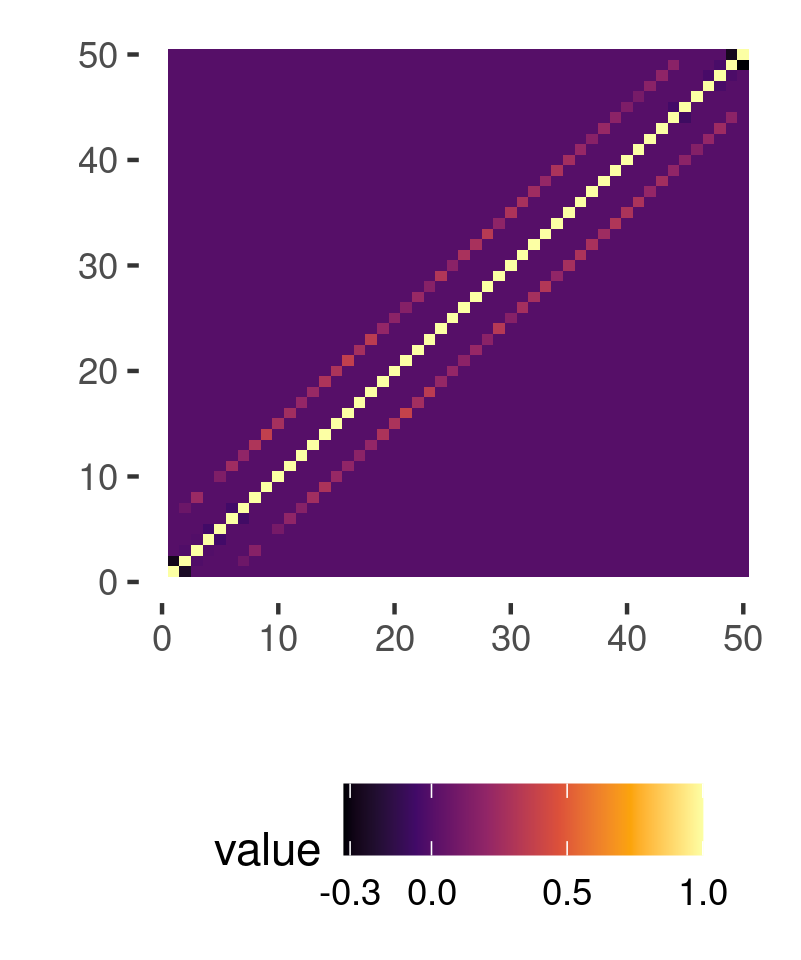}
		\includegraphics[width=.31\linewidth]{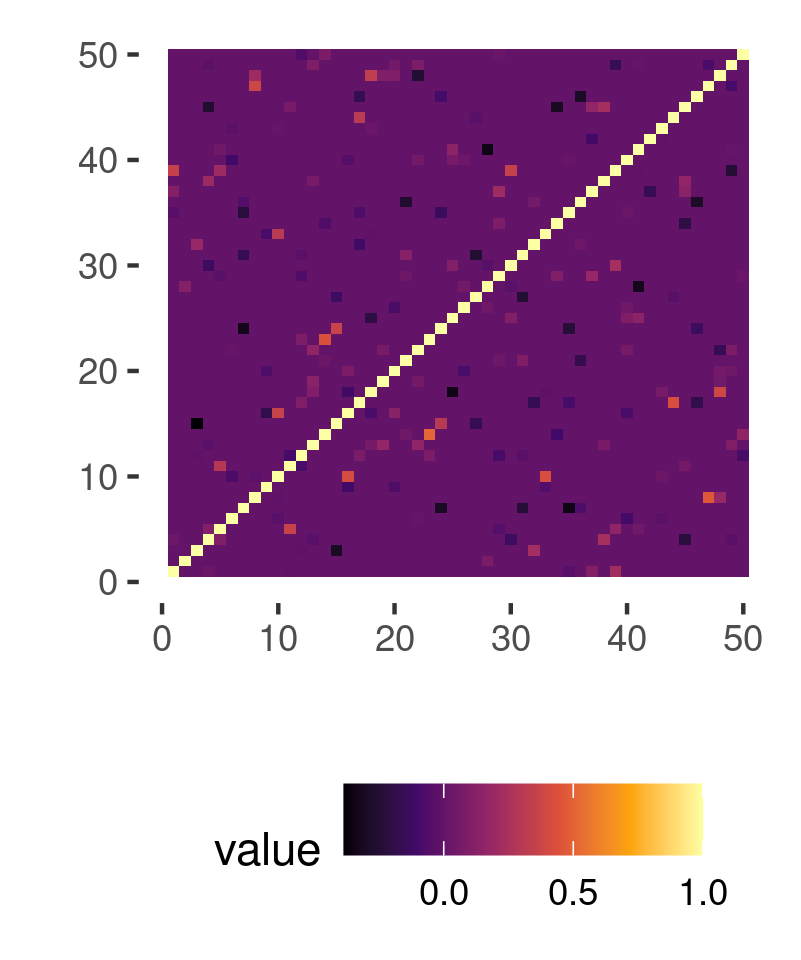}		
		\caption{Results obtained by M\&B.}
		\label{fig:precision_mnb}
	\end{subfigure}
\end{figure}
\begin{figure}[H]\ContinuedFloat
	\begin{subfigure}[b]{.99\linewidth}
		\includegraphics[width=.31\linewidth]{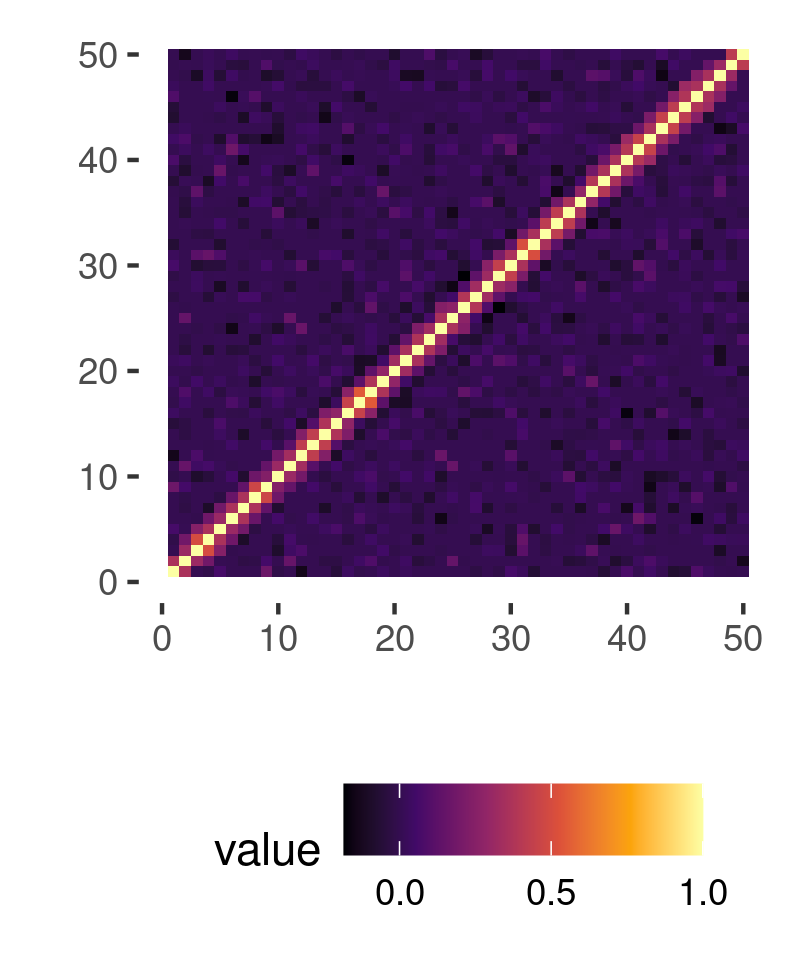}
		\includegraphics[width=.31\linewidth]{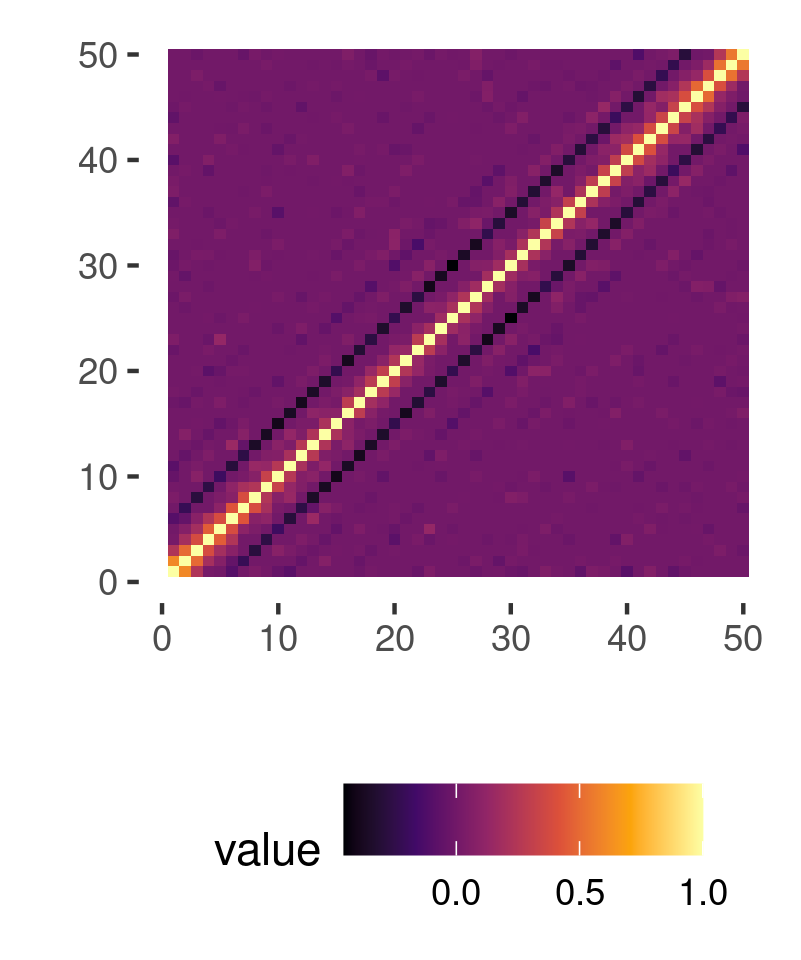}
		\includegraphics[width=.31\linewidth]{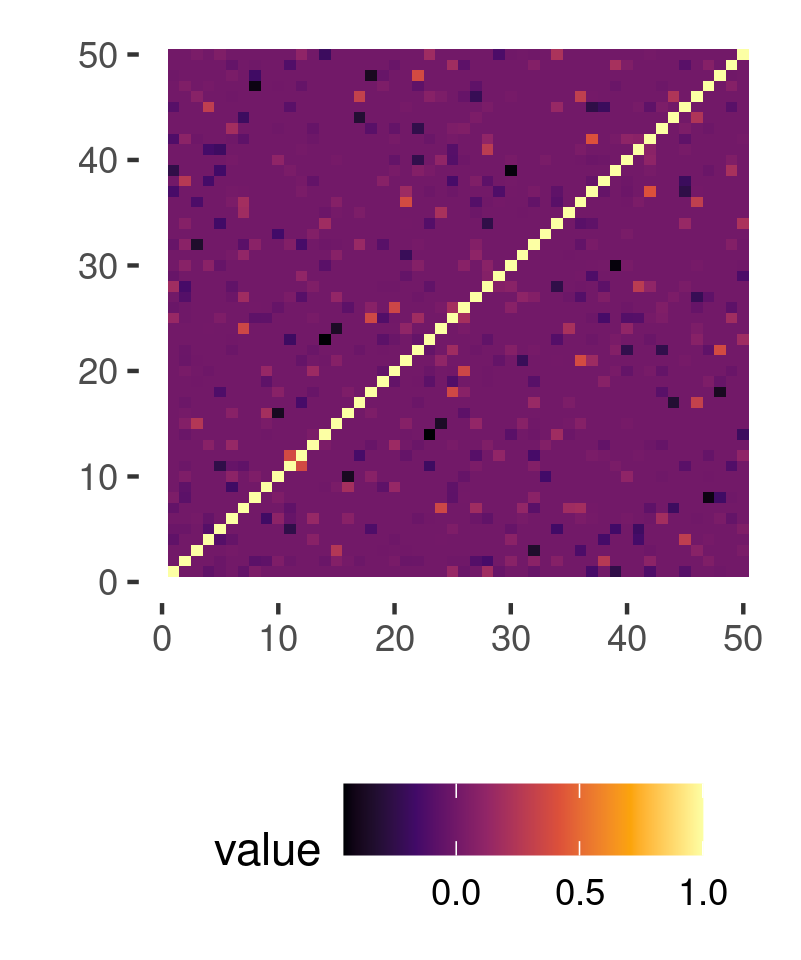}		
		\caption{Results obtained by Glasso.}
		\label{fig:glasso}
	\end{subfigure}
	\caption{Results of simulation experiments: Recovery of the {precision} matrices: 
	Panels (a), (b), (c) and (d) show the heatmaps of the estimated {scaled precision matrices $\bR=\diag(\bOmega)^{-\half} ~\bOmega~ \diag(\bOmega)^{-\half}$} for PF, Bagus, M\&B and Glasso, respectively, 
	in left to right order for the AR(2), banded and RSM model.}
	\label{fig:prec_estimates}
\end{figure}

\section{NASDAQ-100 Stock Price Data}
\label{subsec:finance}
Here we apply the precision factor model to a stock price dataset comprising the top 100 companies listed in NASDAQ for the period Jan 2015 - Dec 2019 recorded every week. 
We obtained the data from \href{https://finance.yahoo.com}{Yahoo finance}. 
After removing missing records, we ended up with data on 91 companies. 
First, we removed the trend by linear trend fitting.
The estimated graph 
shown in Figure \ref{fig: findata} indicates a sparse structure.

\begin{figure}[!ht]
	\centering
	\includegraphics[trim={1.5cm 1.25cm 1.6cm 1.55cm}, clip, width=0.9\linewidth]{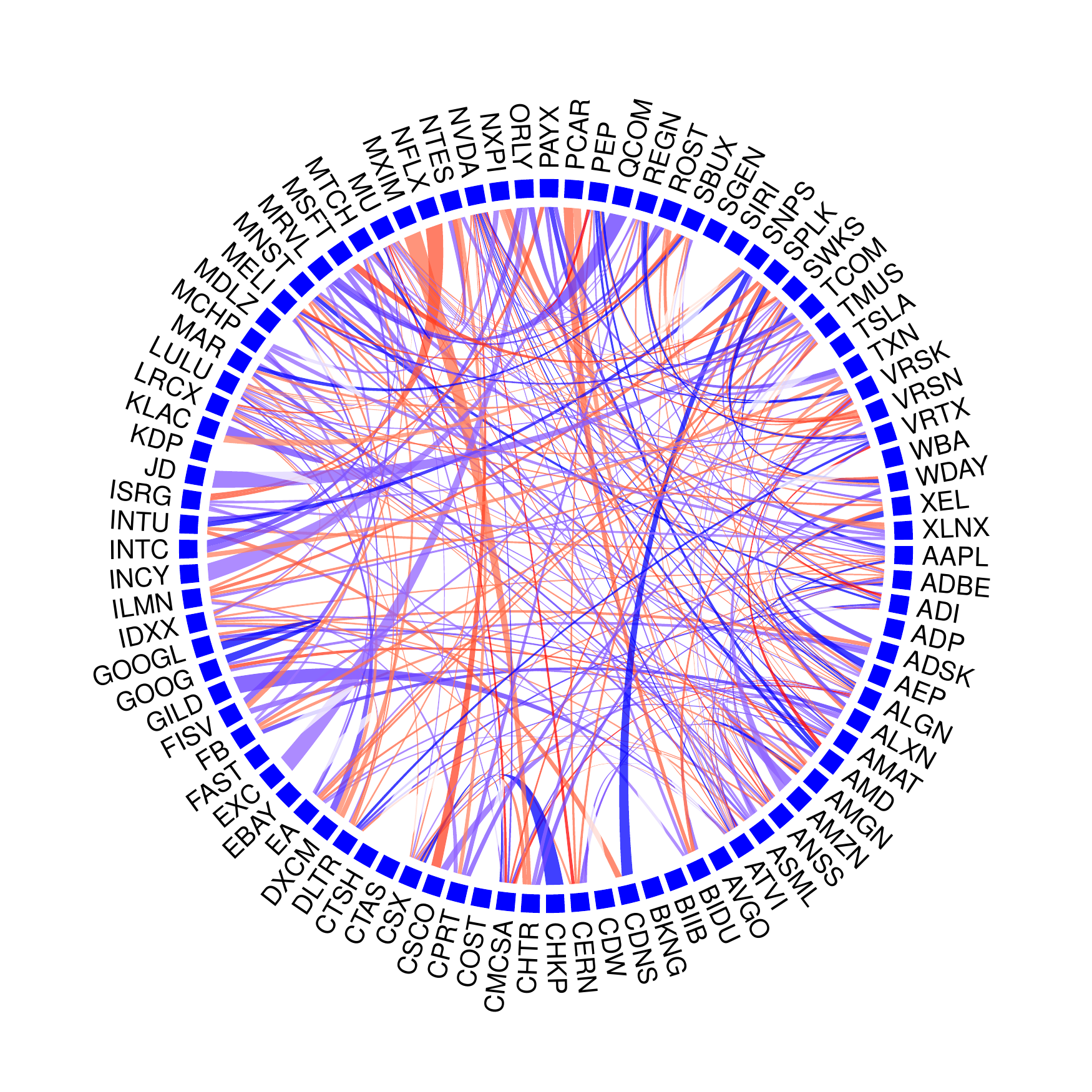}	
	\caption{Results for NASDAQ-100 stock price data: 
		Positive (negative) associations are represented by blue (red) links, 
		their opacities being proportional to the corresponding association strengths. 
		The link widths are inversely proportional to the number of edges associated with the corresponding nodes. 
	}
	\label{fig: findata}
\end{figure}

We highlight some interesting features observed from the analysis.
We observe a strong positive association between the rival GPU developers AMD and Nvidia Corporation (NVDA) consistent with the increasing demand of GPU computing. 
Similarly, a positive association between Qualcomm (QCOM) and Marvell Technology Group (MRVL), 
both of which develop and produce semiconductors and related technology, can be seen.  
The electronic design automation company Synopsys (SNPS) and the software company
Cadence Design Systems, Inc. (CDNS) exhibit strong positive associations.
Electronic Arts (EA) and Activision (ATVI) are video game developers and a Positive association is observed between them. 
ATVI seem to have strong associations with Amazon.com, Inc. (AMZN) and Apple Inc (AAPL).
The semiconductor manufacturer companies 
Texas Instruments (TXN) and Xilinx (XLNX) seem to have a negative association. 
We observe positive association between the computer memory and data storage producing company Micron Technology, Inc. (MU) and Match Group (MTCH).
Interestingly, the American semiconductor company Skyworks Solutions, Inc. (SWKS) exhibit strong negative association with companies like Microsoft (MSFT), Comcast Corporation (CMCSA), MU, etc.
Many of the aforementioned stocks exhibit strong positive association with the Meta Platforms, Inc. (FB) previously known as Facebook Inc.
%
As expected, we observe a strong positive association between GOOGL and GOOG, which are Google shares with and without voting rights, respectively. 

Notably, most of the connected nodes in the NASDAQ-100 listing correspond to technology companies and electronics manufacturers. 
This is in accordance with the ushering in of the `digital era' over the last decade, with technology giants taking over major shares of the market. 
Also, the strong negative associations between several companies in similar domains reflect the competitive nature of the market.

\baselineskip=14pt
\bibliographystylelatex{natbib}
\bibliographylatex{Graphical_Models}
\end{document}